\documentclass[conference,compsoc]{IEEEtran}
\ifCLASSOPTIONcompsoc
  \usepackage[nocompress]{cite}
\else
  \usepackage{cite}
\fi

\ifCLASSINFOpdf
\else
\fi

\usepackage{algorithm}
\usepackage{algpseudocode}
\usepackage[hidelinks]{hyperref}
\usepackage{amsmath}
\usepackage{amssymb}
\usepackage{amsthm}
\usepackage{booktabs}
\usepackage{cleveref}
\usepackage{enumitem}
\usepackage{multirow}
\usepackage{mathtools}
\usepackage{subcaption}
\usepackage[most]{tcolorbox}
\usepackage{tikz}
\usepackage{url}
\usepackage{xcolor}
\usepackage{xspace}
\usepackage{glossaries}
\usepackage[normalem]{ulem}

\usepackage{siunitx}  %
\usepackage{mdframed}

\newmdenv[
  backgroundcolor=white,
  linecolor=black,
  linewidth=1pt,
  roundcorner=5pt,
  innerleftmargin=10pt,
  innerrightmargin=10pt,
  innertopmargin=10pt,
  innerbottommargin=10pt
]{protocolabox}

\newtcolorbox{myalgorithm}[2][]{
enhanced,
colback=white,
sharp corners,
colframe=black,
coltitle=black,
colbacktitle=white,        %
before title={\refstepcounter{algorithm}\label{#2}}, %
title={\textbf{Protocol~\thealgorithm} #1},
top=1mm,
bottom=1mm,
left=1mm,
right=1mm,
boxrule=0.5pt, %
}

\newcommand{\citeextended}{\ifdefined\isnotextended~\cite{Rose2024-UTraceFull}\fi\xspace}

\usepackage{environ}
\usepackage{etoolbox}  %
\usepackage{caption}   %
\usepackage{graphicx}

\makeatletter
\NewEnviron{myhideenv}{%
  \ifdefined\isnotextended
      \begingroup
        \renewenvironment{figure}[1][]{%
          \begingroup
            \def\@captype{figure}%
            \begin{minipage}{0pt}%
        }{%
            \end{minipage}%
          \endgroup}
        
        \renewenvironment{table}[1][]{%
          \begingroup
            \def\@captype{table}%
            \begin{minipage}{0pt}%
        }{%
            \end{minipage}%
          \endgroup}
        
        \renewenvironment{subfigure}[2][]{%
          \begingroup
            \def\@captype{figure}%
            \begin{minipage}{0pt}%
        }{%
            \end{minipage}%
          \endgroup}

        \renewenvironment{algorithm}[2][]{%
          \begingroup
            \def\@captype{algorithm}%
            \begin{minipage}{0pt}%
        }{%
            \end{minipage}%
          \endgroup}
        \let\origcaption\caption      %
        \renewcommand*\caption[2][]{%
          \phantomcaption             %
          \origcaption*[{##1}]{}       %
        }

        \renewcommand*{\includegraphics}[2][]{\relax}       %

        \setbox0\vbox{%
   \let\xxwrite\write
   \protected\def\write{\immediate\xxwrite}%
   {\tiny XX\BODY XX}}
      \endgroup
  \else
      \BODY
  \fi}
\makeatother

\newtheorem{theorem}{Theorem}[section]

\newtheorem{lemma}[theorem]{Lemma}

\newtheorem{definition}[theorem]{Definition}

\newcommand{\abs}[1]{| #1 |}

\newcommand{\bbR}{\ensuremath{\mathbb{R}}}
\newcommand{\bbE}{\ensuremath{\mathbb{E}}}

\newcommand{\Adv}{\ensuremath{\mathcal{A}}\xspace}
\newcommand{\Sim}{\ensuremath{\mathcal{S}}}

\newcommand{\ideal}{\ensuremath{\mathcal{F}}\xspace}

\newcommand{\idealTrunc}{\ensuremath{\ideal_\mathsf{trunc}}\xspace}
\newcommand{\idealSort}{\ensuremath{\ideal_\mathsf{sort}}\xspace}

\newcommand{\idealGradient}{\ensuremath{\ideal_\mathsf{gradient}}\xspace}
\newcommand{\idealDotprod}{\ensuremath{\ideal_\mathsf{dotp}}\xspace}
\newcommand{\idealDotprodNoTrunc}{\ensuremath{\ideal_\mathsf{int-dotp}}\xspace}
\newcommand{\idealRecip}{\ensuremath{\ideal_\mathsf{rsqrt}}\xspace}

\newcommand{\idealMul}{\ensuremath{\ideal_\mathsf{mul}}\xspace}
\newcommand{\idealMulNoTrunc}{\ensuremath{\ideal_\mathsf{int-mul}}\xspace}

\newcommand{\idealSGD}{\ensuremath{\ideal_\mathsf{SGD}}\xspace}
\newcommand{\idealLoss}{\ensuremath{\ideal_\mathsf{loss}}\xspace}
\newcommand{\idealCamelPreprocessing}{\ensuremath{\ideal_\mathsf{Unl-Pre}}\xspace}
\newcommand{\idealCamelOnline}{\ensuremath{\ideal_\mathsf{Unl-Online}}\xspace}
\newcommand{\idealGradPreprocessing}{\ensuremath{\ideal_\mathsf{Grad-Pre}}\xspace}
\newcommand{\idealGradOnline}{\ensuremath{\ideal_\mathsf{Grad-Online}}\xspace}

\newcommand{\protocolCamelPreprocessing}{\ensuremath{\sprotocol_\mathsf{Unl-Pre}}\xspace}
\newcommand{\protocolCamelOnline}{\ensuremath{\sprotocol_\mathsf{Unl-Online}}\xspace}

\newcommand{\smodel}{\ensuremath{f_\theta}\xspace}
\newcommand{\smodelI}{\ensuremath{f_{\theta_i}}\xspace}
\newcommand{\smodelW}{\ensuremath{\theta}\xspace}
\newcommand{\smodelWUnl}{\ensuremath{\theta_{-i}}\xspace}
\newcommand{\ssAbb}[1]{\ensuremath{[\,#1\,]}}

\newcommand{\sadversary}{\Adv\xspace}
\newcommand{\ssim}{\Sim\xspace}
\newcommand{\ssimSGD}{\ensuremath{\mathcal{S}_{\text{SGD}}}\xspace}
\newcommand{\ssimGradient}{\ensuremath{\mathcal{S}_{\text{Gradient}}}\xspace}

\newcommand{\PPT}{\textsf{PPT}\xspace}

\newcommand{\sPartyAll}{\ensuremath{\texttt{P}}\xspace}
\newcommand{\sPartyCorrupted}{\ensuremath{M_\sPartyAll}\xspace}
\newcommand{\Cset}{\mathcal{C}}

\newcommand{\real}{\textnormal{\texttt{Real}}}
\newcommand{\advAux}{\ensuremath{\textsf{aux}_{\sadversary}}}

\newcommand{\sabb}{\ensuremath{\mathcal{F}_{\text{ABB}}}\xspace}
\newcommand{\sabbid}{\ensuremath{\mathcal{F}_{\text{ABB [ID]}}}\xspace}

\newcommand{\sprotocol}{\ensuremath{\Pi}\xspace}

\newcommand{\sprotocolArc}{\ensuremath{\sprotocol_{\text{Arc}}}\xspace}
\newcommand{\sprotocolArcWithPreprocessing}{\ensuremath{\sprotocol_{\text{Arc-P}}}\xspace}

\newcommand{\sabbAAuditID}{\ensuremath{\sabbid\textsf{.Audit}}\xspace}

\newcommand{\sabbShare}{\ensuremath{\sabb\textsf{.SSShare}}\xspace}
\newcommand{\sabbOpen}{\ensuremath{\sabb\textsf{.SSOpen}}\xspace}

\newcommand{\sPartyAuditor}{\ensuremath{\texttt{C}}\xspace}
\newcommand{\sPartyAuditComputer}{\ensuremath{\texttt{AC}}\xspace}
\newcommand{\sPartyModelHolder}{\ensuremath{\texttt{M}}\xspace}
\newcommand{\strainset}{D\xspace}

\newcommand{\sauditfunction}{\ensuremath{f_\texttt{audit}}}

\newcommand{\sNumParties}{\ensuremath{{m}}\xspace}

\newcommand{\sabbAAuditPreprocessID}{\ensuremath{\sabbid\textsf{.Audit-Pre}}\xspace}
\newcommand{\sabbAAuditOnlineID}{\ensuremath{\sabbid\textsf{.Audit-Online}}\xspace}

\newcommand{\SSScheme}{\textsf{SS}}
\newcommand{\SSShare}{\textsf{\SSScheme.Share}}
\newcommand{\SSReconstruct}{\textsf{\SSScheme.Reconstruct}}

\newcommand{\sGset}{\ensuremath{g}\xspace}
\newcommand{\srecord}{\ensuremath{R}\xspace}

\newlist{hybrid}{enumerate}{1}
\setlist[hybrid]{align=left, itemsep=2pt, topsep=8pt, leftmargin=12pt, label={$\text{Hyb}_{\arabic*}$}, ref={\arabic*}}
\newlist{algos}{itemize}{2}
\setlist[algos]{align=left,itemsep=2pt,left=0pt,label=•}
\newlist{protocol_steps}{enumerate}{2}
\setlist[protocol_steps]{align=left,itemsep=0pt,topsep=0pt,left=0pt,label={\arabic*.},ref={\arabic*}}

\newacronym{cl}{CL}{Collaborative Learning}
\newacronym{ml}{ML}{Machine Learning}
\newacronym{lan}{LAN}{Local Area Network}
\newacronym{wan}{WAN}{Wide Area Network}
\newacronym{rtt}{RTT}{Round-Trip Time}
\newacronym{mpc}{MPC}{Multi-Party Computation}
\newacronym{ppml}{PPML}{Privacy-Preserving Machine Learning}
\newacronym{3pc}{3PC}{Three-Party Computation}
\newacronym{lsss}{SSS}{secret-sharing scheme}
\newacronym{knn}{kNN}{k-Nearest Neighbours}

\newcommand{\role}[1]{#1}
\newglossaryentry{r:datasource}{name={\role{data source}},description={},plural={\role{data sources}}}
\newglossaryentry{r:inputparty}{name={\role{data owner}},description={},plural={\role{data owners}}}
\newglossaryentry{r:tcomputer}{name={\role{training computer}},description={},plural={\role{training computers}}}
\newglossaryentry{r:modelowner}{name={\role{model owner}},description={},plural={\role{model owners}}}
\newglossaryentry{r:icomputer}{name={\role{inference computer}},description={},plural={\role{inference computers}}}
\newglossaryentry{r:acomputer}{name={\role{audit computer}},description={},plural={\role{audit computers}}}
\newglossaryentry{r:auditrequester}{name={\role{client}},description={}}
\newglossaryentry{r:client}{name={\role{client}},description={}}

\newglossaryentry{r:participants}{name={party},description={},plural={parties}}

\newcommand{\spartyinput}{\ensuremath{D}\xspace}

\newcommand{\spredictionX}{\ensuremath{\tilde{x}}\xspace}
\newcommand{\spredictionY}{\ensuremath{\tilde{y}}\xspace}
\newcommand{\spartyinputI}{\ensuremath{\spartyinput_1}\xspace}

\newcommand{\spartyinputN}{\ensuremath{\spartyinput_m}\xspace}

\newcommand{\Unl}{\ensuremath{\textsf{Unl}}\xspace}

\newcommand{\spartyinputs}{\ensuremath{\spartyinputI, \ldots, \spartyinputN}\xspace}
\newcommand{\spartyinputsCondensed}{\ensuremath{\mathbf{\spartyinput}}\xspace}
\newcommand{\secpartyinputs}{\ensuremath{\secs{\spartyinputI}, \ldots, \secs{\spartyinputN}}\xspace}

\newcommand{\secspartyinputsCondensed}{\ensuremath{\secs{\mathbf{\spartyinput}}}\xspace}

\newcommand{\scache}{\ensuremath{\secs{\mathbf{C}}}\xspace}
\newcommand{\scacheAbb}{\ensuremath{\ssAbb{\mathbf{C}}}\xspace}

\newcommand{\sample}{\ensuremath{\overset{{\scriptscriptstyle\$}}{\leftarrow}}}
\newcommand{\sring}{\ensuremath{\mathbb{Z}}\xspace}

\newcommand{\aux}{\ensuremath{\textsf{aux}}}

\newcommand{\secs}[1]{\ensuremath{[\![\,#1\,]\!]}}

\newcommand{\protocolOursPreprocessing}{\ensuremath{\sprotocol_\mathsf{Grad-Pre}}\xspace}
\newcommand{\protocolOursOnline}{\ensuremath{\sprotocol_\mathsf{Grad-Online}}\xspace}

\newcommand{\protocolOursHeuristicOnline}{\ensuremath{\sprotocol_\mathsf{Grad-Heuristic-Online}}\xspace}

\newcommand{\mat}[1]{\ensuremath{\mathsf{#1}}\xspace}

\newcommand{\malevent}{misclassification event\xspace}
\newcommand{\sNClasses}{\ensuremath{C}\xspace}
\newcommand{\sequal}{\ensuremath{\mathbf{u}_{\sNClasses}}\xspace}

\makeatletter
\newsavebox{\@brx}
\newcommand{\llangle}[1][]{\savebox{\@brx}{\(\m@th{#1\langle}\)}%
	\mathopen{\copy\@brx\kern-0.6\wd\@brx\usebox{\@brx}}}
\newcommand{\rrangle}[1][]{\savebox{\@brx}{\(\m@th{#1\rangle}\)}%
	\mathclose{\copy\@brx\kern-0.6\wd\@brx\usebox{\@brx}}}
\makeatother

\newcommand{\Sh}{\ensuremath{\mathsf{sh}}}

\newcommand{{\piaSh}}[1]{\ensuremath{\Pi^{#1}_{\Sh}}}

\definecolor{lightmintbg}{rgb}{.53,.80,.92}

\tikzstyle{maldo} = [rectangle, minimum width=2.5cm, minimum height=0.2cm, text centered, draw=black, fill=orange!75]
\tikzstyle{hdo} = [rectangle,,minimum width=2.5cm, minimum height=0.25cm,text centered, draw=black, fill = lightmintbg!60]
\tikzstyle{malserver} = [rectangle,minimum width=0.7cm, minimum height=0.3cm, text centered, draw=black, fill=red!60]

\tikzstyle{hserver} = [rectangle,minimum width=0.7cm, minimum height=0.3cm, text centered, draw=black, fill=lightmintbg!60]
\tikzstyle{soc} = [rectangle, rounded corners,dashed,minimum width=3.6cm, minimum height=2.6cm, draw=black]

\tikzstyle{nota} = [rectangle, minimum width=8.5cm, minimum height= 0.4cm, font = \small, draw=black]

\tikzstyle{mdnota} = [rectangle, minimum width=0.4cm, minimum height=0.2cm, text centered,font = \small, draw=black, fill=orange!75]

\tikzstyle{msnota} = [rectangle,minimum width=0.4cm, minimum height=0.2cm, text centered, font = \small, draw=black, fill=red!60]

\tikzstyle{hnota} = [rectangle, minimum width=0.4cm, minimum height=0.2cm, text centered,font = \small, draw=black, fill=lightmintbg!60]
\tikzstyle{sarrow} = [ultra thin, <->,latex-latex]
\tikzstyle{darrow} = [thin,->,>=stealth]

\newcommand{\cI}{\mathcal{I}}
\newcommand{\cL}{\mathcal{L}}

\newcommand{\cN}{\mathcal{N}}
\newcommand{\cO}{\mathcal{O}}

\newcommand{\cT}{\mathcal{T}}
\newcommand{\cX}{\mathcal{X}}
\newcommand{\cY}{\mathcal{Y}}
\newcommand{\cZ}{\mathcal{Z}}

\DeclareMathOperator*{\argmin}{arg\,min}

\newlength{\maxlen}

\setlist[description]{style=unboxed,leftmargin=0cm}

\newcounter{itemcount}

\newcommand{\figlab}[1]{\label{fig:#1}}

\newenvironment{boxfig*}[2]{%
	\begin{figure*}[h!]		
		\fontsize{5}{5}\selectfont
		\newcommand{\FigCaption}{#1}
		\newcommand{\FigLabel}{#2}
		\vspace{-.05cm}
		\begin{center}
			\begin{small}			 
				\begin{adjustbox}{max width=\textwidth}
					\begin{tabular}{@{}|@{~~}l@{~~}|@{}}
						\hline
						\rule[-1ex]{0pt}{1ex}\begin{minipage}[b]{.95\linewidth}
							\vspace{1ex}	
						}{%
						\end{minipage}\\
						\hline
					\end{tabular}	
				\end{adjustbox}		
			\end{small}
			\vspace{-0.1cm}
			\caption{\FigCaption}
			\figlab{\FigLabel}
		\end{center}
		\vspace{-.38cm}
	\end{figure*}
}

\newenvironment{myboxfig*}[2]{%
	\begin{figure*}[!htb]		
		\fontsize{5}{5}\selectfont
		\newcommand{\FigCaption}{#1}
		\newcommand{\FigLabel}{#2}
		\vspace{-.10cm}
		\begin{center}
			\caption{\FigCaption}
			\begin{small}			 
				\begin{adjustbox}{max width=\textwidth}
					\begin{tabular}{@{}|@{~~}l@{~~}|@{}}
						\hline
						\rule[-1ex]{0pt}{1ex}\begin{minipage}[b]{.95\linewidth}
							\vspace{1ex}	
						}{%
						\end{minipage}\\
						\hline
					\end{tabular}	
				\end{adjustbox}		
			\end{small}
			\vspace{-0.25cm}
			\figlab{\FigLabel}
		\end{center}
		\vspace{-.38cm}
	\end{figure*}
}

\RequirePackage{mdframed}

\newenvironment{titlebox}[5]
{\mdfsetup{
		style=#2,
		innertopmargin=1.1\baselineskip,
		skipabove={\dimexpr0.7\baselineskip+\topskip\relax},
		skipbelow={1.5em},needspace=3\baselineskip,
		singleextra={\node[#3,right=10pt,overlay] at (P-|O){~{\sffamily\bfseries #1 }};},%
		firstextra={\node[#3,right=10pt,overlay] at (P-|O) {~{\sffamily\bfseries #1 }};},
		frametitleaboveskip=9em,
		innerrightmargin=5pt
	}
	\newcommand{\TitleCaption}{#4}
	\newcommand{\TitleLabel}{#5}
	\begin{mdframed}[font=\small]
		\setlist[itemize]{leftmargin=13pt}\setlist[enumerate]{leftmargin=13pt}\raggedright%
	}
	{\end{mdframed}
	\vspace{-1.5em}
	{\captionof{figure}{\normalfont \TitleCaption}\label{\TitleLabel}}
	\medskip
}

\tikzstyle{normal} = [thick, fill=white, text=black, draw, rounded corners, rectangle, minimum height=.7cm, inner sep=3pt]
\tikzstyle{gray} = [thick, fill=gray!90, text=white, rounded corners, rectangle, minimum height=.7cm, inner sep=3pt]

\mdfdefinestyle{commonbox}{%
	align=center, middlelinewidth=1.1pt,userdefinedwidth=\linewidth
	innerrightmargin=0pt,innerleftmargin=3pt,innertopmargin=5pt,
	splittopskip=7pt,splitbottomskip=5pt
}
\mdfdefinestyle{roundbox}{style=commonbox,roundcorner=5pt,userdefinedwidth=\linewidth}

\newenvironment{systembox*}[3]
{\begin{strip}
		\vspace{\baselineskip}\begin{titlebox}{Functionality \normalfont #1}{roundbox}{normal}{#2}{#3}}
		{\end{titlebox}
\end{strip}}

\newenvironment{gsystembox*}[3]
{\begin{strip}
		\vspace{\baselineskip}\begin{titlebox}{Global Functionality \normalfont #1}{roundbox}{normal}{#2}{#3}}
		{\end{titlebox}
\end{strip}}

\newenvironment{protocolbox*}[3]
{\begin{strip}
		\begin{titlebox}{Protocol \normalfont #1}{commonbox}{normal}{#2}{#3}}
		{\end{titlebox}
\end{strip}}

\newenvironment{algobox*}[3]
{\begin{strip}
		\begin{titlebox}{Algorithm \normalfont #1}{commonbox}{normal}{#2}{#3}}
		{\end{titlebox}
\end{strip}}

\newenvironment{reductionbox*}[3]
{\begin{strip}
		\begin{titlebox}{Reduction \normalfont #1}{commonbox}{normal}{#2}{#3}}
		{\end{titlebox}
\end{strip}}

\newenvironment{gamebox*}[3]
{\begin{strip}
		\begin{titlebox}{Game \normalfont #1}{commonbox}{gray}{#2}{#3}}
		{\end{titlebox}
\end{strip}}

\newenvironment{simulatorbox*}[3]
{\begin{strip}
		\begin{titlebox}{Simulator \normalfont #1}{commonbox}{normal}{#2}{#3}}
		{\end{titlebox}
\end{strip}}

\newcounter{func}
\renewcommand{\thefunc}{F.\arabic{func}}

\newenvironment{functionalitybox*}[2]
{
\noindent
\begin{tcolorbox}[float,floatplacement=h,colframe=black, width=\columnwidth, boxrule=1pt, sharp corners=all, colback=white]
\begin{minipage}{1.0\columnwidth}
\textbf{Functionality \thefunc} #1 \\

\textbf{Input:} #2 \\
\textbf{Output:}}
{\end{minipage}\end{tcolorbox}}

\mdfdefinestyle{mycommonbox}{%
	align=center, middlelinewidth=1.1pt,userdefinedwidth=\linewidth
	innerrightmargin=0pt,innerleftmargin=2pt,innertopmargin=3pt,
	splittopskip=7pt,splitbottomskip=3pt
}
\mdfdefinestyle{myroundbox}{style=mycommonbox,roundcorner=5pt,userdefinedwidth=\linewidth}

\newenvironment{titlebox*}[5]
{\mdfsetup{
		style=#2,
		innertopmargin=0.3\baselineskip,
		skipabove={1.2em},
		skipbelow={1em},needspace=3\baselineskip,
		frametitleaboveskip=5em,
		innerrightmargin=5pt
	}
	\newcommand{\TitleCaption}{#4}
	\newcommand{\TitleLabel}{#5}
	\begin{mdframed}[font=\small]
		\setlist[itemize]{leftmargin=13pt}\setlist[enumerate]{leftmargin=13pt}\raggedright%
	}
	{\end{mdframed}
	\vspace{-1.5em}
	{\captionof{figure}{\normalfont \TitleCaption}\label{\TitleLabel}}
	\medskip
}

\newenvironment{mysystembox*}[3]
{\begin{strip}
		\vspace{\baselineskip}\begin{titlebox*}{Functionality \normalfont #1}{myroundbox}{normal}{#2}{#3}}
		{\end{titlebox*}
\end{strip}}

\newenvironment{mygsystembox*}[3]
{\begin{strip}
		\vspace{\baselineskip}\begin{titlebox*}{Global Functionality \normalfont #1}{myroundbox}{normal}{#2}{#3}}
		{\end{titlebox*}
\end{strip}}

\newenvironment{myprotocolbox*}[3]
{\begin{strip}
		\begin{titlebox*}{Protocol \normalfont #1}{mycommonbox}{normal}{#2}{#3}}
		{\end{titlebox*}
\end{strip}}

\newenvironment{myalgobox*}[3]
{\begin{strip}
		\begin{titlebox*}{Algorithm \normalfont #1}{mycommonbox}{normal}{#2}{#3}}
		{\end{titlebox*}
\end{strip}}

\newenvironment{myreductionbox*}[3]
{\begin{strip}
		\begin{titlebox*}{Reduction \normalfont #1}{mycommonbox}{normal}{#2}{#3}}
		{\end{titlebox*}
\end{strip}}

\newenvironment{mygamebox*}[3]
{\begin{strip}
		\begin{titlebox*}{Game \normalfont #1}{mycommonbox}{gray}{#2}{#3}}
		{\end{titlebox*}
\end{strip}}

\newenvironment{mysimulatorbox*}[3]
{\begin{strip}
		\begin{titlebox*}{Simulator \normalfont #1}{mycommonbox}{normal}{#2}{#3}}
		{\end{titlebox*}
\end{strip}}

\makeatletter
\algnewcommand{\ExtendedState}[1]{\State
	\parbox[t]{\dimexpr\linewidth-\ALG@thistlm}{\hangindent=\algorithmicindent\strut\hangafter=3#1\strut}}
\makeatother

\algnewcommand\algorithmicinput{\textbf{Input:}}
\algnewcommand\Input{\item[\algorithmicinput]}

\algrenewcommand{\algorithmiccomment}[1]{{\color{gray}// #1}}

\newcommand{\lstep}[1]{\label{step:#1}}

\let\emptyset\varnothing

\DeclareMathOperator{\unl}{Unl}

\newcommand{\Icosmean}{\cI_{\cos}^{\mathrm{s}}}

\newcommand{\Icos}{\cI_{\cos}}
\newcommand{\Dtr}{D_{\mathrm{tr}}}

\newcommand{\topk}[1]{\mathrm{top\text{-}}#1}
\newcommand{\topki}[1]{\mathrm{top\text{-}}#1_I}
\newcommand{\topkl}[2]{\mathrm{top\text{-}}(#1, #2)}
\newcommand{\roberta}{\text{RoBERTa}_{\textsc{base}}}
\algnewcommand\algorithmicforeach{\textbf{for each}}
\algdef{S}[FOR]{ForEach}[1]{\algorithmicforeach\ #1\ \algorithmicdo}
\DeclarePairedDelimiter\norm{\lVert}{\rVert}
\newcommand{\iprod}[2]{\left\langle #1,\, #2 \right\rangle}

\newcommand{\blue}[1]{{\color{black}{#1}}}
\newcommand{\red}[1]{{\color{red}{#1}}}

\newcommand{\todoCameraReady}[1]{{\color{green}{\small{TODO Camera Ready: #1}}}}
\renewcommand{\todoCameraReady}[1]{}

\newcommand{\sysname}{\text{UTrace}\xspace}
\newcommand{\myparagraph}[1]{\noindent \textbf{#1.}}
\newcommand{\syscamel}{\ensuremath{F_\text{Unl}}\xspace} %
\newcommand{\sysgradient}{\ensuremath{F_\text{Grad}}\xspace}

\begin{document}
\title{\sysname: Poisoning Forensics for Private Collaborative Learning}

\makeatletter
\newcommand{\linebreakand}{%
  \end{@IEEEauthorhalign}
  \hfill\mbox{}\par
  \mbox{}\hfill\begin{@IEEEauthorhalign}
}
\makeatother

\author{
\IEEEauthorblockN{Evan Rose\textsuperscript{*}}
\IEEEauthorblockA{Northeastern University\\
rose.ev@northeastern.edu}\and
\IEEEauthorblockN{Hidde Lycklama\textsuperscript{*}}
\IEEEauthorblockA{ETH Zurich\\
hidde.lycklama@inf.ethz.ch}\and
\IEEEauthorblockN{Harsh Chaudhari}
\IEEEauthorblockA{Northeastern University\\
chaudhari.ha@northeastern.edu}\linebreakand
\IEEEauthorblockN{Niklas Britz}
\IEEEauthorblockA{ETH Zurich\\
nbritz@student.ethz.ch}\and
\IEEEauthorblockN{Anwar Hithnawi}
\IEEEauthorblockA{University of Toronto\\
anwar.hithnawi@cs.toronto.edu}\and
\IEEEauthorblockN{Alina Oprea}
\IEEEauthorblockA{Northeastern University\\
a.oprea@northeastern.edu}
}

\maketitle

\makeatletter
\begingroup
  \renewcommand{\thefootnote}{\fnsymbol{footnote}}
  \long\def\@makefntext#1{%
    \noindent\hbox to 1.8em{\hss\@makefnmark\ }#1%
  }%
  \footnotetext[1]{Equal Contribution}%
\endgroup
\makeatother

\begin{abstract}

   Privacy-preserving machine learning (PPML) systems enable multiple data owners to collaboratively train models without revealing their raw, sensitive data by leveraging cryptographic protocols such as secure multi-party computation (MPC).
   While PPML offers strong privacy guarantees, it also introduces new attack surfaces: malicious data owners can inject poisoned data into the training process without being detected, thus undermining the integrity of the learned model.
   Although recent defenses, such as private input validation within MPC, can mitigate some specific poisoning strategies, they remain insufficient, particularly in preventing stealthy or distributed attacks.
   As the robustness of PPML remains an open challenge, strengthening trust in these systems increasingly necessitates post-hoc auditing mechanisms that instill accountability.
   In this paper we present UTrace, a framework for user-level traceback in PPML that attributes integrity failures to responsible data owners without compromising the privacy guarantees of MPC. UTrace encapsulates two mechanisms: a gradient similarity method that identifies suspicious update patterns linked to poisoning, and a user-level unlearning technique that quantifies each user's marginal influence on model behavior.
   Together, these methods allow UTrace to attribute model misbehavior to specific users with high precision.
   We implement UTrace within an MPC-compatible training and auditing pipeline and evaluate its effectiveness on four datasets spanning vision, text, and malware.
   Across ten canonical poisoning attacks, UTrace consistently achieves high detection accuracy with low false positive rates.
\end{abstract}

\IEEEpeerreviewmaketitle

\section{Introduction}
Privacy-preserving machine learning (PPML)~\cite{mohassel_secureml_2017,mohassel_aby3_2018,wagh_falcon_2020,rathee_cryptflow2_2020,patra_aby20_2021,abspoel_secure_2021} is an emerging collaborative learning paradigm that enables multiple data owners to jointly train a machine learning (ML) model on their combined datasets without revealing any party’s underlying data.
This approach is particularly beneficial in settings where data sharing across multiple sources is restricted due to privacy concerns and regulatory constraints, such as in the healthcare and finance sectors.
By enabling models to be trained on a diverse collection of datasets, collaborative learning can help unlock the potential of ML in areas where limited data availability has traditionally hindered progress.

PPML systems often rely on cryptographic techniques, such as secure multi-party computation~(MPC)~\cite{mohassel_secureml_2017,mohassel_aby3_2018,wagh_falcon_2020,rathee_cryptflow2_2020,patra_aby20_2021,abspoel_secure_2021}, to ensure the privacy of training data. 
While some protocols offer robustness against clients actively deviating from the protocol \cite{wagh_falcon_2020, dalskov_fantastic_2021}, they remain vulnerable to data poisoning attacks where adversaries introduce manipulated input data to influence the outcome of the ML model \cite{chaudhari_safenet_2022}.
Mitigating these attacks is particularly challenging in PPML, where conventional remedies, such as data sanitization~\cite{Activation_Clustering,CertifiedPoisoning,SpectralSign,SPECTRE}, are not directly applicable due to privacy constraints. Although some existing approaches in PPML, such as privacy-preserving input checks \cite{RoFL,ACORN,HOLMES} or ensemble training \cite{chaudhari_safenet_2022}, can help mitigate certain poisoning attacks, they cannot prevent all forms of data poisoning, leaving models susceptible to manipulation.

In response to growing concerns around AI safety and the realization that solely preventive defenses are not sufficient, there is increasing interest in auditing and forensic tools that hold models accountable after deployment\cite{Birhane2024-iy}.
Securing the entire ML lifecycle therefore necessitates a \emph{defense-in-depth strategy}: proactive safeguards during training, complemented by reactive mechanisms that detect failure and misconduct and attribute responsibility at deployment.

A representative class of post-hoc security 
mechanisms is poisoning traceback, which enables 
operators to identify the malicious training data 
responsible for model misbehavior once suspicious 
outputs are detected~\cite{PoisonForensics}. For
example, if a deployed model exhibits unexpected
behavior, traceback techniques can isolate specific
poisoned samples that induced the fault, facilitating
targeted forensic analysis and potential remediation.
However, such techniques inherently disclose additional
information about the system under inspection, raising
critical concerns about privacy and the risk of misuse. 
In the context of privacy-preserving machine learning
(PPML), it is essential that post hoc auditing
mechanisms maintain the underlying privacy guarantees
of the system. To address this tension, we argue that auditing at the party level, %
which shifts the granularity of the analysis from individual samples to the aggregate
behavior of each contributing party, offers a more
privacy-conscious alternative.
This mitigates the risk of exposing fine-grained or
identifiable information. For instance, in financial
applications, disclosing per-sample similarity metrics
may inadvertently reveal sensitive attributes of honest
users. However, auditing at the organization level supports
meaningful accountability while reducing the likelihood
of unintended data leakage.

In this paper, we introduce UTrace, a framework for user-level traceback of data 
poisoning in secure collaborative machine learning. UTrace integrates two general and 
distinct traceback mechanisms: one based on gradient similarity and another on approximate 
unlearning.
The gradient-based mechanism assigns a responsibility score to each user by measuring how 
closely their data gradients align with those that minimize the attacker's objective. This 
score is computed by aggregating gradient alignment across held-out probe points and 
reflects a core assumption of many poisoning attacks: poisoned points actively steer 
training toward malicious goals.
The second mechanism approximates user-level unlearning by measuring how the model's loss 
changes when a user's dataset is removed. This approximation avoids full retraining and 
proves particularly effective under distributed attacks. 
Although these two techniques capture different statistical signals, one from individual 
sample gradients and the other from user-wide removal effects, they are both general and 
attack agnostic. UTrace explicitly addresses challenging cases such as very low poisoning 
rates and coordinated attacks by focusing analysis on each user's most influential 
training points.
\blue{
We implement these traceback algorithms as privacy-preserving auditing functions within a secure multi-party computation (MPC) framework. 
To enable efficient execution, we extend this framework with support for secure preprocessing, removing the requirement for data owners to remain online at audit time.
Specifically, we introduce a secret-shared caching mechanism that stores intermediate preprocessing results, allowing audits to reuse prior computations without accessing raw data or requiring persistent input availability. This design significantly reduces the online overhead and enables audits to be conducted asynchronously, even when data owners are offline.
}
We further optimize the cryptographic protocols for secure cosine similarity, oblivious sample selection and delayed truncation.
We implement UTrace within a secure collaborative training and auditing pipeline and 
evaluate it on diverse datasets from vision, text, and malware domains. Across ten 
canonical poisoning attacks, UTrace consistently achieves high attribution accuracy and 
low false-positive rates, even under low-rate and distributed poisoning scenarios. Our 
results show that accountability after training, as enabled by UTrace, is a practical and 
necessary step toward robust and trustworthy privacy-preserving machine learning.

\section{Background and Related Work} 
We provide background on MPC-based PPML and data poisoning attacks, followed by a discussion of related work.

\subsection{Secure Multi-Party Computation}

There has been significant progress in PPML
in recent years, with secure multi-party computation (MPC) emerging as prominent approach that offers
the best performance by distributing trust among $n$ parties~\cite{goldreich_how_1987,ben-or_completeness_1988,damgard_multiparty_2011}.
In MPC, multiple parties
jointly perform training or inference without revealing their private inputs.
Currently, the most efficient MPC protocols for PPML rely on homomorphic secret sharing over a field $\mathbb{F}_p$ or ring $\mathbb{R}_{2^k}$~\cite{Keller2022-quantizedtraining}.
In the \textit{outsourced computation} paradigm, an arbitrary number of data owners generate secret shares and distribute them to a (typically small) set of servers that collaboratively perform the computation.
More recently, frameworks such as Arc~\cite{lycklama_holding_2024} have laid out the cryptographic building blocks, including lightweight proofs and efficient verification techniques, needed to audit MPC-based PPML computations.

\subsection{Data Poisoning Attacks}

A data poisoning attack is an adversarial manipulation of the training dataset in which an attacker injects malicious samples to influence the behavior of the learned model.
While poisoning has been explored across a wide range of machine learning tasks~ \cite{jagielski_manipulating_2021, carlini_poisoning_2022}, we focus on attacks targeting classification~\cite{gu_badnets_2019, geiping_witches_2021}.
Among these, \textit{backdoor attacks} are the most widely studied; here, the adversary's objective is to make the victim misclassify samples containing a particular perturbation (the ``trigger'') while leaving the model's behavior on clean inputs unaffected~\cite{gu_badnets_2019, saha_hidden-trigger_2020, souri_sleeper_2021}. Another class of attacks, known as~\textit{instance-targeted poisoning}, aims to induce specific model outputs on a preselected set of victim inputs~\cite{shafahi_poison_2018, aghakhani_bullseye_2021, geiping_witches_2021}. Finally,~\textit{indiscriminate attacks} seek to degrade overall model accuracy without targeting specific samples~\cite{lu_indiscriminate_2024, CertifiedPoisoning}.

\subsection{Related Work}
We discuss related work on poisoning defenses and forensics, especially as it relates to PPML.

\myparagraph{Poisoning Defenses}
In response to poisoning attacks, an extensive body of work has emerged to harden the training process. Existing defenses can be broadly categorized into heuristic defenses and certified defenses.
Heuristic defenses typically leverage certain assumptions about the clean training data, or the nature of the poisoning attack \cite{diakonikolas_sever_2018,wang_neural_2019, SpectralSign, liu_fine-pruning_2018, peri_deep_2019, jagielski_manipulating_2021}. Some methods directly identify and remove malicious samples, for instance using spectral analysis of sample gradients~\cite{diakonikolas_sever_2018} or intermediate network activations \cite{SpectralSign}. 
Other methods mitigate the impact of poisoning as a post-training step. Example techniques include pruning uninformative neurons \cite{liu_fine-pruning_2018} or training on randomized labels followed by clean fine-tuning \cite{heng_selective_2023}.

Certified defenses provide provable guarantees on the stability of classifier predictions subject to small changes to the training set within some perturbation model.
These defenses produce a certificate specifying a lower bound on the dataset perturbation required for a certain decrease in test accuracy.
Most certified methods work by training a large ensemble of models under a random transformation to the training data \cite{jia_intrinsic_2020,chen_framework_2020,levine_deep_2021,rezaei_run-off_2023,zhang_pecan_2024}, while other works achieve certification through a careful analysis of a particular learning algorithm \cite{rosenfeld_certified_2020, sosnin_certified_2024}.
However, certified defenses face significant limitations in secure settings, when data is encrypted.
In such cases, the underlying assumptions of the defenses, such as the expectation that large perturbations will be detected by visual inspection, may not hold, thereby reducing their practical applicability.
In the collaborative learning setting, Chaudhari et al. propose the use of model ensembles to mitigate poisoning attacks in the secure setting~\cite{chaudhari_safenet_2022}.
\blue{While this approach can mitigate certain classes of poisoning attacks, particularly those that induce significant model divergence, it remains vulnerable to more subtle or distributed attacks.}

\myparagraph{Poisoning Traceback}
Recent work has focused on forensics in various ML contexts, such as privacy violations~\cite{liu_tracing_2024}, reconstructing backdoor triggers~\cite{cheng_beagle_2023}, and tracing training data samples responsible for poisoning attacks in centralized training settings~\cite{PoisonForensics, hammoudeh_identifying_2022}.
The latter aims to identify the source of issues by identifying relevant samples responsible for a poisoning attack.
For example, Poison Forensics~\cite{PoisonForensics} identifies points by clustering candidate poisons in gradient space.
The most relevant work to ours is the gradient aggregated similarity (GAS) score \cite{hammoudeh_identifying_2022}, which uses a gradient similarity metric to attribute classification events to individual training data samples.
Other works \cite{cao_fedrecover_2022, jia_tracing_2024} focus on forensics for poisoning in the federated learning setting.
However, these techniques require access to raw client gradient updates~\cite{Boenisch2021-ee}, which can still contain sensitive information, making them incompatible with privacy-preserving settings.

\section{Problem Statement and Threat Model}

We define our scenario for user-level poisoning traceback and describe our threat model.
We then outline design goals and challenges for user-level traceback algorithms.

\subsection{Problem Statement}
We consider a PPML training scenario in which data owners, model owners, and computing parties provide a classification service to clients that offers secure model training, inference, and traceback of malicious data owners.

\myparagraph{PPML Scenario}
We consider a typical collaborative ML setting (cf. \Cref{fig:problem-setup}), where $m$ \emph{data owners} $U_1, \ldots, U_m$ provide the training data used in a privacy-preserving training process, which is realized via MPC~\cite{abspoel_secure_2021,mohassel_secureml_2017, mohassel_aby3_2018, wagh_falcon_2020}.
Neither the other data owners nor any other parties involved in the MPC computation can gain access to the underlying private training data.
At the end of the collaborative training, the MPC protocol outputs an ML model (classifier) $f_\theta$ trained on the joint dataset $D = \cup^N_{i=1} D_i$.
The resulting model $f_\theta$ is given to the model owners.
Additionally, the training procedure generates a record $R$ of the training process (that might include model checkpoints, training hyperparameters, validation metrics, or other train-time artifacts).
Note, that the model itself might reveal some information about the private training data. This can be addressed through, e.g., Differential Privacy~\cite{dp_sgd}, and such mechanisms can also be adapted to collaborative ML, but this is outside the scope of this work.
During deployment, the model is used to perform inference for clients using a secure inference protocol to protect the privacy of the client inputs and the model.
\todoCameraReady{Do we comment here on mechanics  of what things are secretely shared and what is not? E.g. is model sent in clear, what about model record}

\begin{figure*}[tp]
    \centering
    \vspace{-0.2cm}
    \includegraphics[width=0.8\textwidth]{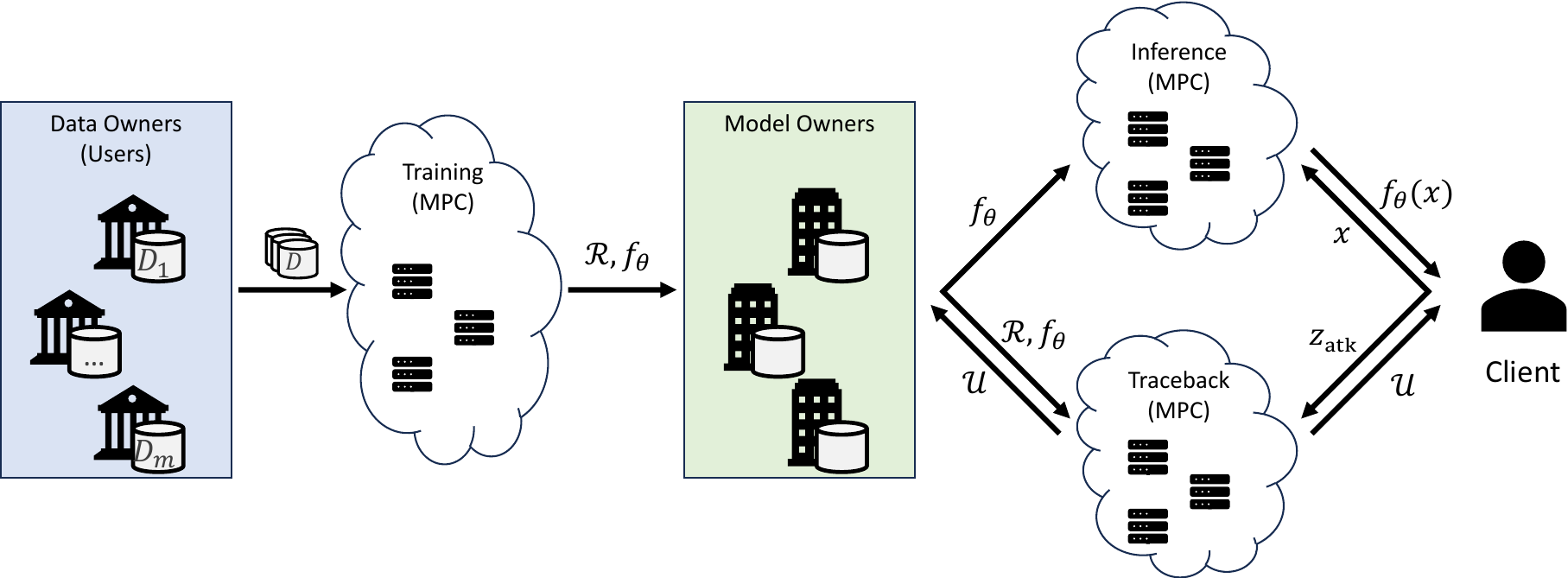}
    \caption{The PPML scenario we consider includes data owners contributing their private datasets to train an ML model $f_{\theta}$ distributively in MPC. MPC servers are ephemeral and do not store long-term state. Model owners store the model $f_\theta$ and training records $R$ that are used for private inference. A traceback service monitors suspicious activity during model deployment and launches a procedure to identify malicious data owners $\mathcal{U}$ given a suspicious input $z_{\mathrm{atk}} = (\tilde{x},\tilde{y})$.}
    \vspace{-0.2cm}
    \label{fig:problem-setup}
\end{figure*}

\myparagraph{User-level Poisoning Traceback} 
At deployment time, our traceback service receives misclassification events from clients, which is a pair of a prediction sample $\tilde{x}$ that led to an unexpected output $\tilde{y}$.
The goal of the service is to analyze these events to identify compromised data owners responsible for poisoning.
To do so, it invokes one of several available auditing functions, selected from the set of supported traceback mechanisms in the system.
The traceback procedure outputs for the model owners a ranking of all users based on raw responsibility scores, and a list of accused users identified as malicious, typically by thresholding the responsibility scores. 
While our service supports arbitrary functions, our focus is on \textit{user-level} poisoning traceback, in contrast to prior work that identifies samples~\cite{PoisonForensics}.

\subsection{Threat Model} \label{sec:threat_model}
We consider an actively malicious adversary that can compromise parties up to half of the data owners, model owners and computational parties across the training, inference, and traceback phases. 
The adversary can observe and modify all inputs, states and network traffic of the compromised parties.
For the compromised data owners, the adversary controls their training datasets to mount a data poisoning attack on the ML model. We consider an \textit{error-specific} misclassification objective~\cite{biggio_wild_2018}, that is, inducing a specific prediction $f_\theta(x) = \hat{y} \ne y$, where $\hat{y}$ is a target label and $y$ is the true label for the point $x$. This objective can be achieved through a number of methods, including backdoor attacks~\cite{gu_badnets_2019,souri_sleeper_2021,saha_hidden-trigger_2020,qi_mind_2021,carlini_poisoning_2022} and label-flipping attacks~\cite{jagielski_subpopulation_2021}. In terms of adversarial knowledge,  we assume the strongest white-box adversary with full knowledge of the training algorithm, the traceback mechanism, and training datasets of compromised clients. We assume a strong adversarial model, in which the malicious data owners might collude to achieve the same adversarial objective, for instance introduce the same backdoor pattern in the model.

Certain instantiations of secure computation might impose additional constraints on the adversary for the computational parties in the training, inference and traceback phases. For example, our system can be used with MPC protocols that assume an honest-but-curious adversary that only passively corrupts the parties \cite{araki_high-throughput_2016,chandran_simc_2022}, which are often significantly more efficient than their actively secure counterparts that consider malicious adversaries~\cite{wagh_falcon_2020, dalskov_fantastic_2021}.

\subsection{Design Goals and Challenges} \label{sec:design-goals}

A robust traceback framework for \gls{ppml} must meet several key design goals, each of which introduces specific challenges:

\noindent \textbf{High effectiveness}. The framework must accurately identify malicious users with high true positive rates (TPR) and minimal false positives (FPR), as false accusations carry legal risks. Achieving this is difficult, especially under stealthy poisoning attacks that corrupt less than 1\% of the training set~\cite{shafahi_poison_2018,severi_explanation-guided_2021}.

\noindent \textbf{Generality across attacks}. Existing defenses often assume specific attack strategies~\cite{SpectralSign, wang_neural_2019, peri_deep_2019}, data distributions~\cite{chen_robust_2013, feng_robust_2014}, or learning algorithms~\cite{CertifiedPoisoning, rosenfeld_certified_2020, levine_deep_2021}. We aim to develop deployment-level traceback mechanisms that generalize across such assumptions.

\noindent \textbf{Robustness to Coordinated Attacks}. Adversaries may distribute poisoned data across multiple users to evade detection. A robust traceback mechanism must detect such coordination.

\noindent \textbf{Privacy}.
Auditing inevitably reveals additional information about the underlying system, creating a tradeoff between accountability and confidentiality. In the context of PPML, traceback mechanisms must carefully manage this leakage to preserve user privacy. For example, sample-level traceback methods that compute and release attribution scores for each training point can inadvertently expose sensitive information about a user’s dataset, such as the types and distribution of their data samples.

\noindent \textbf{Low performance overhead}. The protocol should minimize runtime and communication rounds, scaling efficiently with the number of users $m$ and dataset sizes $\abs{D_i}$.

\section{User-Level Traceback Design} %

In this section, we construct our framework for performing user-level poisoning traceback. 
We first introduce our notion of influence score and then discuss two techniques for computing influence scores, based on gradient similarity and unlearning methods.

\subsection{Influence Scores}

Given a \malevent $(\tilde{x}, \tilde{y})$, we compute an \emph{influence score} $s_i$ for each user relative to this prediction,
where compromised parties will have a high influence score if their provided training data is poisoned in a way as to affect samples like $\tilde{x}$.
This is similar to approaches in the centralized setting that rank samples by computing per-sample influence scores for a given prediction~\cite{hammoudeh_identifying_2022,jia_towards_2023}.
We later show three different approaches to computing a per-user influence score.
First, we define a baseline instantiation of an influence score $s_i$ using kNN,
then an approach based on aggregating per-sample influence scores,
and finally a direct score computation using leave-out models.

However, simply having a high influence score does not necessarily indicate maliciousness,
as a \malevent can also arise even when no poisoning attack took place.
This is a challenge compared to existing approaches based on per-sample influence scores, as the goal is not just to rank parties by relevance but to identify potentially malicious parties.
Since accusing a user of malicious behavior could have significant negative consequences (e.g., exclusion from the system or expensive additional audits), it is essential to distinguish these scenarios in order to avoid harming honest parties.
Therefore, we consider the distribution of influence scores:
if a \malevent was significantly influenced by a small subset of outlier parties, we assume a data poisoning attack.
Conversely, if the influence is distributed evenly across all parties, we consider an attack less likely.

\myparagraph{Strawman: k-Nearest-Neighbors}
We consider using kNN, to identify the $k$ training samples closest to the prediction sample in the model's latent space as an obvious baseline.
The influence score can then be determined simply by the number of those $k$ points contributed by a given user. 
Formally, let $\Gamma$ be the set of the $k$ nearest neighbors to $\tilde{x}$, then
$s_i \leftarrow \textnormal{count}\left( \Gamma \cap D_i \right)$.
This is a simple and effective influence score if we assume that malicious training samples have small 
$l_2$-distances from $\tilde{x}$ in latent space. 
However, as we show in our evaluation, this assumption does not necessarily hold for attacks in practice.

\subsection{Gradient Similarity}
\label{sec:gradient}

A natural starting point for attributing responsibility in poisoning attacks is to leverage existing methods from model explainability, which quantify how individual training samples affect model behavior~\cite{ghorbani_data_2019, pruthi_estimating_2020, jia_towards_2023}. 
These techniques can be adapted to the user level by aggregating sample-level influence scores according to user ownership. In this section, we build on one such method, Gradient Aggregated Similarity (GAS)~\cite{hammoudeh_identifying_2022}, and extend it to efficiently compute user-level responsibility scores under poisoning.

Gradient similarity methods can be viewed as a type of counterfactual analysis: If the model $\theta$ is trained on data point $z$, in what direction does the loss on test point $\hat{z}$ move? This connects to a basic assumption we make about the malicious nature of poisoning attacks -- that benign points do not have sufficiently high gradient alignment to decrease the loss on the attack sample, and thus poisoned points must induce gradients with higher alignment in order to manipulate the model's behavior. This assumption is exceedingly general, making gradient-based scores appropriate as a general-purpose influence measure that does not rely on additional assumptions about the attack strategy.

\myparagraph{User-Level Gradient Scores}\label{sec:extending-to-gradients}
The main challenge is how to create user-level gradient similarity scores that capture a wide range of poisoning attacks, while preserving individual sample privacy in MPC. We begin by recalling the GAS score~\cite{hammoudeh_identifying_2022}. Before training, an iteration set $\cT \subseteq [T]$ is fixed, where $T$ is the number of training steps. Typically $\abs{\cT} \ll T$. During training, in iterations $t \in \cT$, the intermediate model parameters $\theta_t$ are recorded along with the current learning rate $\eta_t$. The sample-level influence score $\Icos$ between training point $z$ and test point $\hat{z}$ is then defined as
\begin{equation}\label{eq:gas}
    \Icos(z, \hat{z}) = \sum_{t \in \cT} \eta_t \frac{\iprod{\nabla_\theta \ell(\theta_t; z)} {\nabla_\theta \ell(\theta_t; \hat{z})}}{\norm{\nabla_\theta \ell(\theta_t; z)}_2 \norm{\nabla_\theta \ell(\theta_t; \hat{z})}_2}.
\end{equation}

We define a user-level event responsibility metric based on weighted sums of gradient similarities.
For a dataset $D$ and model parameters $\theta$, let $\cL(\theta; D) = \sum_{z \in D} \ell(\theta; z)$. Given a training record $R$ consisting of model checkpoints $\theta_t$ and learning rates $\eta_t$, we define the $\Icosmean$  influence score between a user's dataset $D_i$ and test point $\hat{z} $ to be
\begin{align}\label{eq:Icosmean}
    &\Icosmean(D_i, \hat{z}, R) \nonumber \\
    &\quad:= \frac{1}{\abs{D_i}} \sum_{z \in D_i} \sum_{t \in \cT} \eta_t \frac{\iprod{\nabla_\theta \ell(\theta_t; z)} {\nabla_\theta \ell(\theta_t; \hat{z})}}{\norm{\nabla_\theta \ell(\theta_t; z)}_2 \norm{\nabla_\theta \ell(\theta_t; \hat{z})}_2} \nonumber \\
    &\quad= \frac{1}{\abs{D_i}} \sum_{z \in D_i} \Icos(z, \hat{z}) 
\end{align}
Restated, $\Icosmean$ is a weighted sum of the mean gradient cosine similarity with $\hat{z}$, taken over the samples of $D_i$.

\myparagraph{Top-$k$ selection}
Score $\Icosmean$  assumes that the gradient signal from the poisoned data is strong enough to outweigh contributions from other data points in the adversarial user dataset $D_i$. However, when poisoning occurs at exceptionally low poisoning rates (e.g., 1\% of user data), the relevant terms become diluted and fail to indicate malicious behavior. To address this challenge, our influence score should primarily target the malicious portions of a corrupted dataset.

Our main insight is to identify the most influential samples within each user's data and include them in the  aggregated user score.  As $\Icosmean$ aggregates GAS scores across samples, we propose to identify first the top-k samples with highest GAS scores and then aggregate only scores for those samples. If we define:
\begin{align}
    S_i = \{\Icosmean(z, \hat{z}) : z \in D_i\}
\end{align}
then the user responsibility score is:
\begin{equation}\label{eq:topk}
    s_i = \tfrac{1}{k} \sum_{j=1}^k \topk{k}(S_i)_j
\end{equation}
where $\topk{k}(S_i)$ is a multiset containing the largest $k$ elements of $S_i$.

The hyperparameter $k$ determines the sensitivity of the traceback tool to low poisoning rates. At $k=1$, user scores are determined according to the single most suspicious data point in a user's dataset. As $k$ grows large, the reduction method approaches $\Icosmean$. In general, appropriate choices for $k$ depend on the learning task and network architecture.

\myparagraph{Cost-effective gradient computation}
We introduce two additional optimizations to reduce the computational overhead of gradient evaluation in secure settings.
First, we limit the computation to gradients from the later layers of the model, leveraging the fact that, when computing gradients using the backpropagation algorithm, a layer's gradient does not depend on earlier layers. Empirically, we find that using gradients from the final one or two layers of the network is generally sufficient for traceback performance.
Second, we apply data-oblivious dimensionality reduction techniques to produce compact gradient sketches that approximately preserve inner products. This significantly reduces storage, communication and computation costs without compromising accuracy.
\ifdefined\isnotextended
We provide more details on gradient projection and storage in \Cref{appx:gradient} in the extended version \citeextended. 
\else
We provide more details on gradient projection and storage in \Cref{appx:gradient}.
\fi

\myparagraph{Full Procedure}
We describe the $\sysname$ traceback procedure in \Cref{alg:traceback}. The training output consists of both the trained model parameters $\theta_t$ and the training record $R$. Each element in $R$ contains the learning rate $\eta_t$, intermediate model parameters $\theta_{t-1}$, the training data gradients $g_t$, and the projection matrix $G_t$ for some training iteration $t \in \cT$. Given a misclassification event $z_{\mathrm{atk}}$, traceback computes user responsibility scores and outputs a ranked list of users based on these scores.

\begin{algorithm}[htb]
\caption{\sysgradient (Gradient Similarity)}
\label{alg:traceback}
\begin{algorithmic}[1]
    \Input
    User datasets $D_1, \ldots, D_m$,
    Misclassification event $\tilde{x},\tilde{y}$,
    Training record $R = \{(\theta_{t-1}, \eta_t, G_t)\}_{t \in \cT}$,
    Score parameter $k$
    \Ensure
    Ranked owners by responsibility scores $\mathcal{U}$

    \State \Comment{Pre-process projected traceback gradients}
    \State $\Dtr \gets \text{Concat}(D_1, \ldots, D_m)$
    \For{$t \in \cT$}
        \For{$i \gets 1$ to $\abs{\Dtr}$}
                    \State $x_i, y_i \gets \Dtr[i]$
                    \State $g_t^{(i)} \gets G_t \, \nabla_{\theta_W} \ell(\theta_{t-1}; x_i, y_i)$
        \EndFor
    \EndFor
    \State \Comment{Compute gradient similarity}
    \For{$t \in \cT$}
        \State $\widehat{g_t} \gets G_t \nabla_{\theta_W} \ell(\theta_{t-1}; \tilde{x})$ \Comment{$G_t, \theta_{t-1}$ from $R$}
    \EndFor
    \For{$i = 1, \ldots, m$}
        \State $I_i \gets \textsc{Indices}(i)$ \Comment{User $i$'s indices}
        \State $S_i \gets \left\{\sum_{t \in \cT} \eta_t \frac{\iprod{g_t^{(\mathrm{idx})}}{\widehat{g_t}}}{\norm*{g_t^{(\mathrm{idx})}}_2 \norm*{\widehat{g_t}}_2} : \text{idx} \in I_i \right\}$
        \vspace{0.25em}

        \State $s_i \gets \frac{1}{k} \sum_{j=1}^k \topk{k}(S_i)_j$
    \EndFor
    \State \Return $\textsc{RankUsers}({s_1}, \ldots, {s_m})$
\end{algorithmic}
\end{algorithm}

\subsection{Unlearning}
\label{sec:unlearning}
Instead of lifting per-sample influence scores to the user level, we can also consider the influence at the user level directly:
if a \malevent was (at least partially) the result of poisoned data provided by a user,
removing the data of that malicious user would result in the absence (or weakening) of the attack at inference time.
Therefore, one could train a leave-out model $\smodelI$ for each user, using the same dataset but \emph{leaving out} user $i$'s dataset.
Given these, we could compute the influence score as the loss of each model on the \malevent, i.e., $s_i \leftarrow \ell(\theta_i; \tilde{x}, \tilde{y})$.
Note that, since the $\smodel^i$ do not depend on $\tilde{x}$ or $\tilde{y}$, we can use the same set of leave-out models to handle auditing requests for different events.

However, retraining such leave-out models is prohibitively expensive, as it requires training $N$ additional models, each incurring roughly the same costs as the original secure collaborative learning process.
Therefore, our approach instead relies on techniques for \emph{unlearning} a user's data points from the full model~\cite{Cao2015-st,Guo2019-bg,Bourtoule2019-mc}.
While \emph{certified unlearning} approaches provide formal guarantees on their equivalence to leave-out models,
these techniques currently do not scale well to scenarios like ours, 
where large amounts of data points must be removed~\cite{Bourtoule2019-mc,Warnecke2021-aq}.
Instead, we use an efficient unlearning technique~\cite{PoisonForensics} that can scale to the demands of our setting.
Specifically, $\Unl(\smodel, D_i, D, E)$ applies first-order updates to minimize the objective
\begin{equation}
\label{eq:objective}
    \left( \sum_{(x,y) \in D \setminus D_i} \ell\left( \theta; x,y \right) + \sum_{(x,y) \in D_i} \ell\left( \theta; x,\sequal \right) \right),
\end{equation}
where $\ell$ is the cross-entropy loss function, $E$ is the number of epochs, and
\sequal is defined as the uniform probability vector, i.e.,
\mbox{$\sequal = \left( \frac{1}{\sNClasses}, \frac{1}{\sNClasses}, \ldots, \frac{1}{\sNClasses} \right)$}
, representing the model's output when it is uncertain about its prediction~\cite{Vyas2018-jt,Lee2018-wr}.
Our approach requires only a small number of unlearning steps, which can be significantly more efficient than retraining.
While the concrete number is task-dependent, in practice, our unlearning approach requires a fraction of the epochs compared to the full training.

\begin{algorithm}
  \caption{\syscamel (Unlearning)}
  \label{alg:unlearning}
  \begin{algorithmic}[1]             %
  \Input
    User datasets \spartyinputs,
    Misclassification event $\tilde{x}, \tilde{y}$,
    Model \smodel,
    Unlearning epochs $E$
    \Ensure
    Ranked owners by responsibility scores $\mathcal{U}$
        \State \Comment{Pre-process unlearned model}
      \For{$i =1,\ldots,m$}
        \State $\spartyinput_{-i} \gets
               \displaystyle\bigcup_{j \neq i} \spartyinput_{j}$
        \State $\smodelWUnl \gets
               \Unl{(\smodelW,\ \spartyinput_{i},\ \spartyinput_{-i},\ E)}$
      \EndFor
      \State \Comment{Compute leave-out loss}
      \For{$i = 1,\ldots,m$}
        \State $s_i \gets \ell\!\bigl(\smodelWUnl;\tilde{x},\ \tilde{y}\bigr)$
      \EndFor

    \State \Return $\textsc{RankUsers}(s_1, \ldots, s_m)$
  \end{algorithmic}
\end{algorithm}

The coarseness of our unlearning approximation is intentional.
The unlearning objective increases the loss on all samples from a given user, including benign ones, resulting in large gradients and a degradation of the model's overall performance on the learning task. 
However, this degradation is acceptable in our setting, where the goal is not to maintain generalization, but to ensure that the model exhibits a clear loss difference on the specific malicious sample under audit.
In fact, this explicit form of unlearning can be more effective than a precise leave-out approach, as it amplifies the loss on any retained malicious data.
This effect is particularly beneficial in settings with multiple attackers contributing to the same poisoning attack. 
In such scenarios, exact unlearning may allow adversaries to obscure their contributions by duplicating poisoned inputs across users, resulting in the loss for the poisoned sample to remain low because of the malicious samples contributed by malicious parties.
In contrast, our approximation increases the loss for the poisoned samples if they are in the unlearned user's dataset.
\ifdefined\isnotextended
We formally show in~\Cref*{appx:bad-true-unlearn} in the extended version that using true unlearning with models trained with empirical risk minimization fails catastrophically if poisons are duplicated across even two user datasets\citeextended.
\else
We formally show in~\Cref{appx:bad-true-unlearn} that using true unlearning with models trained with empirical risk minimization fails catastrophically if poisons are duplicated across even two user datasets.
\fi

\section{UTrace Framework}
\label{sec:framework}
In this section, we introduce our framework and how we securely realize our user-level auditing functions,
ensuring end-to-end privacy and correctness.

\myparagraph{Cryptographic Auditing}
When a client detects a suspicious \malevent, they initiate an audit to start a user-level traceback function. 
This triggers an interactive MPC protocol involving the client, the model owner, the data owners, and a set of dedicated MPC servers. 
The servers compute the user-level traceback function $F$, after receiving input from the input parties, who secret-share their inputs with the MPC servers.
The servers then validate these inputs against cryptographic commitments that were generated during the training and inference phases.
These commitments, ensuring integrity of the training data, model, and prediction input, prevent parties from altering their data during the auditing phase.
While several techniques exist for enforcing input consistency, we build upon the Arc framework~\cite{lycklama_holding_2024}, which provides an end-to-end framework for \gls{ppml} and offers the most efficient validation mechanism. 
Once all inputs are validated, the MPC servers jointly evaluate the traceback function $F$, which in Arc is modeled cryptographically as a function of the arithmetic black-box ideal functionality \sabbAAuditID.
In the following, we discuss the limitations of this model and introduce our extensions.

\floatname{algorithm}{Protocol}
\begin{algorithm}[t]
\caption{\protocolOursPreprocessing}
\label{protocol:ours:preprocessing}
\begin{algorithmic}[1]
    \Input
        User datasets $\secs{D_1},\ldots,\secs{D_m}$, %
        Training record $\secs{R}=\{(\secs{\theta_{t-1}},\eta_t,\secs{G_t})\}_{t\in\cT}$
    \Ensure
        Pre-processed traceback gradients $\{\secs{\widehat{g_t}}\}_{t\in\cT}$

    \State $\secs{\Dtr} \gets \text{Concat}(\secs{D_1}, \ldots, \secs{D_m})$
    \For{$t \in \cT$}
        \For{$i \gets 1$ to $\abs{\Dtr}$}
                    \State $\secs{x_i}, \secs{y_i} \gets \Dtr[i]$
                    \State \label{grad:protocol:gradient} $\secs{g_t^{(i)}} \gets\idealGradient(\secs{G_t},\secs{\theta_{t-1}},\secs{x_i},\secs{y_i})$
        \EndFor
    \EndFor
    \State \Return $\{\secs{g_t^{(i)}}\}_{t\in\cT,i\in\abs{\Dtr}}$
\end{algorithmic}
\end{algorithm}

\floatname{algorithm}{Protocol}
\begin{algorithm}[t]
\caption{\protocolOursOnline}
\label{protocol:ours:online}
\begin{algorithmic}[1]
    \Input
        Misclassification event $\secs{\tilde{x}},\secs{\tilde{y}}$,
        Pre-processed training gradients $\{\secs{g_t^{(i)}}\}_{t\in\cT,i\in\abs{\Dtr}}$, 
        Training record $\secs{R}=\{(\secs{\theta_{t-1}},\eta_t,\secs{G_t})\}_{t\in\cT}$,
        Score parameters $k,l$
    \Ensure
        Ranked owners by responsibility scores $\{i:U_i\}$

    \For{$t \in \cT$}
        \State Run 
        $\secs{\widehat{g_t}}
           \gets\idealGradient(\secs{G_t},\secs{\theta_{t-1}},\secs{\tilde{x}},\secs{\tilde{y}})$
           \label{grad:protocol:online:score}
    \EndFor
    \For{$i = 1,\ldots,m$}                         \Comment{User loop} \label{grad:protocol:online:loopstart}
        \State $I_i \gets \textsc{Indices}(i)$      \Comment{User $i$’s sample indices}
        \For{$t \in \cT$}                           \Comment{Training-step loop}
            \State $\secs{\text{B}_t}
                \gets\idealDotprod(\widehat{g_t},\widehat{g_t})$
            \For{$j = 1,\ldots,\abs{I_i}$}          \Comment{Sample loop}
                \State $\secs{\mathbf{d}^{(j)}_t}
                    \gets\idealDotprod(
                        \secs{g_t^{I_i^{(j)})}},\widehat{g_t})$
                \State $\secs{\text{A}^{(j)}_t}
                    \gets\idealDotprod(
                        \secs{g_t^{I_i^{(j)})}},
                        \secs{g_t^{I_i^{(j)})}})$
                \State $\secs{\mathbf{h}^{(j)}_t}
                    \gets\idealMul\!\Bigl(
                        \secs{\mathbf{d}^{(j)}_t},
                        \idealMul\!\bigl(
                            \idealRecip(\secs{\text{A}^{(j)}_t}),
                            \idealRecip(\secs{\text{B}_t})
                        \bigr)
                    \Bigr)$
            \EndFor
        \EndFor
        \State $\secs{\mathbf{h}^{(j)}}
            \gets\sum_{t\in\cT}\eta_t\cdot\secs{\mathbf{h}^{(j)}_t}
            \quad\text{for all }j\in[\abs{I_i}]$
        \smallskip
        \State $\secs{\mat{M}}
            \gets\bigl(\secs{\mathbf{h}^{(1)}},
                       \secs{\mathbf{h}^{(2)}},
                       \dots,
                       \secs{\mathbf{h}^{(\abs{I_i})}}\bigr)$
        \State $\secs{\mat{P}}\gets\idealSort(\secs{\mat{M}}^\top)$
        \State $s_i\gets
               k^{-1}\sum_{j=1}^{k}\secs{\mat{P}}_j$
    \EndFor \label{grad:protocol:online:loopend}
    \State \Return $\textsc{RankUsers}(s_1,\ldots,s_m)$
\end{algorithmic}
\end{algorithm}

\begin{algorithm}[t]
\caption{\protocolCamelPreprocessing}
\label{protocol:camel:preprocessing}
\begin{algorithmic}[1]
\Input Datasets ${\secs{\spartyinput_1}, \ldots, \secs{\spartyinput_m}}$, Model $\secs{\smodelW}$, Number of Finetuning Epochs $E$
\Ensure Unlearned Models $\secs{\theta_{-1}},\ldots,\secs{\theta_{-m}}$
\For{$i = 1$ \textbf{to} $m$}
    \State $\secs{\spartyinput_{-i}} \gets \bigcup_{j \neq i}({\secs{\spartyinput_j}})$
    \State \label{camel:preprocessing:sgd} $\secs{\smodelWUnl} \gets \idealSGD(\ell, \secs{\smodelW}, \secs{\spartyinput_i}, \secs{\spartyinput_{-i}}, E)$
    \Comment{where $\ell$ defined as in~\Cref{eq:objective}}
\EndFor
\State \Return $\left\{\secs{\smodelW_{-1}},\ldots,\secs{\smodelW_{-m}}\right\}$
\end{algorithmic}
\end{algorithm}

\begin{algorithm}[t]
\caption{\protocolCamelOnline}
\label{protocol:camel:online}
\begin{algorithmic}[1]
\Input Prediction $\secs{\spredictionX}$, Ground truth label $\secs{\spredictionY}$, Unlearned Models $\secs{\smodelW_{-1}},\ldots,\secs{\smodelW_{-m}}$
\Ensure Ranked owners by responsibility scores $\{i:U_i\}$
\For{$i = 1$ \textbf{to} $m$}
    \State \label{camel:online:score} $\secs{s_i} \gets \idealLoss(\secs{\smodelWUnl};\secs{\spredictionX}, \secs{\spredictionY})$
\EndFor
\State \Return $\textsc{RankUsers}(\secs{s_1},\ldots,\secs{s_m})$
\end{algorithmic}
\end{algorithm}

\subsection{Sample-independent Preprocessing}
Most of the computation for our user-level auditing functions, \syscamel and \sysgradient, are  independent of the particular sample $(\tilde{x}, \tilde{y})$ under audit. 
In both cases, a substantial portion of the workload depends solely on the training data and the model parameters. 
For \syscamel the user-specific leave-out models are independent of the sample under audit, and for \sysgradient we can cache the gradients of the training samples.
To ensure efficient execution, we aim to compute and cache this sample-independent state once, and reuse it across multiple audits. 
While this is relatively straightforward in a plaintext setting, the secure setting poses challenges.
Specifically, the auditing functionality \sabbAAuditID is stateless and does not allow persistent storage between audit invocations, making such caching infeasible by default.

Securely enabling reuse of intermediate computations is non-trivial. 
A natural idea is to extend the training phase with the sample-independent preprocessing, extending the training state $R$. 
However, this introduces two key limitations.
First, it reduces the flexibility of our framework. Incorporating all possible preprocessing steps during training would require precomputing and retaining intermediate values for every supported auditing function, even if many are never used.
Second, it introduces privacy concerns. While releasing the trained model can be acceptable under the right circumstances and mitigations, storing or exposing per-sample gradients can severely compromise data privacy, as these gradients can contain sensitive information about individual training points. 
Although retaining such data in secret-shared form is theoretically possible, maintaining a large shared state across phases incurs significant overhead.

\myparagraph{Protocols}
To allow preprocessing securely, we split the auditing functionality \sabbAAuditID into two phases.
We model each phase of both auditing functions with a distinct functionality: a preprocessing functionality \sabbAAuditPreprocessID, which is sample-independent and is executed once using the training data $\spartyinputs$ and model $\smodel$, and an online audit functionality \sabbAAuditOnlineID, which incorporates, in addition to the training data and the model, a specific test-time query $(\tilde{x},\tilde{y})$. 
With the proper extensions to our auditing framework, we can then cache the output of \sabbAAuditPreprocessID explicitly.
We discuss these extensions in~\Cref{appx:arc_with_prep} and first present the protocols that compute the auditing functions.

The preprocessing and online ideal functionalities for \sysgradient and \syscamel take as input the secret shares of the training data and the model, open the shares, and output the preprocessing and online phases, following their descriptions in~\Cref{alg:traceback,alg:unlearning}.
\ifdefined\isnotextended
Due to space constraints, we defer the definitions of the ideal functionalities to the extended version\citeextended.
\else
We provide their formal definitions in Appendix~\ref{sec:appx:mpc:idealfun}.
\fi
Note that we define our auditing functionalities as non-reactive.
This design choice allows us to focus on securely realizing the core computation within a standalone protocol, which simplifies their exposition and security analysis. 
Correspondingly, we have two separate protocols for the preprocessing and online phases that instantiate each of these functionalities.
For \sysgradient, Protocol~\ref{protocol:ours:preprocessing} computes the gradients of the training data at the set of training checkpoints for preprocessing.
In the online phase, Protocol~\ref{protocol:ours:online} computes the gradient of the prediction sample and computes the (optimized) GAS score for the prediction sample and all training samples, which is returned as influence score.
For \syscamel, the preprocessing phase (cf.~Protocol~\ref{protocol:camel:preprocessing}) computes the unlearned models based on each user's dataset using \idealSGD, and outputs the resulting models as secret shares.
The online phase (cf.~Protocol~\ref{protocol:camel:online}), using the unlearned models as input, computes the loss of the prediction sample $\tilde{x}$ with respect to the misclassified label $\tilde{y}$.
These per-user losses are then used as influence scores to compute user responsibility.
Finally, in both functions, \textsc{RankUsers} computes the relevant ranking and classification metrics based on the influence scores.
\ifdefined\isnotextended
We provide formal security proofs for our protocols in Appendix~\ref*{appx:proofs} in the extended version of this work\citeextended.
\else
We provide formal security proofs for our protocols in Appendix~\ref{appx:proofs}.
\fi

\subsection{\sysname Preprocessing}
\label{appx:arc_with_prep}
Although Arc ensures consistency between the training data, the models, and predictions to run auditing functions,
it does not securely support the preprocessing required to achieve efficient online times.
In addition, it requires the \glspl{r:inputparty} to remain available during auditing to input their training data to the MPC servers for each audit.
While it is reasonable to assume that the \glspl{r:modelowner} responsible for inference remain online, this assumption is less realistic for \glspl{r:inputparty}, which may not be persistently available or willing to remain online solely for audit purposes.
In this section, we show how we extend Arc in order to enable secure preprocessing for auditing functions and remove the dependency on data owners at auditing time.

\myparagraph{Auditing Phase}
The auditing phase in the Arc protocol is evaluated by a set of parties referred to as the \glspl{r:acomputer}, which compute the auditing function using the inputs provided by the~\gls{r:auditrequester}, the~\glspl{r:inputparty}, and the~\gls{r:modelowner}.
To support the evaluation of arbitrary auditing functions, Arc models the computation via an arithmetic black-box interface, denoted~\sabbid, where each audit receives a fresh instance of the interface.
This approach ensures that no state is retained after the audit is completed,
but conflicts with the potential need for preprocessing outputs to persist for future audits.
Preserving the state through re-using \sabbid between audits may seem like a simple solution.
However, in practice, this would require retaining a potentially large amount of state across multiple parties (e.g., many different clients) throughout the protocol,
which introduces additional complexity that the original protocol aims to avoid.
Therefore, we deviate from the approach in Arc by explicitly defining a secret-sharing scheme to store the state of the preprocessing phase between auditing phases.
If the \glspl{r:acomputer} remain online, they can retain this state locally, thereby substantially reducing the end-to-end cost of subsequent audits. 
Importantly, the \glspl{r:inputparty} do not need to remain online or re-upload their datasets, as long as this shared preprocessing state is maintained.

Our design presents a tradeoff between the cost of long-term storage of secret shares and the requirement for a subset of parties to remain available during the entire phase in which predictions can be audited.
Using a secret-sharing scheme not only resolves the state retention issue but also provides flexibility in selecting specialized primitives for the long-term storage of secret shares, such as proactive or verifiable secret sharing~\cite{Schultz2010-mpss,Maram2019-ao}.
This decoupling allows us to adopt storage mechanisms distinct from the realization of \sabb.
To achieve this, we extend the definition of \sabb to include two new functions: \sabbShare, which generates and reveals secret shares of internal values, and \sabbOpen, which inputs and opens these shares.
\ifdefined\isnotextended
We provide a formal definition of a secret-sharing scheme in~\Cref*{appx:definitions} in the extended version of this work\citeextended.
\else
We provide a formal definition of a secret-sharing scheme in~\Cref{def:lsss} in~\Cref{appx:definitions}.
\fi
By the definition of \sabb, these functions can be trivially realized using the existing operations defined in \sabb.
Nevertheless, there are often more efficient methods to implement these functions.
For instance, if \sabb is instantiated using the same secret-sharing scheme, parties can directly store the shares of internal values and reveal them when required.
Alternatively, if a different secret-sharing scheme is employed, a share conversion protocol can efficiently transform the shares between schemes.

\begin{figure*}
    \begin{myalgorithm}[Arc with Preprocessing $\sprotocolArcWithPreprocessing$]{protocol:camel}
        This protocol is based on \sprotocolArc which is defined by Figure 3 in~\cite{lycklama_holding_2024}).\\
            \textbf{Training:} \textit{The training phase is identical to that of \sprotocolArc.}\\
            \textbf{Inference:} \textit{The inference phase is identical to that of \sprotocolArc.}\\
            \textbf{Auditing:} The protocol proceeds on a new instance of $\sabbid$ between computing parties \sPartyAuditComputer, a \gls{r:auditrequester} $\sPartyAuditor_j$ and the \gls{r:modelowner} $\sPartyModelHolder_k$:\\
            \textit{Steps \textsf{A.1}-\textsf{A.4} are identical to those in \sprotocolArc.}
            \begin{protocol_steps}[label=\textsf{A}.\arabic*,start=5]
            \item If secret shares of the preprocessing cache \scache exist in memory, load them as $\scacheAbb \gets \sabbOpen(\scache)$. Otherwise, each data owner $\texttt{DH}_i$ inputs their data $D_i$ to \sabbid, and:
            \begin{itemize}
                \item The protocol performs the input verification steps for each $\texttt{DH}_i$ as described in \textsf{A}.5 in the original protocol.
                \item The \glspl{r:acomputer} compute $\scacheAbb \leftarrow \sabbAAuditPreprocessID(\sauditfunction, \ssAbb{\strainset_1}, \ldots, \ssAbb{\strainset_\sNumParties}, \ssAbb{\smodel}, \aux)$,
                generate a sharing $\scache \gets \sabbShare(\scacheAbb)$ and store \scache in memory.
            \end{itemize}
            \item If secret shares of the preprocessing cache \scache exist in memory, load them as $\scacheAbb \gets \sabbOpen(\scache)$.
            Otherwise, 
                \item The \glspl{r:acomputer} compute $\ssAbb{o} \leftarrow \sabbAAuditOnlineID(\sauditfunction, \ssAbb{\smodel}, \ssAbb{\spredictionX}, \ssAbb{\spredictionY}, \aux, \scacheAbb)$
                and use \sabbid to open $o$ at $\sPartyAuditor_j$. 
            \end{protocol_steps}
        \end{myalgorithm}
\end{figure*}

We present the protocol for Arc with preprocessing in Algorithm~\ref{protocol:camel}.
Note that we represent the internal values of \sabb using single-bracket notation, e.g., $\ssAbb{s}$, to distinguish them from the secret shares in double-bracket notation used throughout this paper, $\secs{s}$.
The steps related to the validation of the model and the prediction in the protocol are the same as in the original auditing phase of Arc.
After the~\glspl{r:inputparty} input their (valid) training datasets, the \glspl{r:acomputer} compute
the \glspl{r:acomputer} compute the preprocessing part of the auditing algorithm using \sabbAAuditPreprocessID,
and save the output $\scache$ as secret shares using $\scache \gets \sabbShare(\scacheAbb)$ if \scache is not stored in memory.
However, if \scache is already in memory, the protocol does not expect any input from the~\glspl{r:inputparty}, and can proceed without them by calling \sabbOpen(\scache) to load the preprocessing output into \sabb.
Afterwards, the \glspl{r:acomputer} evaluate the online phase of the protocol using $\sabbAAuditOnlineID$, which takes in
the prediction, the datasets, the model, the auxiliary input, and the output of the preprocessing \scacheAbb.

This approach still achieves the original auditing functionality as defined in Arc, but provides increased flexibility by allowing shared state to be maintained for as long as it is possible to keep the~\gls{r:acomputer} online.
At the same time, it allows to possibility to decouple the \glspl{r:inputparty} from the online inference and auditing process.
Altogether, this design maintains the privacy and correctness guarantees of Arc, while enhancing its practicality and scalability in real-world deployment scenarios.

\subsection{Optimizations}
Realizing the traceback algorithm within the PPML setting presents significant challenges related to scalability and numerical precision, primarily due to the representation of values and the computational overhead of cryptographic operations.
While the main overhead of \syscamel is a result of generic protocols such as training and inference for which optimized protocols already exist, we focus here on two optimizations specific to the computation of \sysgradient. 
\ifdefined\isnotextended
We discuss additional optimizations in Appendix~\ref*{appx:mpc_opt} in the extended version of this work\citeextended.
\else
We discuss additional optimizations in Appendix~\ref{appx:mpc_opt}.
\fi

\myparagraph{Heuristic Sample Selection}
The primary performance bottleneck in the traceback algorithm lies in evaluating the cosine similarity between training and test gradients.
    This bottleneck arises from the reliance on fixed-point division and square root operations, both of which are significantly more expensive in secure computation settings.
    Integer division and square root protocols such as those based on Goldschmidt and Raphson-Newton iterations~\cite{Aly2019-zz} require significantly more communication rounds compared to other operations, increasing both latency and bandwidth usage.
    
Although we can compute these operations using efficient approximations for the reciprocal square root operation, the protocol still requires evaluating $\cO(\abs{\cT} \cdot \abs{D_i})$ reciprocals per partition for each traceback computation.
To address this, we propose an optimization aimed at improving the computational complexity of selecting the $\topk{k}$ cosine similarity scores.
The core observation is that the $\topk{k}$ selection ultimately depends on only $\topk{k}$ terms from the entire set of training samples.
If we can approximately identify these $k$ terms without fully evaluating the cosine similarity across all samples in a partition, we can significantly reduce the number of expensive operations required.

We propose a $\topkl{k}{l}$ selection scheme, in which a coarser set $\tilde{I}_i$ of $l$ indices are determined via a heuristic function $h(z, \hat{z})$, where $z$ is a training sample $z$ and $\hat{z}$ the test sample. The elements $\left\{ D_i^{(j)} : \text{j} \in \tilde{I}_i \right\}$ are then used instead of $D_i$ for the subsequent scoring function. To compute $\tilde{I}_i$, we adapt the existing $\topk{l}$ procedure as $\topki{k}$ to return indices instead of scores directly:
\begin{equation*}
    \tilde{I}_i = \topk{l}\left(\left\{ (j, h(D_i^{(j)}, \hat{z})) : j \in 1,\ldots,\abs{I_i} \right\}\right).
\end{equation*}
We implement $h$ as the ranking induced by the un-normalized gradient scores \cite{pruthi_estimating_2020}:
\begin{equation}
    h(z, \hat{z}) = \sum_{t \in \cT} \eta_t \iprod{\nabla_\theta \ell(\theta_t; z)} {\nabla_\theta \ell(\theta_t; \hat{z})}.
\end{equation}
Our choice of $h$ is significantly faster than directly computing cosine scores for each sample because multiplications require fewer rounds of communication than divisions. Although this algorithm introduces an extra selection of the $\topk{l}$ samples in addition to $\topk{k}$, it is more efficient when $l \ll \abs{D_i}$.
The full algorithm and protocol are presented in Appendix~\ref{appx:mpc_opt}\citeextended.

\section{Experimental Results}
We test our user-level auditing functions by training classifiers under a total of ten different poisoning attacks and four representative datasets from three modalities: vision, malware, and text.
\ifdefined\isnotextended
Due to space considerations, we discuss detailed results for four attacks, and refer to \Cref{appx:additional-experiments} in the extended version for comprehensive results \cite{Rose2024-UTraceFull}.
\else
\fi

\subsection{Evaluation setup}
We run \sysname on Ubuntu 24.04 machines with  128GB of memory and different GPU configurations including Nvidia RTX A6000, Nvidia TITAN X (Pascal), and Nvidia GeForce GTX TITAN X.

\myparagraph{Datasets and models}
We summarize the datasets and models we use in our experiments. Detailed training configurations, including hyperparameters, can be found in \Cref{appx:exp-details}.

\myparagraph{CIFAR-10}
CIFAR-10 \cite{krizhevsky_cifar-10_nodate} is a 10-class image dataset of 32x32 RGB images. We train image classifiers using ResNet-18 models \cite{he_deep_2015}. CIFAR-10 is well-established as a poisoning benchmark and a wide suite of attacks work against it \cite{aghakhani_bullseye_2021, geiping_witches_2021, souri_sleeper_2021, jagielski_subpopulation_2021, schwarzschild_just_2021}.

\myparagraph{Fashion}
Fashion MNIST~\cite{xiao_fashion-mnist_2017} is a 10-class image dataset of 28x28 grayscale images. We train image classifiers using a small convolutional network.

\myparagraph{EMBER}
EMBER~\cite{anderson_ember_2018} is a malware classification dataset of 2351-dimensional examples with features extracted from Windows Portable Executable (PE) files. We train using the EmberNN architecture from Severi et al. \cite{severi_explanation-guided_2021}.

\myparagraph{SST-2}
SST-2 \cite{socher_recursive_2013} is a sentiment classification text dataset containing film reviews. SST-2 has been used in prior work \cite{hammoudeh_identifying_2022} as a benchmark for sample-level poisoning traceback in the text domain.

\myparagraph{Partitioning Data} We use 10 data owners that contribute their datasets to model training, but we also experiment with 20 owners, leading to similar results. To partition the datasets, we sample class labels from a Dirichlet distribution~\cite{hsu_measuring_2019}  $q \sim \text{Dir}(\alpha p)$, where $p$ is the prior label distribution. We set $\alpha = 100$ for all of our experiments, which simulates a scenario with almost uniform distribution across owners. We vary the number of poisoned owners between 1 and 4, and distribute the poisoned samples randomly across the malicious owners. 

\myparagraph{Metrics}
We measure both the effectiveness and efficiency of \sysname.
\sysname supports two output configurations, producing either a ranking on all users or a list of accused users. 
We thus leverage two types of metrics for evaluating the effectiveness our system: ranking metrics and classification metrics. First, we model traceback as an information retrieval task and report the mean average precision (mAP) and mean reciprocal rank (mRR) scores for each attack scenario. These metrics measure if the malicious owners are ranked high among all owners according to their responsibility scores.  In the second configuration, we need to determine a threshold on responsibility scores to generate a list of accused users. We  train benign models and compute standardized influence scores for benign misclassifications events under different random samplings from the clean data distribution. Thresholds are chosen to correspond to a low false positive rate (1\%) on these benign misclassifications. We report both the true positive rate (TPR) and false positive rate (FPR) for each attack under this threshold, as well as the global AUC metric.

We evaluate the efficiency of our system using three metrics: round complexity, communication complexity, and overall runtime. Round complexity measures the number of communication rounds needed in the MPC protocol. Communication complexity measures the total data exchanged. Overall runtime captures the total execution time, including training and traceback, in both LAN and WAN settings.

\myparagraph{Traceback parameters} We report influence metric hyperparameters for unlearning and \sysname in \Cref{table:traceback-hyperparameters} in \Cref{appx:training-config}.
In general, we sample checkpoints at some uniform interval until reaching a maximum checkpoint set size. We choose unlearning parameters using reasonable fine-tuning settings for each learning task.

\subsection{Dirty-label Poisoning Attacks on Vision}\label{sec:dirty-label}
We begin our evaluation with dirty-label poisoning attacks against CIFAR-10 and Fashion datasets.
\ifdefined\isnotextended
For space, we present detailed results only for CIFAR-10 here, and defer detailed Fashion results to \Cref{appx:fashion} in the extended version \cite{Rose2024-UTraceFull}.
\else
\fi
Our qualitative findings are similar between the two datasets.

We evaluate our algorithms against four dirty-label attacks. The first two attacks are well-known attacks from the literature:
the BadNets backdoor attack \cite{gu_badnets_2019} and the Subpopulation attack \cite{jagielski_manipulating_2021}. Both of these attacks use a label-flipping poisoning strategy, injecting poisoned samples labeled as the target label directly into the training set.
The final two attacks are adaptive attacks, designed specially to subvert the traceback mechanism. The first simply modifies the BadNets objective from a one-to-one to a many-to-many objective via a ``permutation'' backdoor: the attack chooses a random derangement $\sigma : C \to C$ and aims to misclassify samples from each class $i$ as class $\sigma(i) \ne i$. By tying the attack success to the base learning task, the learning dynamic for attack samples may more closely resemble that of benign data, obscuring the presence of the attack during training. The second strategy is to add substantial label noise to the poisoned samples, the idea being that training data attribution may be harder if the model exhibits less confident predictions. We refer to these attack variants as ``$\sigma$-BadNets'' and ``Noisy BadNets'', respectively.
Complete details on attack setup, including poison generation rules and poisoning rates, can be found in \Cref{appx:attack_config}.

\myparagraph{Results}
We report the results for attacks against CIFAR-10 / ResNet-18 in \Cref{table:cifar10-dirty} and attack-conditioned ROC curves in \Cref{fig:cifar-dirt-rocs}. Overall, our results show that \sysname can reliably identify poisoned users with high precision, even when the corrupted parties strategize by distributing poisons across several user datasets.

For BadNets, all methods perform strongly even as the number of poisoned user datasets increases. The Unlearning method maintains above $> 99\%$ TPR. One noteworthy behavior is that the FPR for Unlearning is quite high, reaching nearly $7\%$, despite the threshold having been determined based on a 1\% FPR on benign misclassifications. One possible consequence might be that thresholds for the Unlearning method are more dependent on the details of the attack.
For Subpopulation attacks, the gradient-based method achieves the highest TPR when multiple user datasets are poisoned, maintaining $\ge 60\%$ TPR versus $< 13\%$ for unlearning and $< 4\%$ for kNN. Interestingly, even though kNN achieves higher ranking metrics than unlearning, the scores are not decisive enough under thresholds to outperform unlearning under classification metrics.

Results for both adaptive attacks are qualitatively similar. Both methods fail to evade the traceback mechanism for gradient-based and kNN scores, with detection quality matching that of the standard BadNets case. Performance under the Unlearning method suffers for both adaptive attacks in both ranking and classification metrics, with mAP falling to as low as 0.70 and TPR falling to as low as $11\%$.

\begin{figure}[htb]
    \centering
    \begin{subfigure}[b]{1.0\linewidth}
        \centering
        \includegraphics[width=1.0\linewidth]{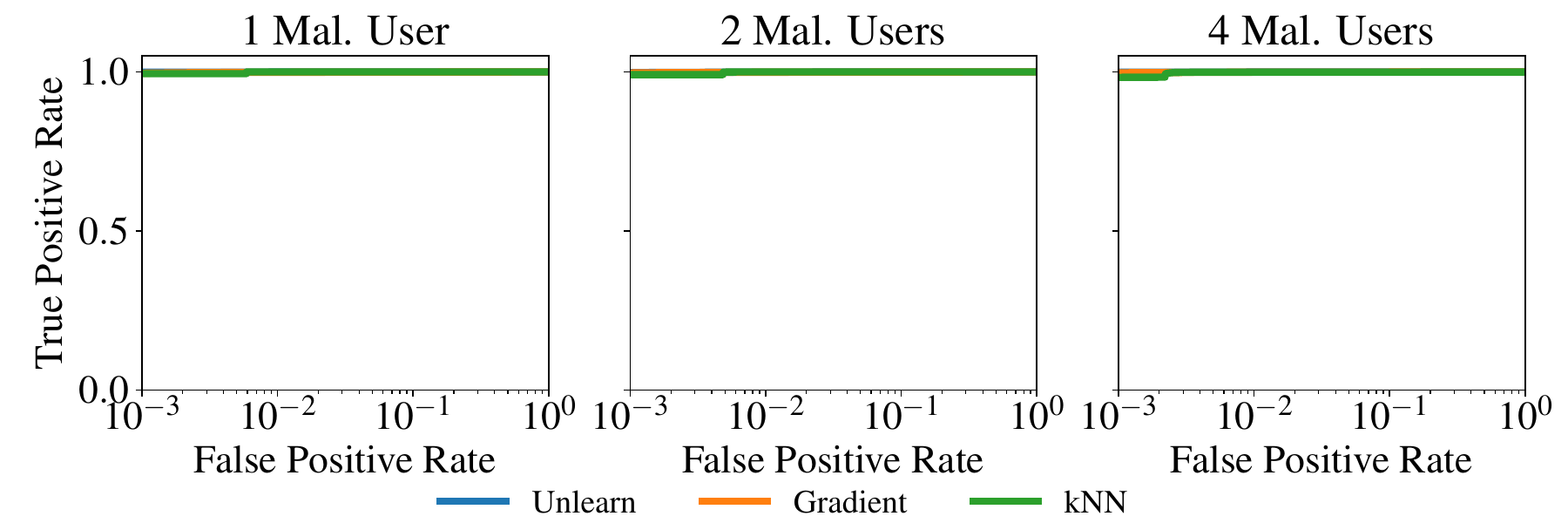}
        \caption{BadNets}
        \label{fig:cifar-badnets-roc}
    \end{subfigure}
    
    \begin{subfigure}[b]{1.0\linewidth}
        \centering
        \includegraphics[width=1.0\linewidth]{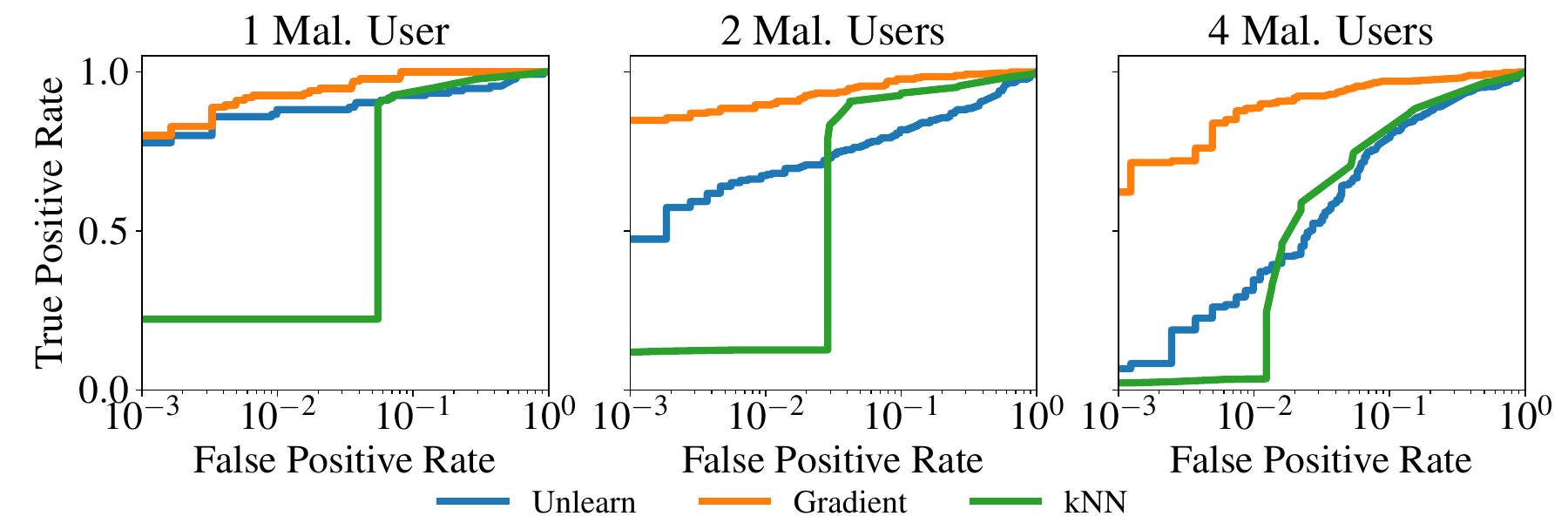}
        \caption{Subpopulation}
        \label{fig:cifar-subpop-roc}
    \end{subfigure}

    \begin{subfigure}[b]{1.0\linewidth}
        \centering
        \includegraphics[width=1.0\linewidth]{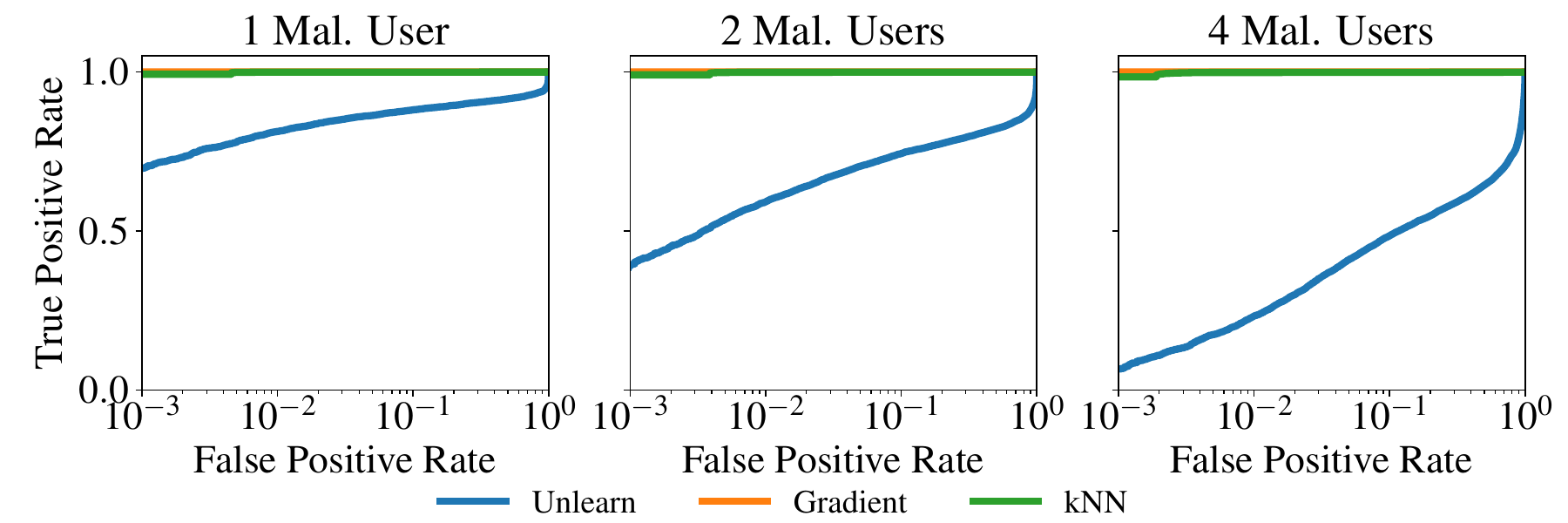}
        \caption{Noisy BadNets}
        \label{fig:cifar-badnets-noise-roc}
    \end{subfigure}
    
    \begin{subfigure}[b]{1.0\linewidth}
        \centering
        \includegraphics[width=1.0\linewidth]{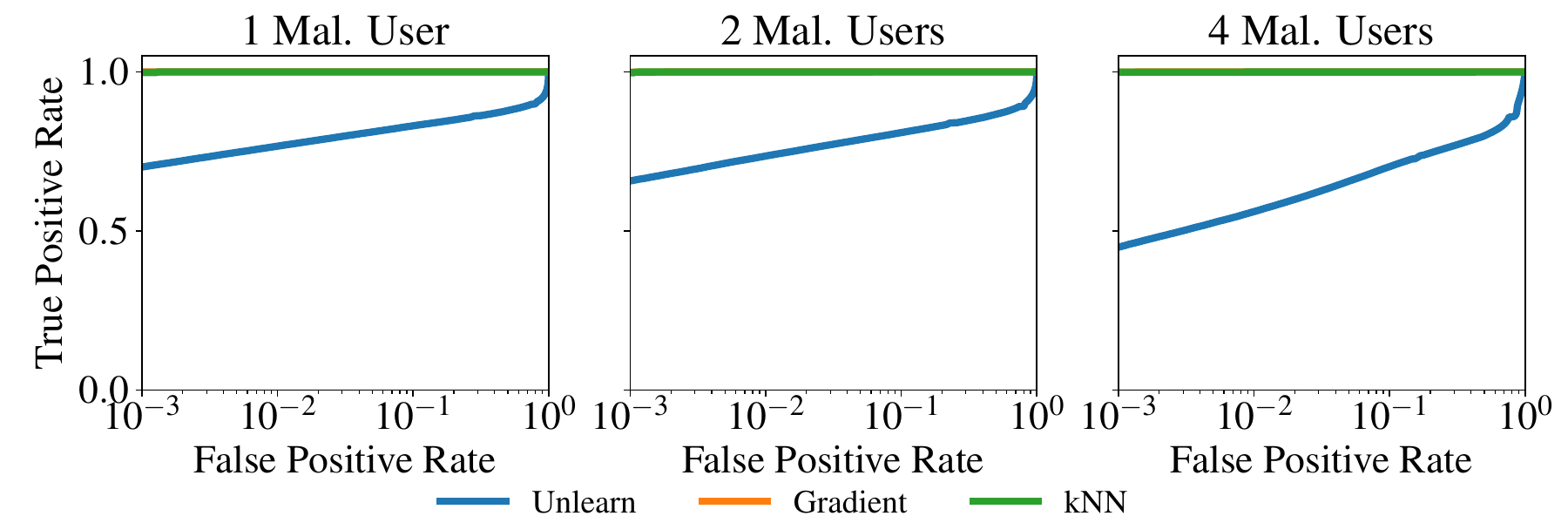}
        \caption{$\sigma$-BadNets}
        \label{fig:cifar-badnets-permute-roc}
    \end{subfigure}
    
    \caption{Malicious user identification ROC curves for dirty-label attacks against CIFAR-10 on ResNet18.}
    \label{fig:cifar-dirt-rocs}
\end{figure}

\begin{table*}[!htb]
\centering
\setlength{\tabcolsep}{3.1pt}
\caption{Results for dirty-label attacks against CIFAR-10 on ResNet18.}
\label{table:cifar10-dirty}
\begin{tabular}{@{}ccSSSSS!{\quad}SSSSS!{\quad}SSSSS@{}}
\toprule
\multicolumn{1}{c}{\multirow{2}{*}{\textbf{Attack}}}
& \multicolumn{1}{c}{\multirow{2}{*}{\shortstack[c]{\textbf{Mal.}\\\textbf{Users}}}}
& \multicolumn{5}{c}{\textbf{kNN (Baseline)}} & \multicolumn{5}{c}{\textbf{Unlearning (\S\ref{sec:unlearning})}}
& \multicolumn{5}{c}{\textbf{Gradient (\S\ref{sec:gradient})}} \\
\cmidrule(lr{1.2em}){3-7} \cmidrule(lr{1.2em}){8-12} \cmidrule(l){13-17}
\multicolumn{1}{c}{}
& \multicolumn{1}{c}{}
& mAP & mRR & {\scriptsize TPR (\%)} & {\scriptsize FPR (\%)} & AUC
& mAP & mRR & {\scriptsize TPR (\%)} & {\scriptsize FPR (\%)} & AUC
& mAP & mRR & {\scriptsize TPR (\%)} & {\scriptsize FPR (\%)} & AUC \\
\hline
\multirow{4}{*}{BadNets}
    & 1 & 0.9999 & 0.9999 & 99.450 & 0.014 & 1.000 & 1.0000 & 1.0000 & 99.990 & 6.989 & 1.000 & 0.9995 & 0.9995 & 99.847 & 0.502 & 1.000 \\
    & 2 & 0.9999 & 0.9999 & 99.155 & 0.009 & 1.000 & 0.9999 & 0.9999 & 99.967 & 2.076 & 1.000 & 0.9998 & 0.9998 & 99.770 & 0.178 & 1.000 \\
    & 3 & 0.9997 & 1.0000 & 98.772 & 0.006 & 1.000 & 0.9999 & 0.9999 & 99.963 & 1.048 & 1.000 & 0.9998 & 0.9999 & 99.442 & 0.044 & 1.000 \\
    & 4 & 0.9997 & 0.9999 & 97.672 & 0.002 & 0.999 & 0.9999 & 0.9999 & 99.874 & 0.019 & 1.000 & 0.9998 & 0.9999 & 98.230 & 0.003 & 1.000 \\
\hline
\multirow{4}{*}{Subpop}
    & 1 & 0.9644 & 0.9644 & 23.490 & 1.119 & 0.940 & 0.9396 & 0.9396 & 87.248 & 0.671 & 0.967 & 1.0000 & 1.0000 & 85.235 & 0.298 & 0.997 \\
    & 2 & 0.9567 & 0.9691 & 13.087 & 0.587 & 0.947 & 0.8777 & 0.9060 & 57.047 & 0.168 & 0.916 & 0.9887 & 0.9966 & 80.537 & 0.084 & 0.991 \\
    & 3 & 0.9473 & 0.9801 & 6.488 & 0.384 & 0.937 & 0.8973 & 0.9497 & 30.425 & 0.096 & 0.929 & 0.9820 & 1.0000 & 72.260 & 0.000 & 0.983 \\
    & 4 & 0.9449 & 0.9717 & 3.691 & 0.447 & 0.921 & 0.8994 & 0.9463 & 12.416 & 0.224 & 0.905 & 0.9798 & 0.9899 & 60.570 & 0.000 & 0.985 \\
\hline
\multirow{4}{*}{\shortstack[c]{Noisy\\BadNets}}
    & 1 & 0.9994 & 0.9994 & 99.213 & 0.013 & 1.000 & 0.8451 & 0.8451 & 71.411 & 0.529 & 0.857 & 0.9998 & 0.9998 & 99.976 & 0.675 & 1.000 \\
    & 2 & 0.9990 & 0.9992 & 99.044 & 0.012 & 0.999 & 0.7825 & 0.8056 & 46.914 & 0.295 & 0.791 & 0.9999 & 0.9999 & 99.964 & 0.324 & 1.000 \\
    & 3 & 0.9989 & 0.9994 & 98.778 & 0.012 & 0.999 & 0.7273 & 0.7523 & 27.713 & 0.218 & 0.702 & 0.9999 & 1.0000 & 99.960 & 0.104 & 1.000 \\
    & 4 & 0.9989 & 0.9995 & 97.948 & 0.004 & 0.999 & 0.7028 & 0.7153 & 11.042 & 0.210 & 0.640 & 0.9999 & 1.0000 & 99.930 & 0.002 & 1.000 \\
\hline
\multirow{4}{*}{$\sigma$-BadNets}
    & 1 & 0.9999 & 0.9999 & 99.838 & 0.002 & 1.000 & 0.8828 & 0.8828 & 81.272 & 1.641 & 0.900 & 1.0000 & 1.0000 & 99.996 & 0.627 & 1.000 \\
    & 2 & 0.9998 & 0.9999 & 99.794 & 0.001 & 1.000 & 0.8644 & 0.8635 & 75.206 & 0.989 & 0.877 & 1.0000 & 1.0000 & 99.993 & 0.260 & 1.000 \\
    & 3 & 0.9998 & 0.9999 & 99.718 & 0.001 & 1.000 & 0.8504 & 0.8435 & 68.167 & 0.621 & 0.853 & 1.0000 & 1.0000 & 99.986 & 0.061 & 1.000 \\
    & 4 & 0.9998 & 0.9999 & 98.992 & 0.000 & 1.000 & 0.8252 & 0.8123 & 56.563 & 0.825 & 0.809 & 1.0000 & 1.0000 & 99.952 & 0.006 & 1.000 \\
\bottomrule
\end{tabular}
\end{table*}

\myparagraph{Parameter ablations}
\blue{
We vary the value of $k$ for the top-$k$ metric for different attacks. We find that setting $k$ too small ($k< 16$ for CIFAR-10) weakens performance as there is not enough signal to distinguish malicious samples. On the other hand, if $k$ is too large, the benign samples influence the score more than poisoned samples. The value of $k=32$ provides good results for several attacks.
We run \sysname with 20 owners to evaluate the impact of the number of users. For BadNets, Unlearning maintains 100\% TPR while the gradient-based score drops to 60\% with 8 poisoned users. On Subpopulation attacks, the trend is reversed, and the gradient-based score achieves $\ge 70\%$ TPR with 8 poisoned owners while Unlearning degrades to a TPR of $1\%$. Our findings suggest that our traceback methods can complement each other: when one method fails, another might still succeed.
\ifdefined\isnotextended
Full ablation results can be found in \Cref{appx:ablations} in the extended version \cite{Rose2024-UTraceFull}.
\else
Full results are in \Cref{appx:ablations}.
\fi
}

\begin{figure}[t]
        \centering
        \includegraphics[width=1.0\columnwidth]{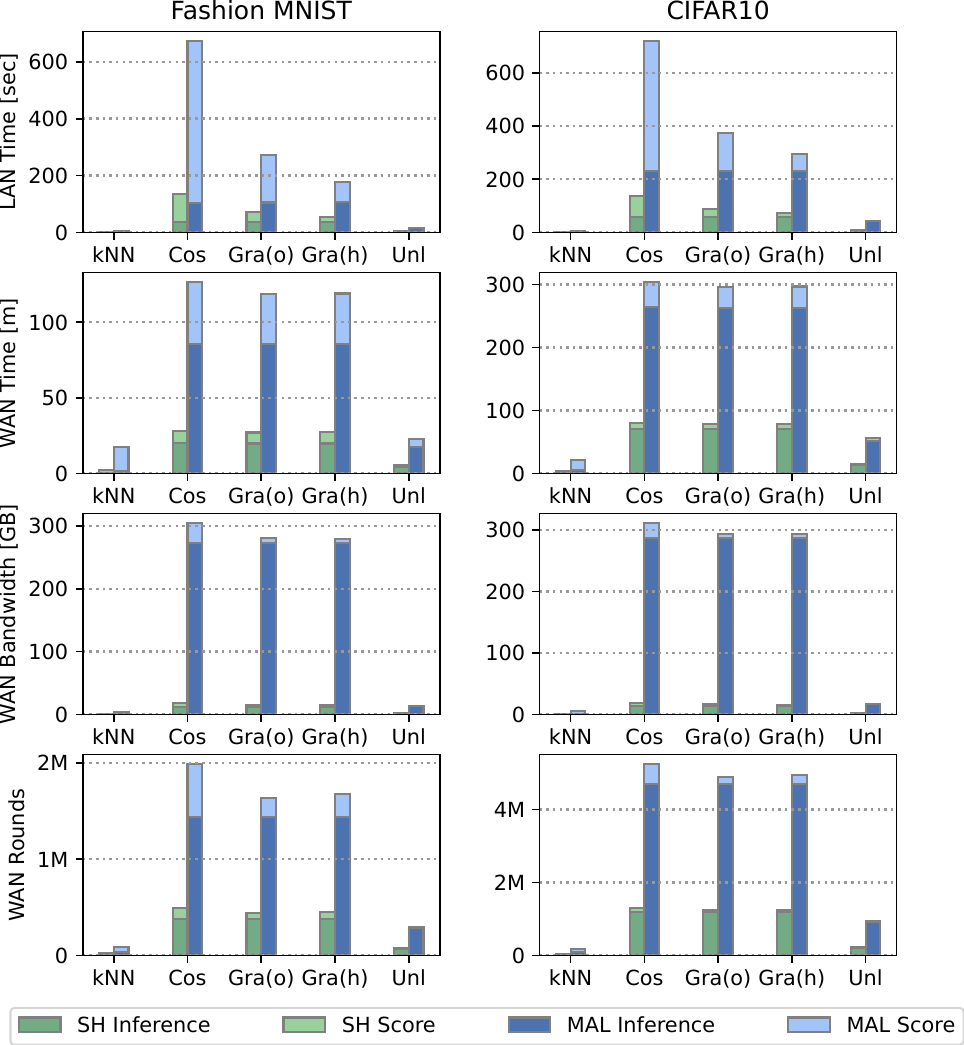}
    \caption{
    Overhead of our user-level auditing functions \sysgradient and \syscamel.
    \texttt{kNN} refers to the kNN baseline,
    \texttt{Cos} to the unoptimized GAS score.
    Further, \texttt{Grad(o)} refers to the version of \sysgradient with the reciprocal square root optimization,
    and \texttt{Grad(h)} refers to \sysgradient with the heuristic sample selection.
    Finally \texttt{Unl} represents the results of \syscamel.}
    \label{fig:eval:mpc}
\end{figure}

\subsection{Text and Malware Experiments}
We present the results for experiments against the SST-2 and Ember datasets in \Cref{appx:modality} and summarize the findings here. First, we find that the inner-product heuristic does not select relevant poisoning samples in the text modality. Hence, we use full $\topk{k}$ selection for attacks against SST-2. We experiment with four types of backdoor attack against SST-2 and two backdoor attacks against Ember.
In general, we find that the gradient-based score is more consistent in flagging malicious users at low FPR in experiments from other modalities. For example, under a BadNets-style text backdoor against SST-2, the gradient score achieves a TPR of $93\%$, versus the unlearning TPR of $91\%$ and kNN TPR of $89\%$. A qualitatively similar result follows for an explanation-guided backdoor attack against Ember, where the gradient score achieves a TPR of $100\%$ versus the unlearning TPR of $72\%$ and kNN's $98\%$.

\red{
}

\subsection{Clean-label Poisoning Attacks on Vision}
We evaluate our method against clean-label poisoning attacks. Clean-label attacks require the poisoned samples to be assigned the correct label and generally require that the poisoned samples are close in distance to some existing clean training sample. Due to these additional constraints, clean-label attacks can be harder for anomaly-based defenses to detect.
We evaluate against two clean-label poisoning attacks, Witches' Brew \cite{geiping_witches_2021} and Sleeper Agent \cite{souri_sleeper_2021}. Witches' Brew is a targeted poisoning attack that imposes small adversarial perturbations on top of clean training instances in order to optimize a gradient mimicry objective. Since Witches' Brew is a targeted attack, there is only a single target sample. Sleeper Agent is a refinement of Witches' Brew with a backdoor attack objective.

\ifdefined\isnotextended
We present the full results for clean-label attacks in \Cref{appx:clean-label} in the extended version \cite{Rose2024-UTraceFull} and summarize the findings here.
\else
We present the full results for clean-label attacks in \Cref{appx:clean-label} and summarize the findings here.
\fi
\blue{
Clean-label poisoning attacks prove more challenging for all traceback methods. For Witches' Brew, the gradient-based score outperforms unlearning across all metrics, achieving mRR scores of $>0.99$ in all settings, meaning the top-ranked user is nearly always malicious, and achieving $>89\%$ TPR with 1 poisoned user, compared with $<15\%$ for kNN and $<86\%$ for Unlearning. However, all methods achieve $<6\%$ TPR when the attacker distributes poisons across 4 poisoned users, indicating a substantial decrease in effectiveness. Sleeper Agent attacks show a little more promise, gradient-based scores achieving $>36\%$ and Unlearning $> 55\%$ TPR, respectively, with 4 poisoned users, in constrast with kNN's $<20\%$ TPR.
}

\subsection{Secure \sysname}
We evaluate the performance of \sysname in MPC. Our implementation is based on MP-SPDZ~\cite{keller2020mp}, a popular framework for MPC that supports a variety of protocols. We compare \sysgradient and \syscamel with a kNN-based baseline.

\myparagraph{Experimental Setup}
We run \sysname on AWS \texttt{c5.9xlarge} machines running Ubuntu 20.04, each equipped with 36 vCPUs of an Intel Xeon 3.6~Ghz processor and 72~GB of RAM. The machines are connected over a local area network (LAN) through a 12~Gbps network interface with an average round-trip time (RTT) of $0.5$~ms. We additionally perform our experiments in a simulated wide area network (WAN) setting using the \texttt{tc} utility to introduce an RTT of $80$~ms and limit the bandwidth to $2$~Gbps. We report the total wall clock time and the communication cost based on the data sent by each party, which includes the time and bandwidth required for setting up any cryptographic material for the MPC protocol, such as correlated randomness. 
We split the cost into preprocessing and online phases, corresponding to the protocols in~\Cref{sec:framework}.

    \myparagraph{Precision}
    We conducted experiments on the BadNets~\cite{gu_badnets_2019} attack to identify a fixed-point representation that aligns the cosine similarities and party attribution results of the MPC protocol with those of the floating-point baseline.
    The choice of fractional $(f)$ and integer ($i$) bit lengths is guided by several constraints: the division protocol requires
    $k \approx 2f$ and $k$ must be strictly less than half the domain size for efficiency~\cite{catrina2010secure}. Additionally, a domain size aligned with 64-bit word sizes is optimal for hardware efficiency.
    While a representation of $f=16$ and $i=31$ is generally regarded as sufficient for ML inference and training tasks~\cite{Keller2022-quantizedtraining}, it results in diverging cosine similarities and correct party attribution in only 30\% of our experiments.
    This issue arises because the attack causes poisoned samples to produce extremely small gradient magnitudes, which cannot be accurately represented in the cosine computation with lower precision.
    Increasing the precision to $f=32$ and $i=62$ resolves these issues, enabling consistent party attribution across all experiments despite minor variations in cosine similarity.
    We use this configuration for the performance evaluation of our protocols.

\myparagraph{Preprocessing} 
We anlyze the preprocessing costs of the two auditing mechanisms, \sysgradient and \syscamel. The gradient similarity approach employed by \sysgradient involves computing and compressing the gradients of the training dataset $D$ across all model checkpoints in $\cT$, resulting in an overall cost of $O(\abs{D} \cdot \abs{\cT})$ gradient computations. \sysgradient only requires gradients of the model’s final layers, which eliminates the need for full backward passes and substantially reduces computational overhead. In contrast, \syscamel relies on approximate unlearning by training $m$ leave-one-out models, which incurs $O(\abs{D} \cdot m \cdot e)$ forward and backward passes, where $e$ denotes the number of unlearning epochs. This makes the preprocessing phase of \syscamel more computationally intensive than for \sysgradient.
The dominant cost driver in \sysgradient is the number of model checkpoints, whereas \syscamel’s cost scales with the number of parties $m$. For CIFAR-10 in the LAN setting, \syscamel requires the equivalent of an additional full training run, resulting in 696 hours and 2734 hours of preprocessing time under semi-honest and malicious protocols, respectively. In comparison, \sysname reduces this overhead by approximately 2–3.5x, with corresponding preprocessing times of 374 hours and 1372 hours. Although both methods incur substantial preprocessing costs, performing this work offline is essential for enabling efficient online audits.

\myparagraph{Online}
We compare the wall-clock time and bandwidth costs of the online time of the kNN baseline, the optimized and heuristic versions of \sysgradient, as well as \syscamel in~\Cref{fig:eval:mpc}. For CIFAR-10, the overhead of Camel is 41 seconds for the malicious protocol in  LAN  and 73 minutes in  WAN, because it only requires $O(m)$ forward passes of the test sample, i.e., one inference per input party. \sysgradient, on the other hand, requires $O(\abs{\cT})$ forward passes in order to compute the gradients for each checkpoint, which take 229 seconds for the malicious protocol in  LAN  and 263 minutes in WAN. The key reason for \syscamel’s lower latency is that, in our setting, the number of users $m$ is substantially smaller than the number of training checkpoints $\abs{\cT}$.

In addition to the inference, computing cosine distances in \sysgradient introduces an additional overhead of $2.1-5.5$x (an extra $487-570$ seconds) in the malicious LAN setting. We reduce this overhead to $0.62-1.58$x in LAN  using the reciprocal square root optimization. \sysgradient with the heuristic optimization further reduces this overhead to $0.27-0.67$x of the forward pass cost. In absolute terms, the cost of doing a traceback is $72-89$ seconds ($27-78$ minutes) and $294-374$ seconds ($118-297$ minutes), respectively, for the semi-honest and malicious protocols in the LAN (WAN) setting for CIFAR-10. 
The online phase of \syscamel is significantly faster than \sysgradient, $15$ minutes for semi-honest and $57$ minutes for the malicious WAN setting, because its main overhead is computing the inferences of which it requires fewer.
For reference, a single inference takes 1.9 and 4.1 seconds in the malicious LAN setting, and 3.5 and 7 minutes in the WAN setting.
 While the overall online cost of \sysgradient and \syscamel is non-negligible, it remains practical in settings where only a subset of inferences is subject to audit.
Note that these numbers include the cost to compute input-independent cryptographic material of the underlying MPC protocols, and MPC-protocol-dependent preprocessing optimizations can be applied to reduce latency further.

\section{Conclusions, Limitations, and Future Work}\label{sec:discussion}

\blue{

Privacy-preserving machine learning (PPML) holds significant potential for enabling 
collaborative model training across sensitive and siloed datasets. However, by design, 
these systems conceal input data and training processes, creating opportunities for 
malicious actors to poison the model without detection. Ensuring accountability in such 
settings remains a pressing and underexplored challenge.

This paper introduces the task of \textit{user level data poisoning traceback} in secure 
collaborative learning as a measure to instill accountability at deployment time. Our 
proposed framework, \sysname, is the first to articulate and address this problem with a 
concrete and practical solution. \sysname combines two complementary detection mechanisms,
gradient similarity and approximate unlearning, to attribute model misbehavior to specific 
users, even when poisoning is subtle or distributed. We address key performance challenges 
in realizing \sysname in multi-party computation (MPC) settings, including efficient 
gradient storage, secure score computation, and auditing support for offline users.

Our experiments show that \sysname is robust across a diverse range of attacks, 
including backdoor, targeted, and subpopulation attacks, and operates effectively on 
multiple data modalities such as vision, text, and malware. It remains effective 
under low poisoning rates and coordinated attacks involving multiple adversaries. Through 
a series of system-level optimizations, we make \sysname practical for PPML settings by 
significantly reducing online computational cost and interaction overhead during the 
auditing phase.

Despite these contributions, several challenges remain. Clean-label poisoning attacks, 
where poisoned samples appear semantically legitimate, remain especially difficult to 
detect.
Moreover, while our responsibility scores are tailored for efficient execution within MPC protocols, 
the overall cost of auditing remains nontrivial at scale. A study on the learning dynamics 
of clean-label attacks to derive stronger responsibility
scoring methods would make interesting future work. In the
MPC vein, developing additional mechanisms to facilitate
low-cost, long-term deployments for the entire ML lifecycle
is an interesting challenge.

In the broader context, \sysname highlights the growing need for accountability mechanisms 
in privacy-preserving machine learning. While cryptographic protocols have made great 
strides in ensuring confidentiality, robustness and trustworthiness demand equal 
attention. We believe that frameworks like \sysname provide a foundation for bridging this 
gap and that continued research at the intersection of secure computation, machine 
learning, and adversarial robustness is critical for moving PPML from an experimental to a 
deployable state. }

\ifCLASSOPTIONcompsoc
  \section*{Acknowledgments}
\else
  \section*{Acknowledgment}
\fi

This research was sponsored by NSF awards CNS-2312875 and CNS-2331081, and the contract W911NF-24-2-0115, issued by US ARMY ACC-APG-RTP.

{\tiny
\bibliographystyle{IEEEtran}
\bibliography{references}
}

\begin{appendices}
\crefalias{section}{appendix}

\section{Additional Experiment Details}\label{appx:exp-details}

\begin{table*}[!htb]
    \centering
    \small
    \setlength{\tabcolsep}{2pt}
    \caption{Summary of Training Configurations. Reported parameters are consistent across all poisoning settings.}
    \begin{tabular}{@{}lcccccccc@{}}
        \toprule \\
        \textbf{Dataset} & \textbf{Model} & \textbf{Training} & \textbf{Epochs} & \textbf{Batch Size} & \textbf{Weight Decay} & \textbf{Learning Rate} &\textbf{Momentum} & \textbf{LR Drop Schedule} \\ \midrule
        CIFAR-10 & ResNet-18 & Scratch & 50 & 64 & $2 \times 10^{-4}$ & 0.1 & 0.9 & $0.32@10$ \\
        Fashion MNIST & ConvNet & Scratch & 50 & 64 & $2 \times 10^{-4}$ & 0.01 & 0.9 & $0.32@10$ \\
        SST-2 & $\roberta$ & Fine-tuning & 5 & 64 & 0 & $2 \times 10^{-5}$ & 0.9 & N/A \\
        EMBER & EmberNN & Scratch & 10 & 512 & $1 \times 10^{-6}$ & 0.1 & 0.9 & N/A
    \end{tabular}
    \label{table:training-config}
\end{table*}

\subsection{Detailed Training Configurations}\label{appx:training-config}
We report detailed training configurations in \Cref{table:training-config} and traceback configurations for each learning setting in \Cref{table:traceback-hyperparameters}.

\begin{table}[!htb]
\centering
\small
\caption{Traceback hyperparameters for different datasets. Unlearning is parameterized by unlearning epochs $E$ and learning rate $\eta$. Our method depends on parameters for the $\topkl{k}{l}$ metric, projection dimension $d$, number of gradient layers $L$, and the number of checkpoints $\abs{\cT}$.}
\label{table:traceback-hyperparameters}
\begin{tabular}{@{}cccccccc@{}}
\toprule
\multicolumn{1}{c}{\multirow{2}{*}{\textbf{Dataset}}} & \multicolumn{2}{c}{\textbf{Unlearning (\S\ref{sec:unlearning})}} & \multicolumn{5}{c}{\textbf{Gradient (\S\ref{sec:gradient})}} \\ \cmidrule(l){2-3} \cmidrule(l){4-8}
\multicolumn{1}{c}{} & \textbf{$\eta$} & \textbf{$E$} & $k$ & $l$ & $d$ & L & $\abs{\cT}$ \\
\hline
CIFAR-10 & $1 \times 10^{-4}$ & 5 & 32 & 512 & 64 & 1 & 50 \\
Fashion &$1 \times 10^{-4}$ &5 &32 & 512 & 8 & 1 & 50 \\
Ember &$1 \times 10^{-2}$ &1 &128 &4096 &64 &2 &20 \\
SST-2 &$5 \times 10^{-6}$ &1 &32 &- &64 &2 &25 \\
\bottomrule
\end{tabular}
\end{table}

\subsection{Attack Algorithms and Configurations}\label{appx:attack_config}

\myparagraph{Attacks against CIFAR-10}
For BadNets, each trial selects a source and target class at random. Poisons are generated by applying the backdoor trigger to training samples from the source class and changing their label to the target class. We use a poisoning rate of 2500 poisons (5\%).

For Subpopulation attacks, we generate each attack by the following procedure. First, we extract features using a feature extractor trained on the clean CIFAR-10 dataset. We form the attack targets using the test set from the 32 nearest neighbors around a randomly selected test point. Poisons are formed by bootstrap sampling a small pool of training points centered on the subpopulation centroid according to the desired poisoning rate. This poisoning strategy gives us finer control over the number of attack targets and poison samples, compared to the original clustering-based algorithm. We use a poisoning rate of 1000 poisons (2\%).

For Witches' Brew, we use the poisons from \cite{schwarzschild_just_2021} which were generated with a perturbation budget of $\varepsilon=8$ and a poisoning budget of 500 poisons (1\%). Each attack targets the misclassification of a single test sample. Sleeper Agent is a backdoor attack based on Witches' Brew. We generate poisons using a perturbation budget of $\varepsilon=16$ and a poisoning rate of 500 poisons (1\%). Each attack is generated using the source-to-target backdoor objective.

\section{Additional Experiments} \label{appx:additional-experiments}

\begin{myhideenv}

\subsection{Ablations}\label{appx:ablations}
We evaluate the sensitivity of our tool to changes in the hyperparameter $k$ and the total number of users $m$. For simplicity, we consider only the BadNets and Subpopulation attacks against CIFAR-10 using the same training and attack configurations as in \Cref{sec:dirty-label}.

\myparagraph{Hyperparameter $k$ Ablation}
We conduct an ablation on \sysname hyperparameter $k$. To remove the impact of the inner-product heuristic at different values of $\ell$, we use the simpler $\topk{k}$ influence metric. Results are shown in \Cref{fig:hyperparam-k}. Different attacks have different responses to changes in $k$. We observe that BadNets experiences a stark drop in ranking performance with $k \ge 64$, whereas ranking of users under Subpopulation attacks is more stable with large $k$. A similar observation holds for TPR at low FPR: for BadNets, TPR drops from nearly 100\% at $k=32$ to under 40\% at $k=64$ with 4 poisoned datasets. However, the TPR for subpopulation attacks with 4 poisoned users is relatively stable up to $k=256$.

\myparagraph{20 Users Ablation}
We evaluate the dependence of \sysname on the total number of users by simulating a scenario with 20 total users. We consider scenarios with 2, 4, 6, and 8 poisoned users. Results are given in \Cref{table:cifar10-20u}. We notice slight degradation in ranking performance for BadNets with our method, whereas unlearning maintains perfect rankings. On the other hand, using thresholds determined from FPR analysis produces relatively large FPR for unlearning, qualitatively similar to findings with 10 users. For Subpopulation attacks, the degradation under unlearning is more severe than with 10 users: a mAP of 0.87 with 4 poisoned owners and 20 users compared with 0.91 with 10 users. Our method, on the other hand, actually improves in performance with 20 users, maintaining $\ge 0.99$ mAP with up to 8 poisoned users. This trend is also reflected in the classification metrics: for Subpopulation attacks, unlearning TPR drops to $7\%$ with 8 poisoned datasets, whereas our method achieves a TPR of $86\%$, which is better than with only 10 total users. We speculate that this difference can be attributed to the lower variance in standardized user scores: we observe that the $1\%$ FPR threshold for gradient scores is 7.5 with 10 users and 6.8 with 20 users. A similar drop occurs for unlearning.
\end{myhideenv}

\begin{myhideenv}
\subsection{Fashion MNIST}\label{appx:fashion}
We report the metrics for the Fashion-MNIST on ConvNet training scenario in \Cref{table:fashion}. For the BadNets attack, we see the performance of the gradient score drop slightly compared to unlearning. On label flipping and subpopulation attacks, the gradient score outperforms unlearning, especially as the number of poisoned users increases. For example, at 4 poisoned users, unlearning achieves a TPR of $25\%$ and $33\%$ on label flipping and subpopulation attacks, respectively, while the gradient score maintains TPRs of above $90\%$ for both attacks. These results indicate the resilience of the gradient score to poisoning attacks distributed across multiple datasets. We note that the kNN baseline fails to reliably detect malicious users for Subpopulation attacks, achieving only $<2\%$ TPR.

\end{myhideenv}

\subsection{Text and Malware Data Modalities}
\label{appx:modality}
We evaluate \sysname on learning tasks from text and malware domains. For text, we use the SST-2 dataset \cite{socher_recursive_2013} and fine-tune on a $\roberta$ model \cite{liu_roberta_2019}. For malware, we use the Ember dataset \cite{anderson_ember_2018} and train with the EmberNN model architecture \cite{severi_explanation-guided_2021}.

\myparagraph{Attack Setup}
We perform backdoor poisoning attacks against RoBERTa/SST-2 using four backdoor variants used by Pei et al. \cite{pei_textguard_2023, cui_unified_2022}. Each variant leverages a different text trigger. First, a straightforward ``BadNets'' extension injects a pair of arbitrary characters at a random position in the sample. For our experiments, we apply a random selection from the set $\{\text{``fe''}, \text{``hr''}, \text{``qp''}, \text{``ge''}\}$. Second, we use a ``style'' backdoor that applies style transfer to clean examples \cite{qi_mind_2021}. We use the default ``Bible'' style for our attacks. Third, we consider a ``syntax'' backdoor in which the trigger is encoded in the syntactic structure of the sentence. For our experiments, we use the syntax \textit{S(SBAR)(,)(NP)(VP)(.)} as the trigger. Finally, we consider a ``sentence'' backdoor using the fixed sentence ``I watch this 3D movie'' as the trigger.

For Ember / EmberNN, we consider two different backdoor attacks. First, we implement a dirty-label backdoor attack that works by choosing a random subset of features from the input space and assigning a fixed pattern to those features across all backdoored samples. We call this attack BadNets for Ember. For the remaining attack, we use the clean-label explanation-guided backdoor attack of Severi et al. \cite{severi_explanation-guided_2021}. For both attacks we use a poison size of 6\,000 (1\%).

\myparagraph{Results}
We report the metrics for attacks against SST-2 on $\roberta$ in \Cref{table:sst2}.
We report the metrics for attacks against Ember on EmberNN in \Cref{table:ember}.
We also plot attack-conditioned ROC curves in \Cref{fig:sst2-rocs,fig:ember-rocs}.

In attacks against Ember, the gradient score and kNN baseline achieve high TPR against the explanation-guided backdoor ($97\%$ and $100\%$, respectively), while the unlearning performs comparably worse ($<73\%$). A similar pattern occurs for the BadNets variant, with the gradient and kNN scores behaving comparably around 77-79\% TPR, and with the unlearning score dropping to $<6\%$ TPR. In general, the explanation-guided backdoor is easier to detect.

One possible explanation for these results is that because the explanation-guided backdoor targets features that are more relevant to the model, the terms in the gradient responsible for encoding the backdoor are more prominent than when the features are selected at random. This affects the gradient score directly--because its effectiveness relies on alignment between poisons and attack samples--and unlearning indirectly, because the rate at which the malicious features are approximately unlearned will be slower.

All methods perform well detecting attacks against SST-2. In three out of four attacks, the gradient and unlearning scores maintain above $98\%$ TPR. For the BadNets attack, the gradient and unlearning scores maintain above $90\%$ TPR.

\subsection{Clean-label Attacks}
We report the metrics for clean-label attacks against CIFAR-10 on ResNet-18 in \Cref{table:cifar10-clean}.
We also plot attack-conditioned ROC curves in \Cref{fig:cifar-clean-rocs}.
We observe that clean-label attacks pose a challenge to our traceback scoring methods.
Both Unlearning and the gradient-based score experience low TPR of below $6\%$ for Witches' Brew when poisons are distributed across multiple poisoned owners.

\begin{figure}[!htbp]
    \centering

    \begin{subfigure}[b]{1.0\linewidth}
        \centering
        \includegraphics[width=1.0\linewidth]{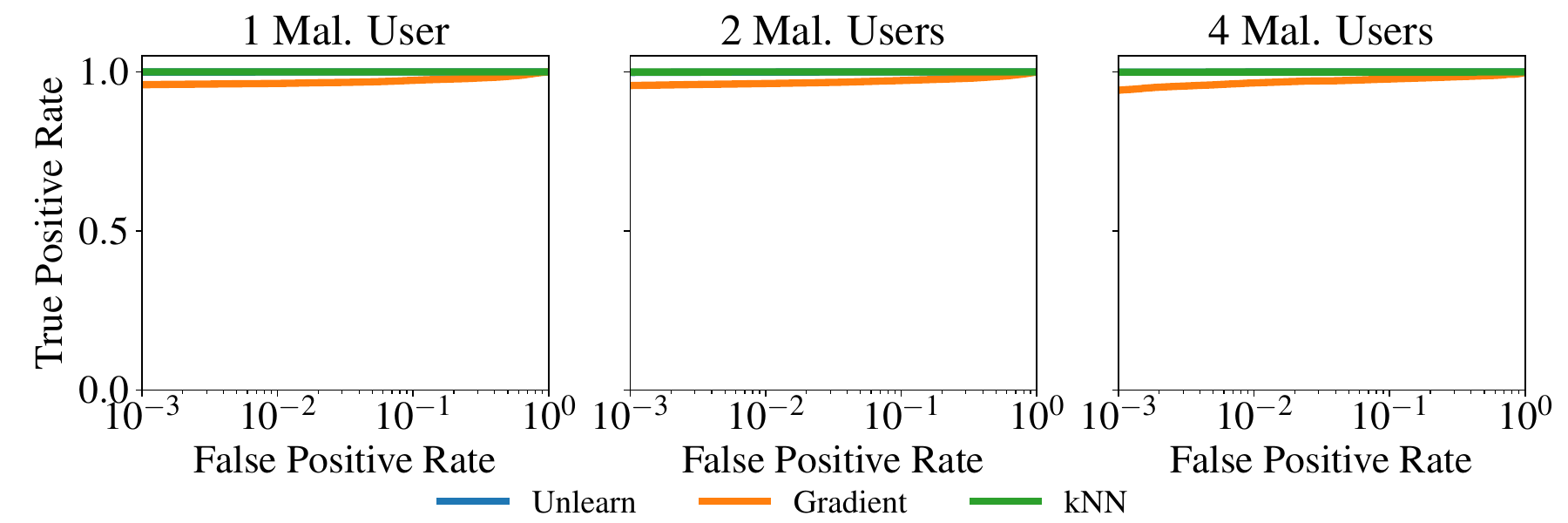}
        \caption{BadNets}
        \label{fig:fashion-badnets-roc}
    \end{subfigure}
    
    \begin{subfigure}[b]{1.0\linewidth}
        \centering
        \includegraphics[width=1.0\linewidth]{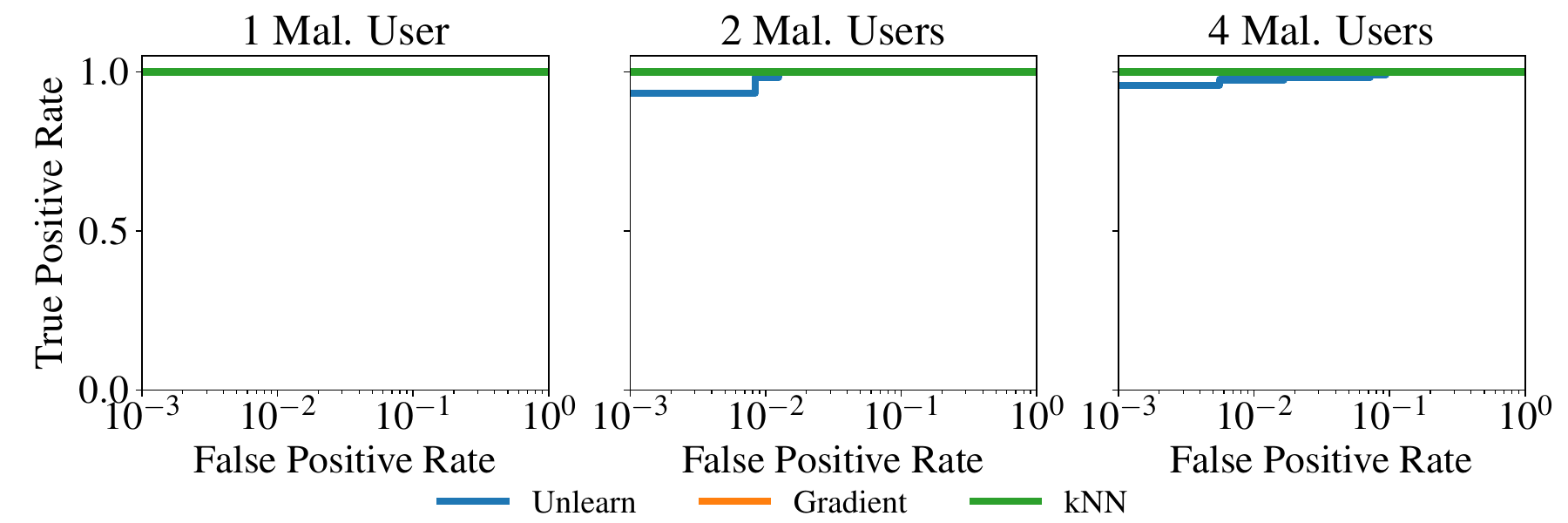}
        \caption{Label Flip}
        \label{fig:fashion-labelflip-roc}
    \end{subfigure}
    
    \begin{subfigure}[b]{1.0\linewidth}
        \centering
        \includegraphics[width=1.0\linewidth]{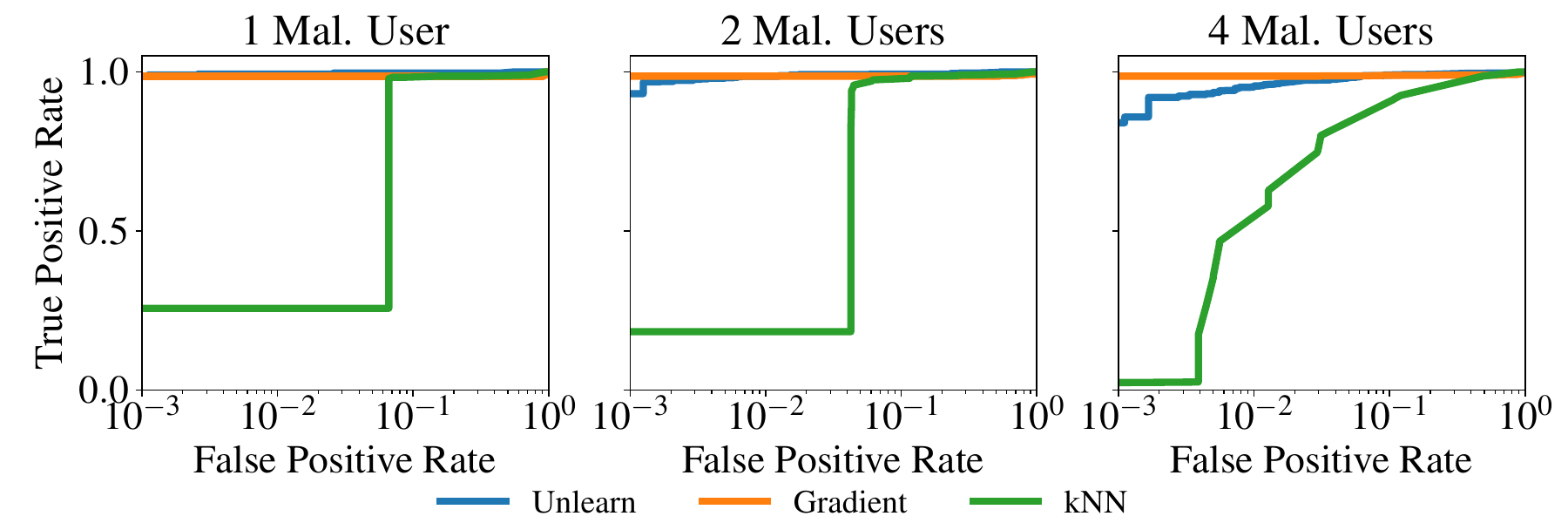}
        \caption{Subpopulation}
        \label{fig:fashion-subpop-roc}
    \end{subfigure}
    
    \caption{Malicious user identification ROC curves for standard attacks against Fashion / ConvNet.}
    \label{fig:fashion-rocs}
\end{figure}

\begin{myhideenv}
\begin{figure}[!htbp]
    \centering
    \begin{subfigure}[b]{1.0\linewidth}
        \centering
        \includegraphics[width=1.0\linewidth]{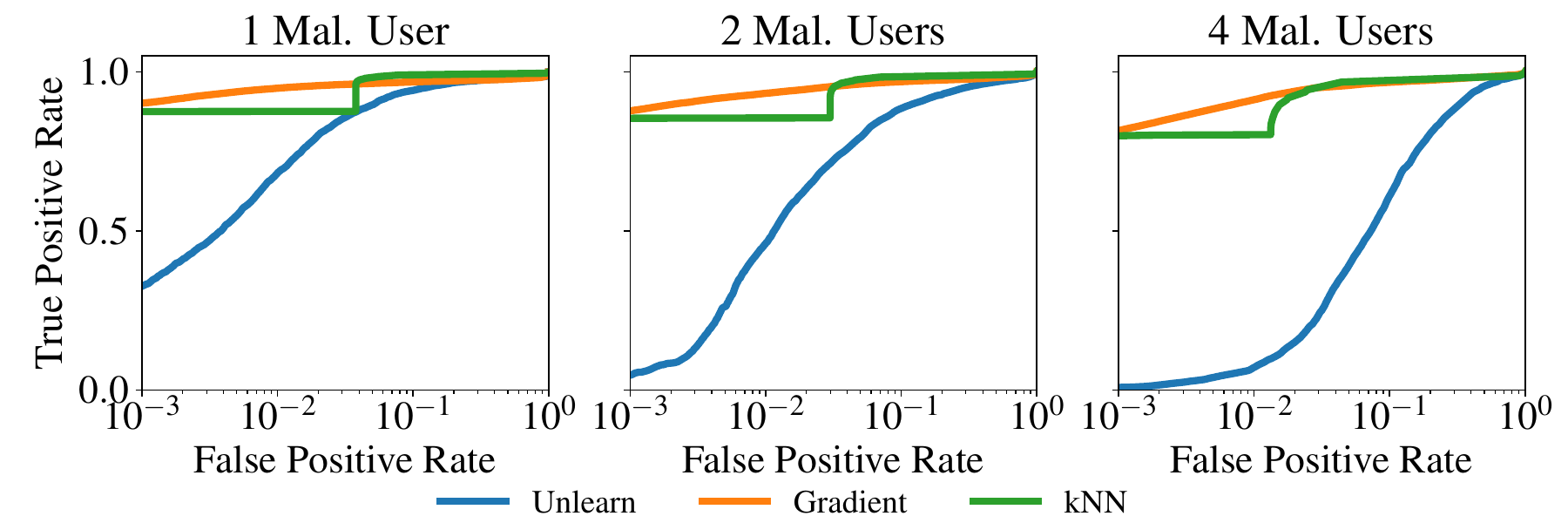}
        \caption{BadNets}
        \label{fig:ember-badnets-roc}
    \end{subfigure}
    
    \begin{subfigure}[b]{1.0\linewidth}
        \centering
        \includegraphics[width=1.0\linewidth]{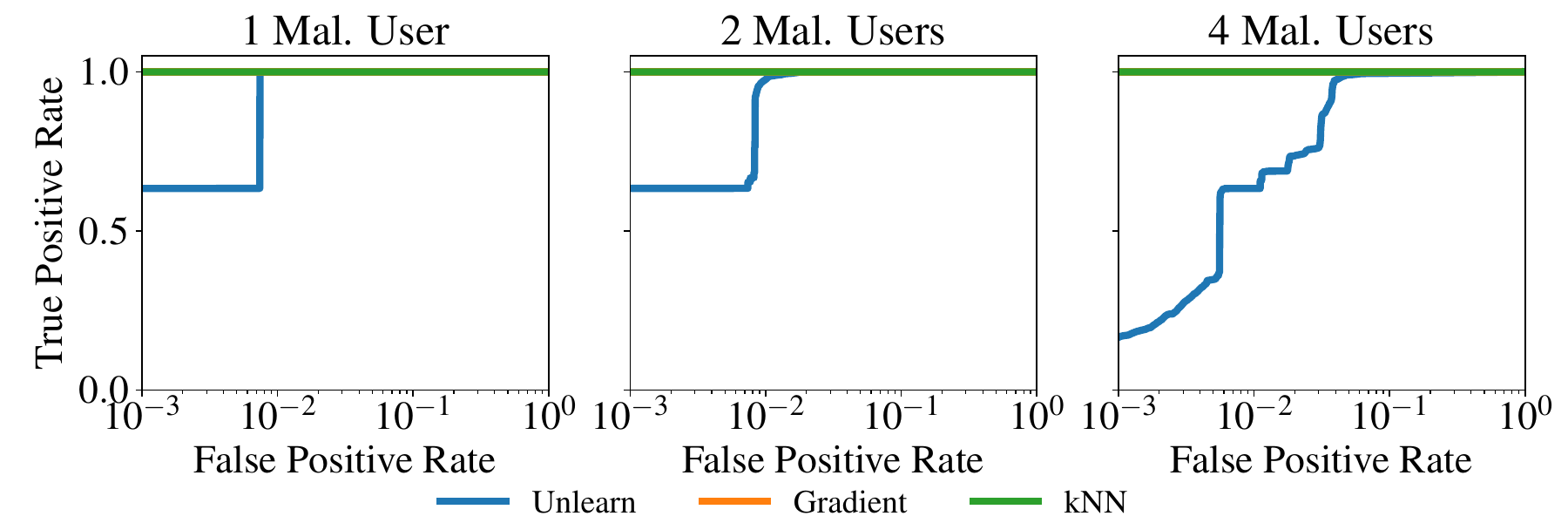}
        \caption{Explanation}
        \label{fig:ember-explanation-roc}
    \end{subfigure}
    
    \caption{Malicious user identification ROC curves for attacks against Ember / EmberNN.}
    \label{fig:ember-rocs}
\end{figure}
\end{myhideenv}

\begin{myhideenv}
\label{appx:clean-label}
\begin{figure}[!htbp]
    \centering
    \begin{subfigure}[b]{1.0\linewidth}
        \centering
        \includegraphics[width=1.0\linewidth]{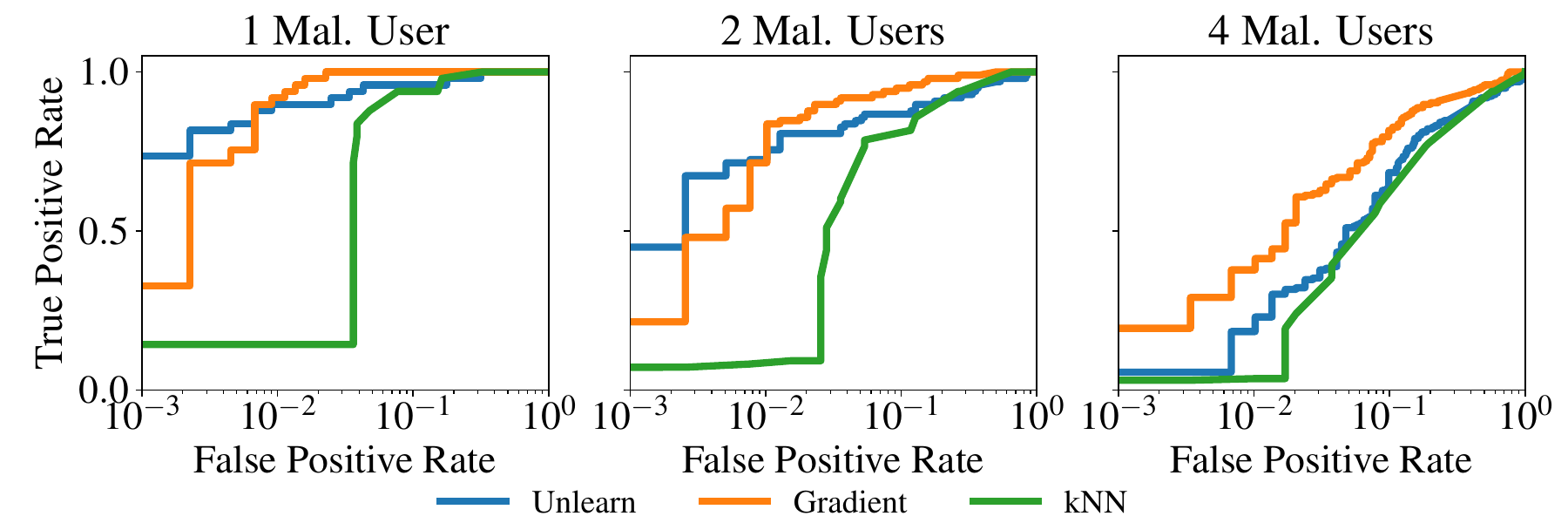}
        \caption{Witches' Brew}
        \label{fig:cifar-witch-roc}
    \end{subfigure}

    \begin{subfigure}[b]{1.0\linewidth}
        \centering
        \includegraphics[width=1.0\linewidth]{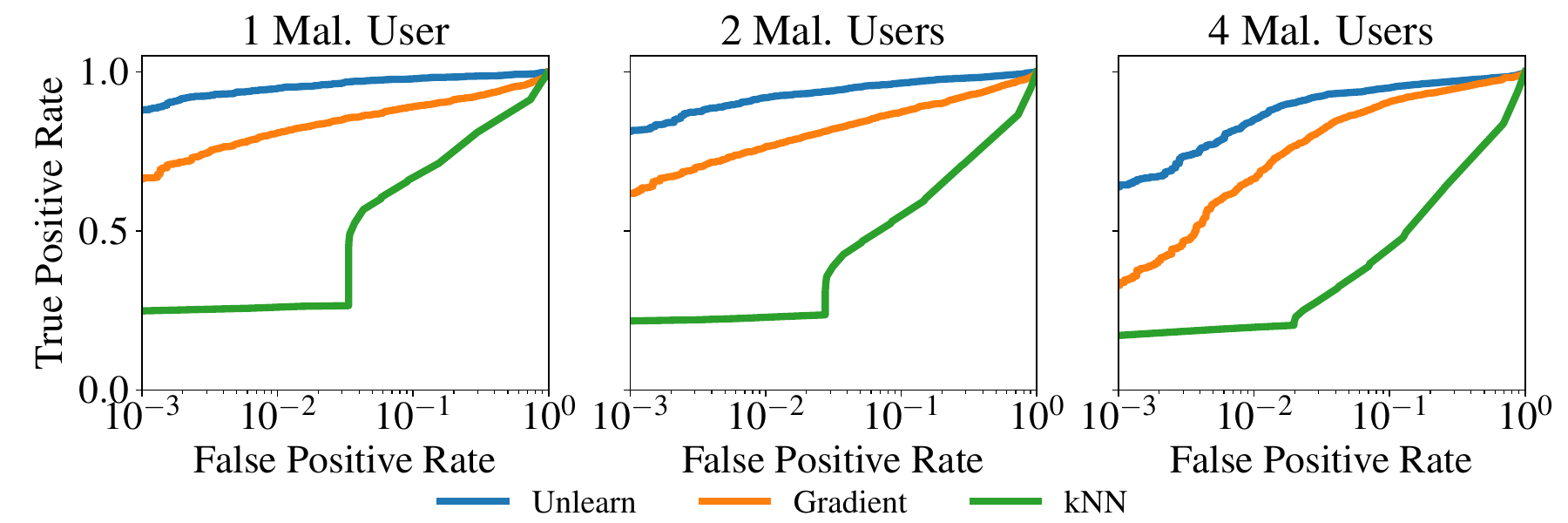}
        \caption{Sleeper Agent}
        \label{fig:cifar-sleeper-roc}
    \end{subfigure}
    
    \caption{Malicious user identification ROC curves for clean-label attacks against CIFAR-10 / ResNet18.}
    \label{fig:cifar-clean-rocs}
\end{figure}
\end{myhideenv}

\begin{figure}[!htbp]
    \centering
    \begin{subfigure}[b]{1.0\linewidth}
        \centering
        \includegraphics[width=0.95\linewidth]{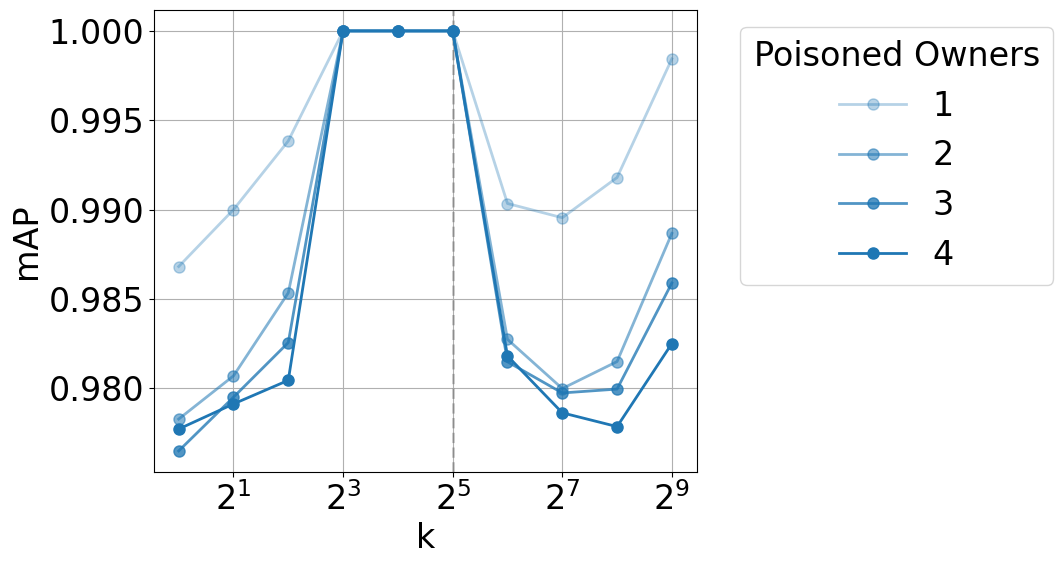}
        \caption{BadNets mAP vs $k$}
        \label{fig:ablate-k-map-badnets}
    \end{subfigure}
    \begin{subfigure}[b]{1.0\linewidth}
        \centering
        \includegraphics[width=0.95\linewidth]{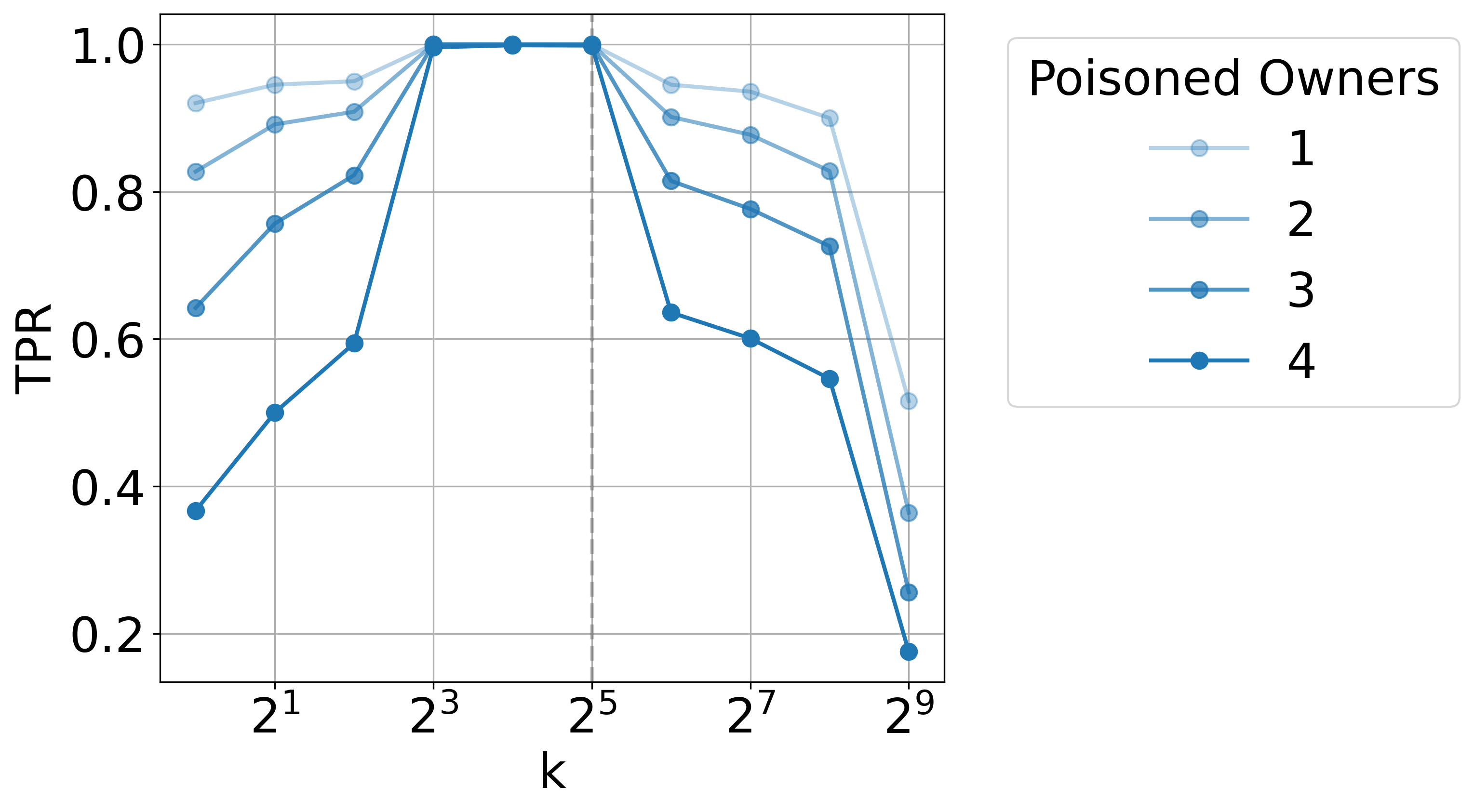}
        \caption{BadNets TPR vs $k$}
        \label{fig:ablate-k-tpr-badnets}
    \end{subfigure}
    \begin{subfigure}[b]{1.0\linewidth}
        \centering
        \includegraphics[width=0.95\linewidth]{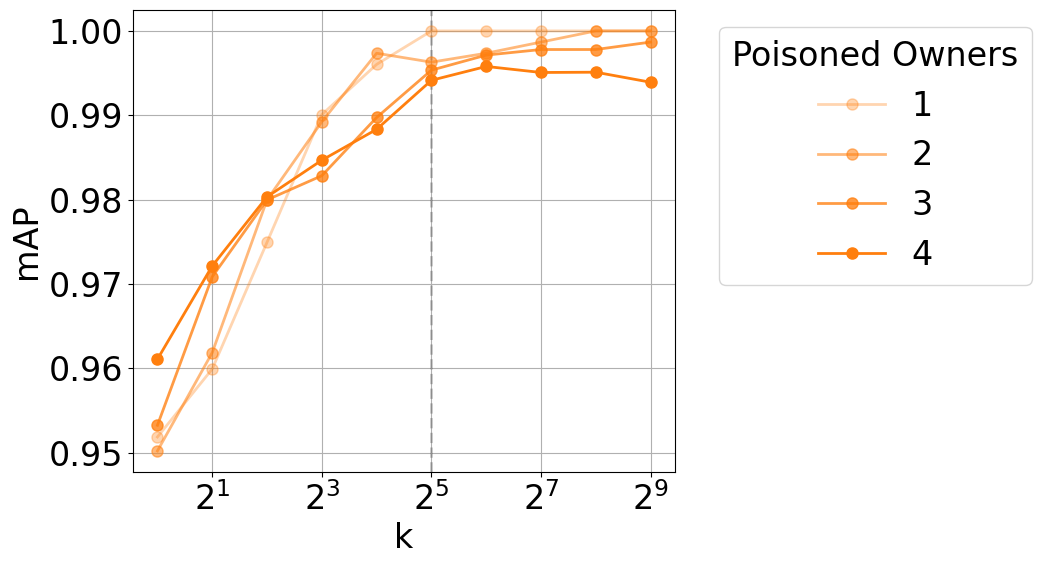}
        \caption{Subpop mAP vs $k$}
        \label{fig:ablate-k-map-subpop}
    \end{subfigure}
    \begin{subfigure}[b]{1.0\linewidth}
        \centering
        \includegraphics[width=0.95\linewidth]{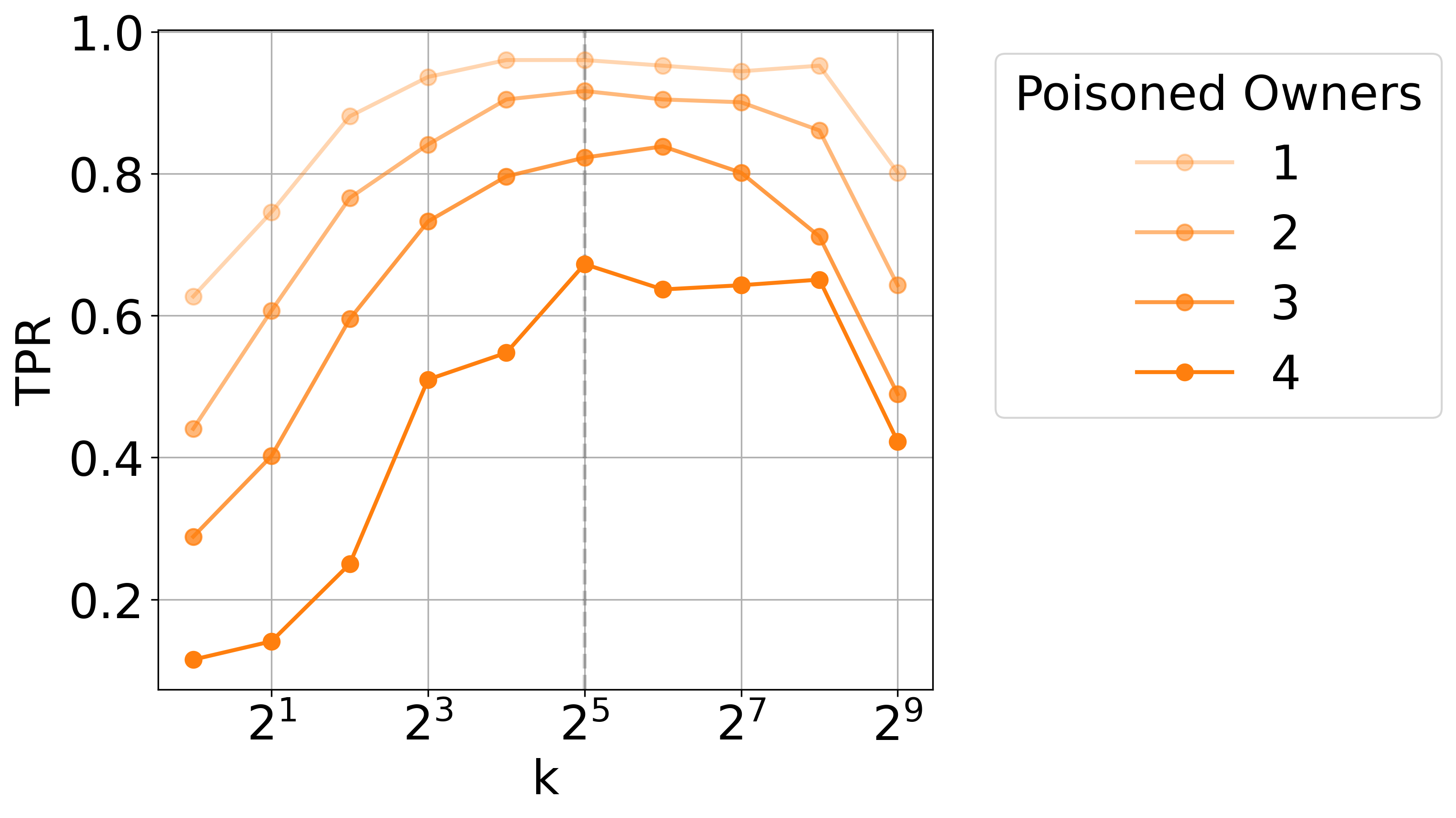}
        \caption{Subpop TPR vs $k$}
        \label{fig:ablate-k-tpr-subpop}
    \end{subfigure}
    \caption{Hyperparameter $k$ ablation. We plot the mAP score and TPR at 1\% FPR for various values of $k$ for BadNets and Subpopulation attacks against CIFAR-10.}
    \label{fig:hyperparam-k}
\end{figure}

\begin{figure}[!htbp]
    \centering
    \begin{subfigure}[b]{1.0\linewidth}
        \centering
        \includegraphics[width=1.0\linewidth]{figs/cifar-n_badnets-roc.pdf}
        \caption{Noisy BadNets}
        \label{fig:fashion-badnets-noise-roc}
    \end{subfigure}
    
    \begin{subfigure}[b]{1.0\linewidth}
        \centering
        \includegraphics[width=1.0\linewidth]{figs/cifar-s_badnets-roc.pdf}
        \caption{$\sigma$-BadNets}
        \label{fig:fashion-badnets-permute-roc}
    \end{subfigure}
    
    \begin{subfigure}[b]{1.0\linewidth}
        \centering
        \includegraphics[width=1.0\linewidth]{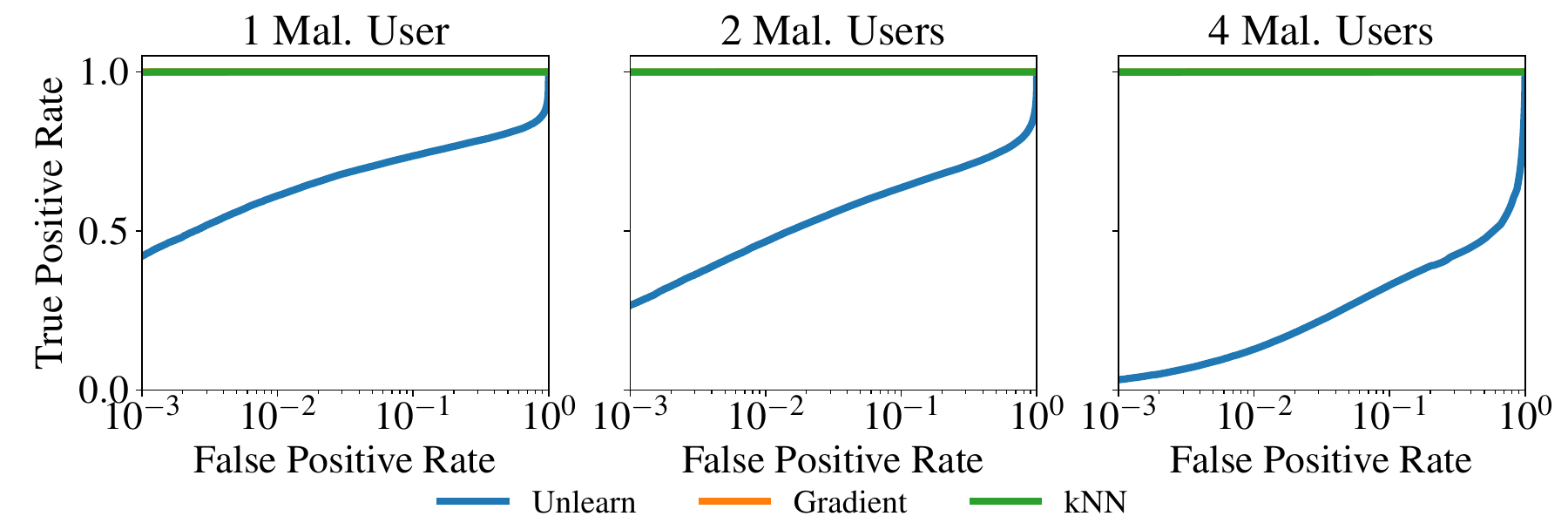}
        \caption{Noisy $\sigma$-BadNets}
        \label{fig:fashion-badnets-permute-noise-roc}
    \end{subfigure}
    
    \caption{Malicious user identification ROC curves for unlearning-aware attacks against Fashion / ConvNet.}
    \label{fig:fashion-extra-rocs}
\end{figure}

\begin{myhideenv}
\begin{figure}[!htbp]
    \centering
    \begin{subfigure}[b]{1.0\linewidth}
        \centering
        \includegraphics[width=1.0\linewidth]{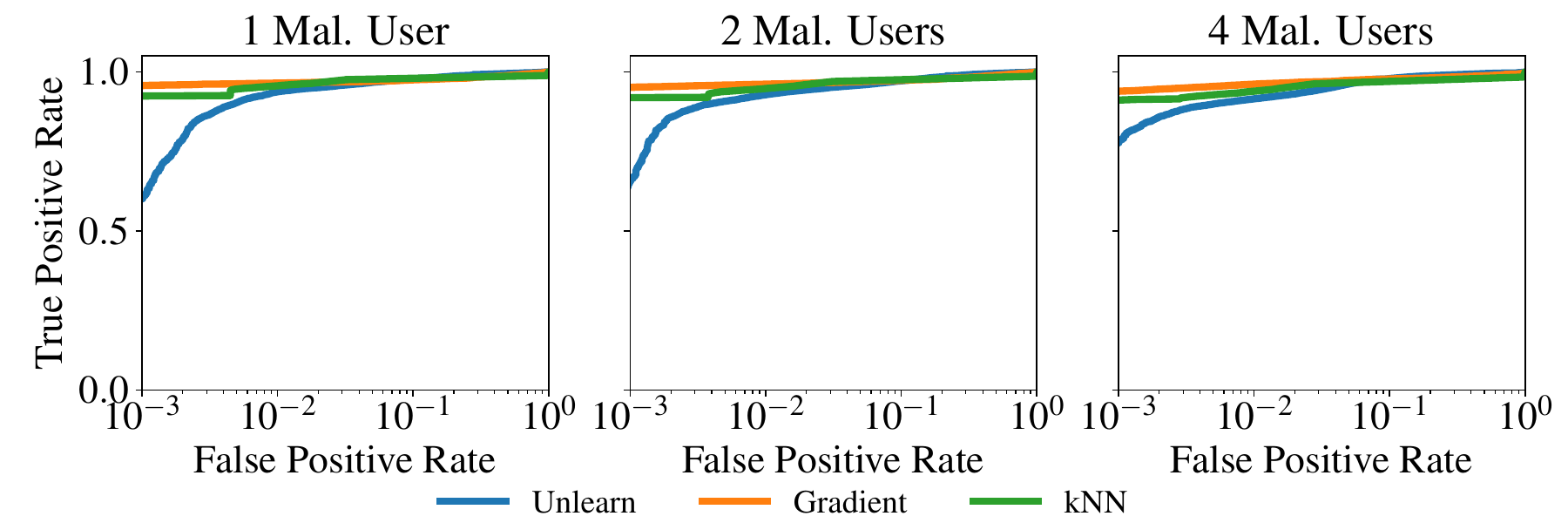}
        \caption{BadNets}
        \label{fig:sst2-badnets-roc}
    \end{subfigure}
    
    \begin{subfigure}[b]{1.0\linewidth}
        \centering
        \includegraphics[width=1.0\linewidth]{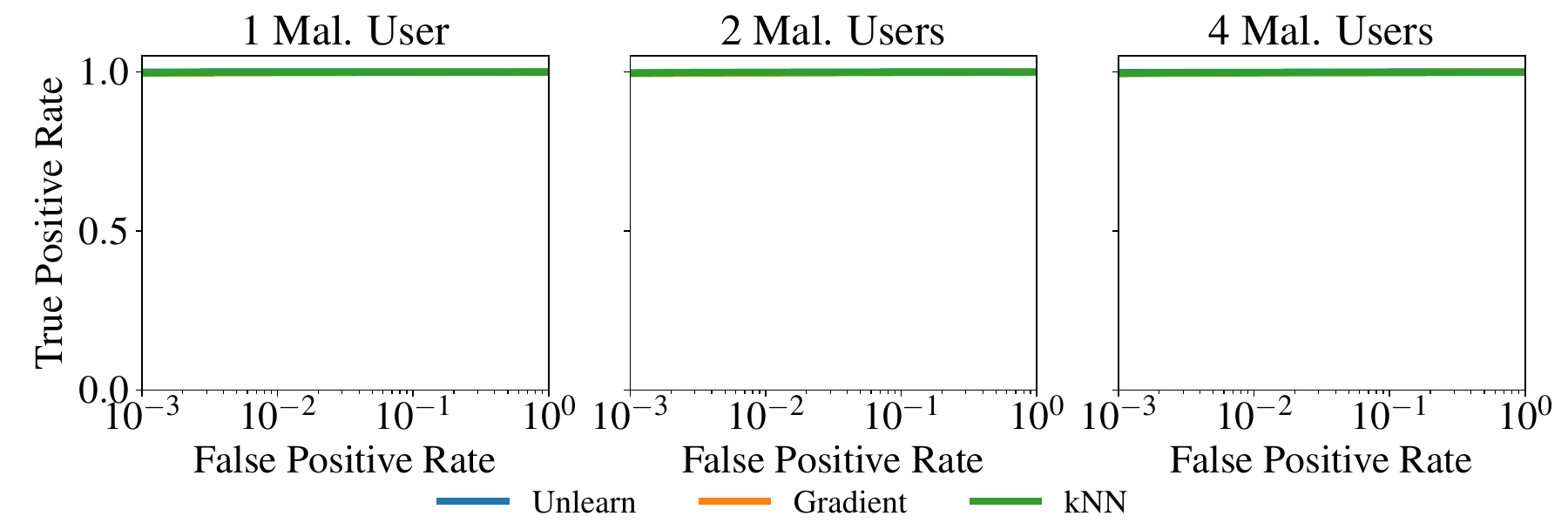}
        \caption{Style}
        \label{fig:sst2-style-roc}
    \end{subfigure}
    
    \begin{subfigure}[b]{1.0\linewidth}
        \centering
        \includegraphics[width=1.0\linewidth]{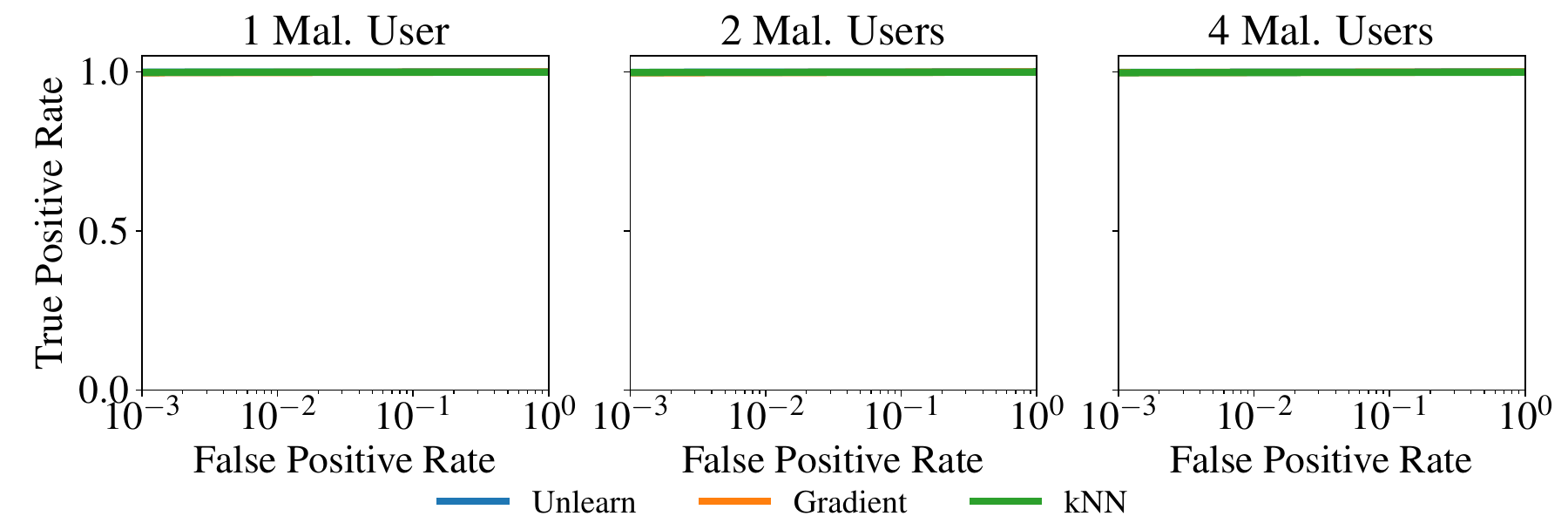}
        \caption{Syntax}
        \label{fig:sst2-syntax-roc}
    \end{subfigure}

    \begin{subfigure}[b]{1.0\linewidth}
        \centering
        \includegraphics[width=1.0\linewidth]{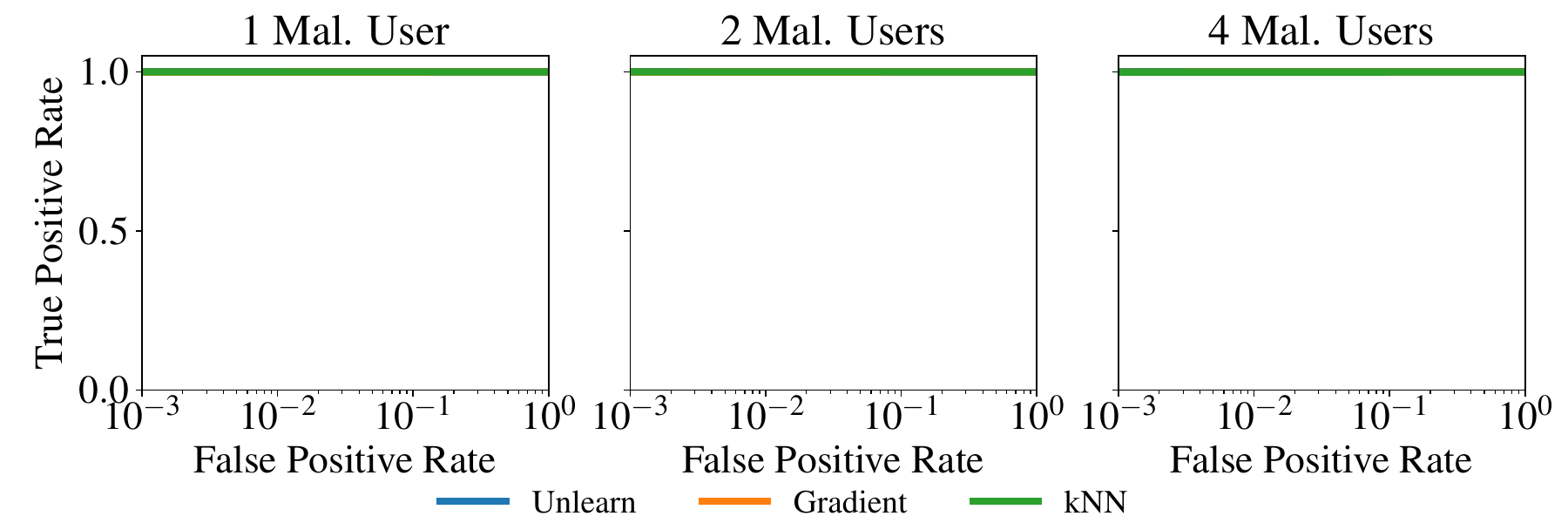}
        \caption{Sentence}
        \label{fig:sst2-sentence-roc}
    \end{subfigure}
    
    \caption{Malicious user identification ROC curves for attacks against SST-2 / $\roberta$.}
    \label{fig:sst2-rocs}
\end{figure}
\end{myhideenv}

\begin{table*}[!htbp]
\centering
\setlength{\tabcolsep}{3.1pt}
\caption{Results for attacks against CIFAR-10 on ResNet18 with 20 users.}
\label{table:cifar10-20u}
\begin{tabular}{@{}ccSSSSS!{\quad}SSSSS!{\quad}SSSSS@{}}
\toprule
\multicolumn{1}{c}{\multirow{2}{*}{\textbf{Attack}}}
& \multicolumn{1}{c}{\multirow{2}{*}{\shortstack[c]{\textbf{Mal.}\\\textbf{Users}}}}
& \multicolumn{5}{c}{\textbf{kNN (Baseline)}} & \multicolumn{5}{c}{\textbf{Unlearning (\S\ref{sec:unlearning})}}
& \multicolumn{5}{c}{\textbf{Gradient (\S\ref{sec:gradient})}} \\
\cmidrule(lr{1.2em}){3-7} \cmidrule(lr{1.2em}){8-12} \cmidrule(l){13-17}
\multicolumn{1}{c}{}
& \multicolumn{1}{c}{}
& mAP & mRR & {\scriptsize TPR (\%)} & {\scriptsize FPR (\%)} & AUC
& mAP & mRR & {\scriptsize TPR (\%)} & {\scriptsize FPR (\%)} & AUC
& mAP & mRR & {\scriptsize TPR (\%)} & {\scriptsize FPR (\%)} & AUC \\
\hline
\multirow{4}{*}{BadNets}
    & 2 & 0.9999 & 0.9999 & 99.674 & 0.007 & 1.000 & 0.9999 & 0.9999 & 99.980 & 0.881 & 1.000 & 1.0000 & 1.0000 & 99.963 & 0.381 & 1.000 \\
    & 4 & 0.9995 & 0.9999 & 98.836 & 0.005 & 1.000 & 0.9998 & 0.9998 & 99.968 & 0.355 & 1.000 & 0.9999 & 1.0000 & 99.922 & 0.097 & 1.000 \\
    & 6 & 0.9974 & 0.9999 & 91.434 & 0.003 & 0.998 & 0.9998 & 0.9999 & 99.948 & 0.045 & 1.000 & 0.9999 & 1.0000 & 99.837 & 0.010 & 1.000 \\
    & 8 & 0.9916 & 0.9999 & 77.383 & 0.001 & 0.992 & 0.9998 & 0.9999 & 99.879 & 0.004 & 1.000 & 0.9999 & 1.0000 & 99.374 & 0.000 & 1.000 \\
\hline
\multirow{4}{*}{Subpop}
    & 2 & 0.9931 & 0.9937 & 73.125 & 1.250 & 0.927 & 0.9215 & 0.9519 & 83.125 & 0.139 & 0.950 & 0.9892 & 0.9891 & 98.750 & 0.000 & 0.994 \\
    & 4 & 0.9760 & 0.9937 & 56.250 & 1.016 & 0.906 & 0.8735 & 0.9435 & 59.688 & 0.000 & 0.927 & 0.9915 & 0.9917 & 98.750 & 0.078 & 0.996 \\
    & 6 & 0.9398 & 0.9937 & 21.458 & 0.446 & 0.910 & 0.8590 & 0.9233 & 38.125 & 0.000 & 0.900 & 0.9942 & 1.0000 & 96.667 & 0.000 & 0.997 \\
    & 8 & 0.9262 & 0.9937 & 2.344 & 0.312 & 0.917 & 0.8385 & 0.9260 & 7.031 & 0.104 & 0.849 & 0.9953 & 1.0000 & 86.719 & 0.000 & 0.998 \\
\bottomrule
\end{tabular}
\end{table*}

\begin{table*}[!htbp]
\centering
\setlength{\tabcolsep}{3.1pt}
\caption{Results for standard attacks against Fashion on ConvNet.}
\label{table:fashion}
\begin{tabular}{@{}ccSSSSS!{\quad}SSSSS!{\quad}SSSSS@{}}
\toprule
\multicolumn{1}{c}{\multirow{2}{*}{\textbf{Attack}}}
& \multicolumn{1}{c}{\multirow{2}{*}{\shortstack[c]{\textbf{Mal.}\\\textbf{Users}}}}
& \multicolumn{5}{c}{\textbf{kNN (Baseline)}} & \multicolumn{5}{c}{\textbf{Unlearning (\S\ref{sec:unlearning})}}
& \multicolumn{5}{c}{\textbf{Gradient (\S\ref{sec:gradient})}} \\
\cmidrule(lr{1.2em}){3-7} \cmidrule(lr{1.2em}){8-12} \cmidrule(l){13-17}
\multicolumn{1}{c}{}
& \multicolumn{1}{c}{}
& mAP & mRR & {\scriptsize TPR (\%)} & {\scriptsize FPR (\%)} & AUC
& mAP & mRR & {\scriptsize TPR (\%)} & {\scriptsize FPR (\%)} & AUC
& mAP & mRR & {\scriptsize TPR (\%)} & {\scriptsize FPR (\%)} & AUC \\
\hline
\multirow{4}{*}{BadNets}
    & 1 & 0.9998 & 0.9998 & 99.863 & 0.000 & 1.000 & 1.0000 & 1.0000 & 100.000 & 2.691 & 1.000 & 0.9785 & 0.9785 & 96.296 & 0.599 & 0.986 \\
    & 2 & 0.9998 & 0.9999 & 99.850 & 0.000 & 1.000 & 1.0000 & 1.0000 & 99.987 & 2.533 & 1.000 & 0.9787 & 0.9841 & 96.008 & 0.266 & 0.985 \\
    & 3 & 0.9998 & 0.9999 & 99.663 & 0.000 & 1.000 & 1.0000 & 1.0000 & 99.987 & 1.818 & 1.000 & 0.9803 & 0.9853 & 95.440 & 0.132 & 0.986 \\
    & 4 & 0.9998 & 0.9999 & 97.410 & 0.000 & 1.000 & 1.0000 & 1.0000 & 99.978 & 0.566 & 1.000 & 0.9822 & 0.9859 & 94.205 & 0.097 & 0.987 \\
\hline
\multirow{4}{*}{LabelFlip}
    & 1 & 1.0000 & 1.0000 & 100.000 & 0.000 & 1.000 & 1.0000 & 1.0000 & 100.000 & 1.852 & 1.000 & 1.0000 & 1.0000 & 100.000 & 0.741 & 1.000 \\
    & 2 & 1.0000 & 1.0000 & 100.000 & 0.000 & 1.000 & 1.0000 & 1.0000 & 96.667 & 0.833 & 0.999 & 1.0000 & 1.0000 & 100.000 & 0.417 & 1.000 \\
    & 3 & 1.0000 & 1.0000 & 100.000 & 0.000 & 1.000 & 0.9956 & 1.0000 & 77.778 & 0.476 & 0.997 & 1.0000 & 1.0000 & 100.000 & 0.000 & 1.000 \\
    & 4 & 1.0000 & 1.0000 & 99.167 & 0.000 & 1.000 & 1.0000 & 1.0000 & 25.000 & 0.000 & 0.998 & 1.0000 & 1.0000 & 90.000 & 0.000 & 1.000 \\
\hline
\multirow{4}{*}{Subpop}
    & 1 & 0.9885 & 0.9885 & 25.667 & 0.185 & 0.941 & 0.9972 & 0.9972 & 99.333 & 0.259 & 0.998 & 0.9880 & 0.9880 & 98.667 & 0.185 & 0.987 \\
    & 2 & 0.9890 & 0.9914 & 18.333 & 0.042 & 0.957 & 0.9956 & 0.9972 & 96.000 & 0.125 & 0.997 & 0.9888 & 0.9883 & 98.667 & 0.042 & 0.988 \\
    & 3 & 0.9846 & 0.9917 & 6.333 & 0.048 & 0.972 & 0.9922 & 0.9940 & 72.333 & 0.000 & 0.994 & 0.9901 & 0.9892 & 98.667 & 0.048 & 0.990 \\
    & 4 & 0.9731 & 0.9942 & 1.750 & 0.000 & 0.955 & 0.9918 & 0.9961 & 32.833 & 0.000 & 0.993 & 0.9910 & 0.9894 & 98.333 & 0.056 & 0.991 \\
\bottomrule
\end{tabular}
\end{table*}

\begin{table*}[!htbp]
\centering
\setlength{\tabcolsep}{3.1pt}
\caption{Results for unlearning-aware attacks against Fashion on ConvNet.}
\label{table:fashion-extra}
\begin{tabular}{@{}ccSSSSS!{\quad}SSSSS!{\quad}SSSSS@{}}
\toprule
\multicolumn{1}{c}{\multirow{2}{*}{\textbf{Attack}}}
& \multicolumn{1}{c}{\multirow{2}{*}{\shortstack[c]{\textbf{Mal.}\\\textbf{Users}}}}
& \multicolumn{5}{c}{\textbf{kNN (Baseline)}} & \multicolumn{5}{c}{\textbf{Unlearning (\S\ref{sec:unlearning})}}
& \multicolumn{5}{c}{\textbf{Gradient (\S\ref{sec:gradient})}} \\
\cmidrule(lr{1.2em}){3-7} \cmidrule(lr{1.2em}){8-12} \cmidrule(l){13-17}
\multicolumn{1}{c}{}
& \multicolumn{1}{c}{}
& mAP & mRR & {\scriptsize TPR (\%)} & {\scriptsize FPR (\%)} & AUC
& mAP & mRR & {\scriptsize TPR (\%)} & {\scriptsize FPR (\%)} & AUC
& mAP & mRR & {\scriptsize TPR (\%)} & {\scriptsize FPR (\%)} & AUC \\
\hline
\multirow{4}{*}{\shortstack[c]{Noisy\\BadNets}}
    & 1 & 1.0000 & 1.0000 & 100.000 & 0.000 & 1.000 & 0.8999 & 0.8999 & 65.656 & 0.344 & 0.921 & 0.9999 & 0.9999 & 99.977 & 0.513 & 1.000 \\
    & 2 & 1.0000 & 1.0000 & 100.000 & 0.000 & 1.000 & 0.8653 & 0.8896 & 56.054 & 0.373 & 0.878 & 0.9999 & 0.9999 & 99.972 & 0.171 & 1.000 \\
    & 3 & 1.0000 & 1.0000 & 99.845 & 0.000 & 1.000 & 0.8660 & 0.8874 & 38.078 & 0.072 & 0.866 & 0.9999 & 0.9999 & 99.953 & 0.043 & 1.000 \\
    & 4 & 1.0000 & 1.0000 & 97.257 & 0.000 & 1.000 & 0.8792 & 0.8858 & 15.458 & 0.179 & 0.868 & 0.9999 & 0.9999 & 99.925 & 0.011 & 1.000 \\
\hline
\multirow{4}{*}{$\sigma$-BadNets}
    & 1 & 1.0000 & 1.0000 & 99.907 & 0.001 & 1.000 & 0.9480 & 0.9480 & 89.753 & 2.212 & 0.968 & 0.9998 & 0.9998 & 99.923 & 0.603 & 1.000 \\
    & 2 & 0.9999 & 1.0000 & 99.856 & 0.000 & 1.000 & 0.9234 & 0.9235 & 80.174 & 1.438 & 0.945 & 0.9997 & 0.9998 & 99.820 & 0.218 & 1.000 \\
    & 3 & 0.9999 & 1.0000 & 99.616 & 0.000 & 1.000 & 0.9007 & 0.8966 & 66.344 & 1.296 & 0.923 & 0.9997 & 0.9998 & 99.543 & 0.049 & 1.000 \\
    & 4 & 0.9999 & 1.0000 & 97.118 & 0.000 & 1.000 & 0.8882 & 0.8783 & 46.720 & 1.081 & 0.896 & 0.9997 & 0.9998 & 98.643 & 0.005 & 1.000 \\
\hline
\multirow{4}{*}{\shortstack[c]{Noisy\\$\sigma$-BadNets}}
    & 1 & 0.9999 & 0.9999 & 99.961 & 0.000 & 1.000 & 0.6429 & 0.6429 & 58.543 & 0.767 & 0.612 & 0.9972 & 0.9972 & 99.459 & 0.668 & 0.998 \\
    & 2 & 0.9999 & 0.9999 & 99.955 & 0.000 & 1.000 & 0.5360 & 0.5640 & 41.261 & 0.705 & 0.457 & 0.9971 & 0.9976 & 99.225 & 0.259 & 0.998 \\
    & 3 & 0.9999 & 0.9999 & 99.816 & 0.000 & 1.000 & 0.4798 & 0.4846 & 27.387 & 0.388 & 0.343 & 0.9971 & 0.9978 & 98.920 & 0.060 & 0.998 \\
    & 4 & 0.9999 & 0.9999 & 97.702 & 0.000 & 1.000 & 0.4721 & 0.4362 & 17.083 & 0.362 & 0.271 & 0.9972 & 0.9981 & 98.083 & 0.008 & 0.998 \\
\bottomrule
\end{tabular}
\end{table*}

\begin{table*}[!htbp]
\centering
\setlength{\tabcolsep}{3.1pt}
\caption{Results for attacks against SST-2 on $\roberta$.}
\label{table:sst2}
\begin{tabular}{@{}ccSSSSS!{\quad}SSSSS!{\quad}SSSSS@{}}
\toprule
\multicolumn{1}{c}{\multirow{2}{*}{\textbf{Attack}}}
& \multicolumn{1}{c}{\multirow{2}{*}{\shortstack[c]{\textbf{Mal.}\\\textbf{Users}}}}
& \multicolumn{5}{c}{\textbf{kNN (Baseline)}} & \multicolumn{5}{c}{\textbf{Unlearning (\S\ref{sec:unlearning})}}
& \multicolumn{5}{c}{\textbf{Gradient (\S\ref{sec:gradient})}} \\
\cmidrule(lr{1.2em}){3-7} \cmidrule(lr{1.2em}){8-12} \cmidrule(l){13-17}
\multicolumn{1}{c}{}
& \multicolumn{1}{c}{}
& mAP & mRR & {\scriptsize TPR (\%)} & {\scriptsize FPR (\%)} & AUC
& mAP & mRR & {\scriptsize TPR (\%)} & {\scriptsize FPR (\%)} & AUC
& mAP & mRR & {\scriptsize TPR (\%)} & {\scriptsize FPR (\%)} & AUC \\
\hline
\multirow{4}{*}{BadNets}
    & 1 & 0.9690 & 0.9690 & 92.381 & 0.031 & 0.982 & 0.9810 & 0.9810 & 96.261 & 3.945 & 0.990 & 0.9795 & 0.9795 & 96.238 & 0.458 & 0.987 \\
    & 2 & 0.9687 & 0.9741 & 91.858 & 0.023 & 0.978 & 0.9784 & 0.9819 & 94.734 & 2.497 & 0.989 & 0.9792 & 0.9823 & 95.352 & 0.165 & 0.987 \\
    & 3 & 0.9719 & 0.9795 & 91.469 & 0.028 & 0.975 & 0.9529 & 0.9501 & 93.014 & 2.934 & 0.983 & 0.9799 & 0.9848 & 94.152 & 0.062 & 0.987 \\
    & 4 & 0.9756 & 0.9844 & 89.347 & 0.029 & 0.973 & 0.9827 & 0.9877 & 90.647 & 0.681 & 0.989 & 0.9816 & 0.9870 & 92.970 & 0.041 & 0.987 \\
\hline
\multirow{4}{*}{Style}
    & 1 & 0.9985 & 0.9985 & 99.294 & 0.004 & 0.999 & 0.9992 & 0.9992 & 99.974 & 0.457 & 1.000 & 0.9991 & 0.9991 & 99.748 & 0.650 & 0.999 \\
    & 2 & 0.9982 & 0.9986 & 99.217 & 0.002 & 0.999 & 0.9990 & 0.9989 & 99.898 & 1.583 & 1.000 & 0.9989 & 0.9991 & 99.601 & 0.306 & 0.999 \\
    & 3 & 0.9982 & 0.9989 & 98.948 & 0.002 & 0.998 & 0.9992 & 0.9992 & 99.793 & 0.930 & 1.000 & 0.9988 & 0.9992 & 99.423 & 0.053 & 0.999 \\
    & 4 & 0.9984 & 0.9991 & 96.626 & 0.001 & 0.998 & 0.9994 & 0.9993 & 99.752 & 0.071 & 1.000 & 0.9989 & 0.9993 & 99.138 & 0.004 & 0.999 \\
\hline
\multirow{4}{*}{Syntax}
    & 1 & 0.9987 & 0.9987 & 99.513 & 0.003 & 0.999 & 0.9994 & 0.9994 & 99.930 & 0.948 & 1.000 & 0.9991 & 0.9991 & 99.810 & 0.375 & 0.999 \\
    & 2 & 0.9986 & 0.9990 & 99.457 & 0.002 & 0.999 & 0.9992 & 0.9992 & 99.890 & 0.343 & 1.000 & 0.9991 & 0.9992 & 99.747 & 0.140 & 0.999 \\
    & 3 & 0.9986 & 0.9990 & 99.190 & 0.002 & 0.999 & 0.9995 & 0.9995 & 99.879 & 0.624 & 1.000 & 0.9991 & 0.9993 & 99.519 & 0.024 & 0.999 \\
    & 4 & 0.9988 & 0.9992 & 96.907 & 0.001 & 0.999 & 0.9996 & 0.9997 & 99.853 & 0.107 & 1.000 & 0.9991 & 0.9993 & 98.788 & 0.006 & 0.999 \\
\hline
\multirow{4}{*}{Sentence}
    & 1 & 0.9999 & 0.9999 & 99.879 & 0.001 & 1.000 & 1.0000 & 1.0000 & 100.000 & 0.134 & 1.000 & 1.0000 & 1.0000 & 99.996 & 0.971 & 1.000 \\
    & 2 & 0.9999 & 0.9999 & 99.865 & 0.001 & 1.000 & 1.0000 & 1.0000 & 100.000 & 0.034 & 1.000 & 1.0000 & 1.0000 & 99.982 & 0.390 & 1.000 \\
    & 3 & 0.9998 & 0.9999 & 99.655 & 0.001 & 1.000 & 1.0000 & 1.0000 & 100.000 & 0.032 & 1.000 & 1.0000 & 1.0000 & 99.909 & 0.082 & 1.000 \\
    & 4 & 0.9998 & 0.9999 & 97.315 & 0.001 & 1.000 & 1.0000 & 1.0000 & 100.000 & 0.001 & 1.000 & 1.0000 & 1.0000 & 99.549 & 0.008 & 1.000 \\
\bottomrule
\end{tabular}
\end{table*}

\begin{table*}[!htbp]
\centering
\setlength{\tabcolsep}{3.1pt}
\caption{Results for attacks against Ember on EmberNN.}
\label{table:ember}
\begin{tabular}{@{}ccSSSSS!{\quad}SSSSS!{\quad}SSSSS@{}}
\toprule
\multicolumn{1}{c}{\multirow{2}{*}{\textbf{Attack}}}
& \multicolumn{1}{c}{\multirow{2}{*}{\shortstack[c]{\textbf{Mal.}\\\textbf{Users}}}}
& \multicolumn{5}{c}{\textbf{kNN (Baseline)}} & \multicolumn{5}{c}{\textbf{Unlearning (\S\ref{sec:unlearning})}}
& \multicolumn{5}{c}{\textbf{Gradient (\S\ref{sec:gradient})}} \\
\cmidrule(lr{1.2em}){3-7} \cmidrule(lr{1.2em}){8-12} \cmidrule(l){13-17}
\multicolumn{1}{c}{}
& \multicolumn{1}{c}{}
& mAP & mRR & {\scriptsize TPR (\%)} & {\scriptsize FPR (\%)} & AUC
& mAP & mRR & {\scriptsize TPR (\%)} & {\scriptsize FPR (\%)} & AUC
& mAP & mRR & {\scriptsize TPR (\%)} & {\scriptsize FPR (\%)} & AUC \\
\hline
\multirow{4}{*}{BadNets}
    & 1 & 0.9865 & 0.9865 & 87.508 & 0.057 & 0.989 & 0.9421 & 0.9421 & 82.766 & 2.475 & 0.964 & 0.9729 & 0.9729 & 93.614 & 0.463 & 0.976 \\
    & 2 & 0.9822 & 0.9876 & 85.380 & 0.035 & 0.984 & 0.9037 & 0.9292 & 52.846 & 1.265 & 0.940 & 0.9756 & 0.9792 & 89.460 & 0.182 & 0.980 \\
    & 3 & 0.9801 & 0.9885 & 82.796 & 0.025 & 0.980 & 0.8499 & 0.8830 & 16.289 & 0.850 & 0.896 & 0.9741 & 0.9816 & 84.190 & 0.064 & 0.980 \\
    & 4 & 0.9808 & 0.9900 & 78.171 & 0.021 & 0.976 & 0.8426 & 0.9087 & 5.451 & 0.759 & 0.867 & 0.9750 & 0.9827 & 77.088 & 0.040 & 0.980 \\
\hline
\multirow{4}{*}{Explanation}
    & 1 & 1.0000 & 1.0000 & 99.998 & 0.000 & 1.000 & 1.0000 & 1.0000 & 100.000 & 1.211 & 0.997 & 1.0000 & 1.0000 & 100.000 & 0.068 & 1.000 \\
    & 2 & 1.0000 & 1.0000 & 99.998 & 0.000 & 1.000 & 1.0000 & 1.0000 & 98.936 & 1.188 & 0.997 & 1.0000 & 1.0000 & 100.000 & 0.000 & 1.000 \\
    & 3 & 1.0000 & 1.0000 & 99.845 & 0.000 & 1.000 & 1.0000 & 1.0000 & 85.619 & 0.792 & 0.997 & 1.0000 & 1.0000 & 100.000 & 0.000 & 1.000 \\
    & 4 & 1.0000 & 1.0000 & 97.626 & 0.000 & 1.000 & 0.9780 & 0.9789 & 72.189 & 1.821 & 0.986 & 1.0000 & 1.0000 & 100.000 & 0.000 & 1.000 \\
\bottomrule
\end{tabular}
\end{table*}

\begin{table*}[!htbp]
\centering
\setlength{\tabcolsep}{3.1pt}
\caption{Results for clean-label attacks against CIFAR-10 on ResNet18.}
\label{table:cifar10-clean}
\begin{tabular}{@{}ccSSSSS!{\quad}SSSSS!{\quad}SSSSS@{}}
\toprule
\multicolumn{1}{c}{\multirow{2}{*}{\textbf{Attack}}}
& \multicolumn{1}{c}{\multirow{2}{*}{\shortstack[c]{\textbf{Mal.}\\\textbf{Users}}}}
& \multicolumn{5}{c}{\textbf{kNN (Baseline)}} & \multicolumn{5}{c}{\textbf{Unlearning (\S\ref{sec:unlearning})}}
& \multicolumn{5}{c}{\textbf{Gradient (\S\ref{sec:gradient})}} \\
\cmidrule(lr{1.2em}){3-7} \cmidrule(lr{1.2em}){8-12} \cmidrule(l){13-17}
\multicolumn{1}{c}{}
& \multicolumn{1}{c}{}
& mAP & mRR & {\scriptsize TPR (\%)} & {\scriptsize FPR (\%)} & AUC
& mAP & mRR & {\scriptsize TPR (\%)} & {\scriptsize FPR (\%)} & AUC
& mAP & mRR & {\scriptsize TPR (\%)} & {\scriptsize FPR (\%)} & AUC \\
\hline
\multirow{4}{*}{\shortstack[c]{Witches'\\Brew}}
    & 1 & 0.9898 & 0.9898 & 14.286 & 0.454 & 0.958 & 0.9847 & 0.9847 & 85.714 & 0.680 & 0.987 & 1.0000 & 1.0000 & 89.796 & 0.680 & 0.996 \\
    & 2 & 0.9361 & 0.9796 & 8.163 & 0.765 & 0.926 & 0.9001 & 0.9152 & 59.184 & 0.255 & 0.947 & 0.9728 & 1.0000 & 40.816 & 0.255 & 0.981 \\
    & 3 & 0.8986 & 0.9898 & 6.803 & 0.583 & 0.859 & 0.8712 & 0.9190 & 21.088 & 0.875 & 0.911 & 0.9495 & 1.0000 & 10.884 & 0.000 & 0.954 \\
    & 4 & 0.8845 & 0.9694 & 3.061 & 0.340 & 0.843 & 0.8619 & 0.9422 & 5.612 & 0.340 & 0.861 & 0.9257 & 0.9898 & 5.612 & 0.000 & 0.921 \\
\hline
\multirow{4}{*}{\shortstack[c]{Sleeper\\Agent}}
    & 1 & 0.7578 & 0.7578 & 25.691 & 0.620 & 0.827 & 0.9815 & 0.9815 & 93.552 & 0.489 & 0.988 & 0.9110 & 0.9110 & 78.864 & 0.677 & 0.938 \\
    & 2 & 0.7210 & 0.7876 & 22.492 & 0.640 & 0.760 & 0.9732 & 0.9781 & 86.745 & 0.249 & 0.983 & 0.9107 & 0.9245 & 70.548 & 0.345 & 0.937 \\
    & 3 & 0.7420 & 0.8140 & 20.710 & 0.687 & 0.747 & 0.9717 & 0.9873 & 74.701 & 0.088 & 0.979 & 0.9216 & 0.9415 & 58.888 & 0.146 & 0.944 \\
    & 4 & 0.7695 & 0.8567 & 19.089 & 0.546 & 0.723 & 0.9669 & 0.9839 & 55.566 & 0.051 & 0.972 & 0.9343 & 0.9574 & 36.592 & 0.136 & 0.949 \\
\bottomrule
\end{tabular}
\end{table*}

\begin{myhideenv}
\section{Additional Details on Secure \sysname}
\label{appx:secure}
In this appendix, we provide formal definitions of the ideal functionalities, the security proofs for our protocols, and discuss further optimizations.

\subsection{Definitions}
\label{appx:definitions}

We provide a minimal definition of a \gls{lsss} used in our extension of Arc's auditing phase with preprocessing.

\begin{definition}[Secret-Sharing Scheme]
    \label{def:lsss}
   A \Acrfull{lsss} enables one party, the dealer, to share a secret $s$ among $n$ parties such that any $t$ parties can reconstruct the secret $s$.
   We define a \gls{lsss} by a tuple of protocols $(\SSShare, \SSReconstruct)$ defined as follows:
   \begin{algos}
       \item $\SSShare(s) \rightarrow \{\secs{s}_1, \secs{s}_2, \ldots, \secs{s}_n\}$: The dealer sends a share of its secret $\secs{s}$ to each of the $n$ parties.
       \item $\SSReconstruct(\{\secs{s}_{1}, \secs{s}_{2}, \ldots, \secs{s}_{t}\}) \rightarrow s$: A protocol in which at least $t$ parties input their share $\secs{s}$ to reconstruct $s$.
   \end{algos}

\end{definition}
A valid secret-sharing scheme satisfies correctness, meaning that any $t$ parties are able to recover the original secret $s$, and privacy, meaning that any $t-1$ parties learn no information about $s$.
In addition, the linearity property implies that the shares can be added together to obtain a share of the sum of the secrets, i.e., $\secs{s_1} + \secs{s_2} = \SSShare(s_1 + s_2)$.

\myparagraph{Ideal Functionalities}\label{sec:appx:mpc:idealfun}
We describe the ideal functionalities \idealMul, \idealTrunc, \idealDotprod, \idealDotprodNoTrunc, \idealSort, \idealGradient, \idealSGD, and \idealRecip,  referenced in our protocols.

\noindent
{\setlength{\fboxsep}{1em}%
    \refstepcounter{func}
    \begin{functionalitybox*}{\idealMul}{
        The functionality receives inputs $\secs{a}$ and $\secs{b}$.
    }
        Compute the following \label{ideal:mul}
    \begin{enumerate}
        \item Reconstruct $a$ and $b$ to compute $ab^\prime \leftarrow a \cdot b$.
        \item Truncate the lower $f$ bits of $ab^\prime$ to get $ab$.
        \item Distribute shares \secs{ab} to the parties.
    \end{enumerate}
    \end{functionalitybox*}
    \refstepcounter{func}
    \begin{functionalitybox*}{\idealTrunc}{The functionality receives input $\secs{x}$.}
        Compute the following
        \begin{enumerate}
            \item Reconstruct $x$ and truncate the lower $f$ bits to get $x^\prime$.
            \item Distribute shares $\secs{x^\prime}$ to the parties.
        \end{enumerate}
    \end{functionalitybox*}
    \refstepcounter{func}
    \begin{functionalitybox*}{\idealDotprod}{
        The functionality receives input vectors $\secs{a}$ and $\secs{b}$.
    }
        Compute the following \label{ideal:dotprod}
    \begin{enumerate}
        \item Reconstruct $a$ and $b$ to compute the dot product $c^\prime \leftarrow a \cdotp b$.
        \item Truncate the lower $f$ bits of $c^\prime$ to get $c$.
        \item Distribute shares \secs{c} to the parties.
    \end{enumerate}
    \end{functionalitybox*}
    \refstepcounter{func}
    \begin{functionalitybox*}{\idealDotprodNoTrunc}{
        The functionality receives input vectors $\secs{a}$ and $\secs{b}$.
    }
        Compute the following \label{ideal:dotprodnotrunc}
    \begin{enumerate}
        \item Reconstruct $a$ and $b$ to compute the dot product $c \leftarrow a \cdotp b$.
        \item Distribute shares \secs{c} to the parties.
    \end{enumerate}
    \end{functionalitybox*}
    \refstepcounter{func}
    \begin{functionalitybox*}{\idealSort}{The functionality receives input matrix $\secs{\mat{M}}$.}
        Compute the following
        \begin{enumerate}
            \item Reconstruct $\mat{M}$ and sort \mat{M} according to the values in the first column $\mat{M}^{(1)}$ to get $\mat{M}^\prime$.
            \item Distribute shares $\secs{\mat{M}^\prime}$ to the parties.
        \end{enumerate}
    \end{functionalitybox*}
    \refstepcounter{func}
    \label{fun:gradient}
    \begin{functionalitybox*}{\idealGradient}{The functionality receives input samples $\secs{D}$,
        model parameters $\secs{\theta}$, and projection matrix $\secs{G}$.}
        Compute the following
        \begin{enumerate}
            \item Reconstruct $D$, $\theta$, and $G$, and compute the (averaged) gradient $\nabla \theta$ on $D$ for model parameters $\theta$.
            \item Compute the lower-dimensional projection $g \leftarrow G \cdot \nabla \theta$.
            \item Distribute shares $\secs{g}$ to the parties.
        \end{enumerate}
    \end{functionalitybox*}
    \refstepcounter{func}
    \begin{functionalitybox*}{\idealRecip}{The functionality receives input $\secs{x}$.}
        Compute the following.
        \begin{enumerate}
            \item Reconstruct $x$ and compute $x^\prime \leftarrow 1 / {\sqrt{x}}$.
            \item Distribute shares $\secs{x^\prime}$ to the parties.
        \end{enumerate}
    \end{functionalitybox*}
    \refstepcounter{func}
    \begin{functionalitybox*}{\idealSGD}{The functionality receives input $\mathcal{L}, \secs{\smodel}, \secs{\spartyinput_{i}}, \secs{\spartyinput_{-i}}, E$.}
        Compute the following. \label{ideal:sgd}
        \begin{enumerate}
            \item Reconstruct $\smodel, \spartyinput_{i}, \spartyinput_{-i}$ and compute the updated model $\smodel^\prime$
            by applying stochastic gradient descent updates to the model $\smodel$
            that minimize the objective function $\mathcal{L}$ for $E$ epochs.
            \item Distribute shares $\secs{\smodel^\prime}$ to the parties.
        \end{enumerate}
    \end{functionalitybox*}

    \refstepcounter{func}
    \begin{functionalitybox*}{\idealLoss \\
    The functionality is parameterized by a loss function $\ell$}{The functionality receives input $\secs{\spredictionY}, \secs{\spredictionY^\prime}$.}
        Compute the following. \label{ideal:loss}
        \begin{enumerate}
            \item Reconstruct $\spredictionY, \spredictionY^\prime$ and compute the loss $l \leftarrow \ell(\spredictionY, \spredictionY^\prime)$.
            \item Distribute shares $\secs{l}$ to the parties.
        \end{enumerate}
    \end{functionalitybox*}

    }

\end{myhideenv}

\begin{myhideenv}

    \myparagraph{Algorithm Functionalities}
    We specify the ideal functionalities for \Cref{alg:traceback,alg:unlearning}. Each algorithm is split into a preprocessing functionality and an online functionality. In the outputs, we use the bracket notation $\secs{\cdot}$ to indicate that the functionality distributes secret shares of the enclosed value to the parties.
            \newcommand{\RETURN}{\State \textbf{return}\xspace}

\refstepcounter{func}
\begin{functionalitybox*}{\idealGradPreprocessing
  \label{functionality:grad:preprocessing}
}{
  The functionality receives inputs $\secpartyinputs, \secs{R}$,
  where $R = \{(\theta_{t-1}, \eta_t, G_t)\}_{t \in \cT}$.%
}
  Reconstruct $\spartyinputs$ from $\secpartyinputs$, 
  and $R$ from $\secs{R}$, then compute:%
  \label{ideal:grad:preprocessing}

  \begin{algorithmic}[1]
    \State $\Dtr \gets \textsc{Concat}(\spartyinput_{1}, \ldots, \spartyinput_{m})$
    \For{$t \in \cT$}
      \For{$i \gets 1$ to $\abs{\Dtr}$}
        \State $(x_i, y_i) \gets \Dtr[i]$
        \State $g_t^{(i)} \gets G_t \, \nabla_{\smodelW}\, \ell(\theta_{t-1}; x_i, y_i)$
      \EndFor
    \EndFor
    \RETURN $\{\secs{g_t^{(i)}}\}_{t\in\cT,i\in\abs{\Dtr}}$
  \end{algorithmic}
\end{functionalitybox*}

\refstepcounter{func}
\begin{functionalitybox*}{\idealGradOnline\\
  Parameterized by a loss function $\ell$.
  \label{functionality:grad:online}
}{
  The functionality receives inputs
$\secs{\spredictionX}, \secs{\spredictionY}$, $\secs{g_t}=\{\secs{g_t}\}_{t\in\cT}$, $\secs{R} =\{(\secs{\theta_{t-1}},\eta_t,\secs{G_t})\}_{t\in\cT}$,
  and score parameter $k$.%
}
  Reconstruct $\spredictionX$ from $\secs{\spredictionX}$, $\spredictionY$ from $\secs{\spredictionY}$,
  $g_t$ from $\secs{g_t}$
  and $R$ from $\secs{R}$.
  Send $\eta_t$ to \sadversary.
  Then compute:%
  \label{ideal:grad:online}

  \begin{algorithmic}[1]
    \State $\Dtr \gets \textsc{Concat}(\spartyinput_{1}, \ldots, \spartyinput_{m})$
    \For{$t \in \cT$}
      \State $\widehat{g_t} \gets G_t \nabla_{\smodelW}\, \ell(\theta_{t-1}; \spredictionX)$
    \EndFor
    \For{$i \gets 1,\ldots,m$}
      \State $I_i \gets \textsc{Indices}(i)$
              \State $S_i \gets \left\{\sum_{t \in \cT} \eta_t \frac{\iprod{g_t^{(\mathrm{idx})}}{\widehat{g_t}}}{\norm*{g_t^{(\mathrm{idx})}}_2 \norm*{\widehat{g_t}}_2} : \text{idx} \in I_i \right\}$
      \State\lstep{grad:online:sim} $s_i \gets \frac{1}{k} \sum_{j=1}^{k} \topk{k}(S_i)_j$
    \EndFor
    \State $\secs{O_i}_{i\in \abs{m}} \gets \textsc{RankUsers}(s_1,\ldots,s_m)$
    \RETURN $\secs{O_1},\ldots,\secs{O_m}$
  \end{algorithmic}
\end{functionalitybox*}

  \refstepcounter{func}
\begin{functionalitybox*}{\idealCamelPreprocessing\\
  The functionality is parameterized by an unlearning function \Unl}{
  The functionality receives inputs $\secpartyinputs, \secs{\smodelW}, E$.%
  \label{functionality:camel:preprocessing}
}
  Reconstruct $\spartyinputs$ from $\secpartyinputs$ and $\smodelW$ from $\secs{\smodelW}$.
  Send $E$ to \sadversary and compute the following.%
  \label{ideal:camel:preprocessing}

  \begin{algorithmic}[1]
    \For{$i \gets 1,\ldots,N$}
      \State $\spartyinput_{-i} \gets \displaystyle\bigcup_{j\neq i} \spartyinput_{j}$
      \State $\smodelW_{-i} \gets \Unl\!\bigl(\smodelW,\ \spartyinput_{i},\ \spartyinput_{-i},\ E\bigr)$
    \EndFor
    \RETURN $\{\secs{\smodelW_{-1}},\ldots,\secs{\smodelW_{-N}}\}$
  \end{algorithmic}
\end{functionalitybox*}

\refstepcounter{func}
\begin{functionalitybox*}{\idealCamelOnline\\
  Parameterized by loss function $\ell$.%
  \label{functionality:camel:online}
}{
  The functionality receives inputs
  $\secpartyinputs$, $\secs{\smodel}, \secs{\spredictionX}, \secs{\spredictionY}$, $\{\secs{\smodelW_{-1}}, \ldots, \secs{\smodelW_{-N}}\}$.
}
  Reconstruct $\spartyinputs$ from $\secpartyinputs$, $\smodelW$ from $\secs{\smodelW}$,
  \spredictionX from \secs{\spredictionX}, \spredictionY from \secs{\spredictionY},
  and $\smodelW_{-1},\ldots,\smodelW_{-N}$ from
  $\secs{\smodelW_{-1}},\ldots,\secs{\smodelW_{-N}}$, then compute:%
  \label{ideal:camel:online}

  \begin{algorithmic}[1]
    \For{$i \gets 1,\ldots,m$}
      \State\lstep{camel:online:inf} $s_i \gets \ell\!\bigl(\smodelW_{-i}(\spredictionX),\ \spredictionY\bigr)$
    \EndFor
    \State $\secs{O_1},\ldots,\secs{O_m} \gets \textsc{RankUsers}(s_1,\ldots,s_m)$
    \RETURN $\secs{O_1},\ldots,\secs{O_m}$
  \end{algorithmic}
\end{functionalitybox*}

\subsection{Security Proofs}
\label{appx:proofs}
In the following, we provide a proof that $\protocolCamelPreprocessing$ securely instantiates $\idealCamelPreprocessing$,
and that $\protocolCamelOnline$ securely instantiates $\idealCamelOnline$.
Our proof follows the real/ideal-world paradigm~\cite{Canetti2000-plainmodel}, which considers the following two worlds:
\begin{algos}
    \item \textbf{In the real world,} the parties run the actual protocol \sprotocol.
    The adversary \sadversary can statically, actively, corrupt a subset of parties before the start of the protocol.

    \item \textbf{In the ideal world,} the honest parties send their inputs to the ideal functionality \ideal,
    which executes the behavior of a secure version of \sysname's protocols.
    The ideal world defines the ideal behavior of the functionality that the protocol aims to emulate.
\end{algos}
A real-world protocol is secure if it manages to instantiate an ideal functionality in the ideal world.
To show that a protocol is secure, we must show that the adversary cannot distinguish between the real and the ideal world with high probability.
We can do this by defining a non-uniform probabilistic polynomial-time simulator $\ssim$ that interacts with the adversary \sadversary and the ideal functionality $\ideal$ in the ideal world
in such a way that \sadversary's view is indistinguishable when interacting with the protocol in the real world.

\myparagraph{Sequential Composition}
We model protocols for MPC operations such as multiplication and truncation in our protocol as ideal functionalities \ideal~(c.f.~\Cref{sec:appx:mpc:idealfun}).
According to the sequential composition theorem, if a protocol securely computes a functionality in the $\ideal_\textsf{G}$-hybrid model for some functionality $\ideal_\textsf{G}$,
then it remains secure when composed with a protocol that securely computes $\ideal_\textsf{G}$~\cite{Canetti2000-plainmodel}.
We use this model to abstract the dependencies of our protocols.
In the security proof, the simulator $\ssim$ internally simulates these ideal functionalities towards the adversary \sadversary.
As a result, the simulator can directly receive the inputs of the adversary to these functionalities.

We are now ready to prove the security of our protocols \protocolCamelPreprocessing, \protocolCamelOnline, \protocolOursPreprocessing, and \protocolOursOnline.
The proofs proceed largely in a similar fashion,
by defining a simulator $\ssim$ through a series of subsequent modifications to the real
    execution, so that the views of $\sadversary$ in any two subsequent executions are computationally indistinguishable.
    Without loss of generality, we assume that if the simulator receives inconsistent values from some of the parties that should be the same according to the real protocol, the simulator aborts.

\begin{theorem}[Security of \protocolOursPreprocessing]
\label{thm:gradient:preprocessing}
Given a set of $\sNumParties$ parties,
an adversary \sadversary who controls a set $\sPartyCorrupted$ of at most $\sNumParties - 1$ corrupted parties.
There exists a \PPT simulator $\ssim$ in the $(\idealGradient)$-hybrid model such that the distributions:
\begin{equation*}
    \begin{gathered}
        \{ \texttt{Ideal}_{\protocolOursPreprocessing, \ssim(\advAux), \sPartyCorrupted}
        \bigl(\secspartyinputsCondensed, \secs{\srecord}, \lambda\bigr)
        \\ \}_{\secspartyinputsCondensed, \secs{\srecord}, \advAux, \lambda}
        \\
        \approx
        \\
        \{ \real_{\idealGradPreprocessing, \sadversary(\advAux), \sPartyCorrupted}
        \bigl(\secspartyinputsCondensed, \secs{\srecord}, \lambda\bigr)
        \\ \}_{\secspartyinputsCondensed, \secs{\srecord}, \advAux, \lambda}
    \end{gathered}
\end{equation*}
\end{theorem}

\begin{proof}
    We define a simulator $\ssim$ through a series of subsequent modifications to the real execution.

    \begin{hybrid}
        \item The view of $\sadversary$ in this hybrid is distributed exactly as the view of $\sadversary$ in $\real$.
        \item In this hybrid, the real execution is emulated by a simulator that knows the real inputs of the parties $\secs{\spartyinput_i},\secs{R}$
        and runs a full execution of the protocol with $\sadversary$, which includes emulating the ideal interaction with \idealGradient.
        The view of the adversary in this hybrid is the same as the previous one.

        \item In this hybrid, we replace the secret-shared inputs to the ideal functionality on behalf of the parties controlled by \sadversary
        with inputs extracted from \sadversary.
        Through internally simulating \idealGradient towards \sadversary using \ssimGradient, the simulator directly receives the malicious secret shares of the inputs and the training record $\secs{\spartyinput_1}_j,\ldots,\secs{\spartyinput_N}_j,\secs{R}_j$ for $j \in \Cset$,
        which it forwards to the ideal functionality.
        The view of the adversary in this hybrid is the same as the previous one because of sequential composition.

        \item In this hybrid, we replace the real inputs used by \ssim with dummy inputs.
        The simulator generates randomly sampled secret shares for the honest parties' inputs $\secs{\spartyinput_1}_j,\ldots,\secs{\spartyinput_N}_j,\secs{R}_j \sample \sring$ for $j \notin \Cset$.
        The view of the adversary in this hybrid is indistinguishable from the previous one, due to the privacy property of the secret sharing scheme.
        Note that in Step~\ref{grad:protocol:gradient}, the simulator still uses the real outputs of \idealGradient in its simulation towards \sadversary.

        \item In this hybrid, we replace the real output to the \sadversary with the output of the ideal functionality.
        The simulator $\ssim$ receives the output of the ideal functionality 
        $\{\secs{g_t^\prime}\}_{t\in\cT}$
        and forwards it to the parties controlled the adversary in Step~\ref{grad:protocol:gradient}.
        The ideal output, $\{\secs{g_t^\prime}\}_{t\in\cT}$, is the gradient of the sample with respect to the model checkpoint, whereas the real outputs are computed via \idealGradient, which computes the same~(c.f.~Functionality~\ref{fun:gradient}).
        Therefore, the outputs of the ideal functionality and the real protocol are the same, and the view of the adversary is indistinguishable from its view in the previous hybrid.

        \item This hybrid is defined as the previous one, with the only difference being that the simulator now does not receive the inputs of the honest parties.
        Because the simulator no longer relies on receiving inputs from the honest parties, the view of the adversary is perfectly indistinguishable from the previous hybrid.

        \end{hybrid}
\end{proof}

\begin{theorem}[Security of \protocolOursOnline]
\label{thm:gradient:online}
Given a set of $\sNumParties$ parties,
an adversary \sadversary who controls a set $\sPartyCorrupted$ of at most $\sNumParties - 1$ corrupted parties.
There exists a \PPT simulator $\ssim$ in the $(\idealGradient,\idealDotprod,\idealMul,\idealRecip,\idealSort)$-hybrid model such that the distributions:
\begin{equation*}
    \begin{gathered}
        \{ \texttt{Ideal}_{\protocolOursOnline, \ssim(\advAux), \sPartyCorrupted}
        \bigl(\secs{\spredictionX}, \secs{\spredictionY},
        \secs{\sGset_t}, \secs{\srecord}, k, \lambda\bigr)
        \\ \}_{\secs{\spredictionX}, \secs{\spredictionY},
        \{\secs{\sGset_t}, \secs{\srecord}, k, \advAux, \lambda}
        \\
        \approx
        \\
        \{ \real_{\idealGradOnline, \sadversary(\advAux), \sPartyCorrupted}
        \bigl(\secs{\spredictionX}, \secs{\spredictionY},
        \secs{\sGset_t}, \secs{\srecord}, k, \lambda\bigr)
        \\ \}_{\secs{\spredictionX}, \secs{\spredictionY}, 
        \secs{\sGset_t}, \secs{\srecord}, k, \advAux, \lambda}
    \end{gathered}
\end{equation*}
\end{theorem}

\begin{proof}
We define a simulator $\ssim$ through a series of subsequent modifications to the real execution.

    \begin{hybrid}
        \item The view of $\sadversary$ in this hybrid is distributed exactly as the view of $\sadversary$ in $\real$.

        \item In this hybrid, the real execution is emulated by a simulator that knows the real inputs of the parties $\secs{\spredictionX}, \secs{\spredictionY},
        \secs{\sGset_t},  \secs{\srecord}$
        and runs a full execution of the protocol with $\sadversary$, which includes emulating the ideal interactions with the ideal functionalities.
        The view of the adversary in this hybrid is the same as the previous one.

        \item In this hybrid, we replace the real inputs to the ideal functionality on behalf of the malicious parties
        with inputs extracted from \sadversary.
        The simulator receives \sadversary's secret shares of the input $\secs{\spredictionX}, \secs{\spredictionY}$ and $\{\secs{G_t},\secs{\theta_{t-1}}\}_{t\in\mathcal{T}}$ 
        in~Step~\ref{grad:protocol:online:score} by emulating
        the ideal functionality \idealGradient.
        It also receives \sadversary's secret shares of $\secs{g_t}$ in Steps~\ref{grad:protocol:online:loopstart}-\ref{grad:protocol:online:loopend} through the emulation of \idealDotprod.
        It forwards these to the ideal functionality, along with the public $\eta_t$'s that every client holds. 
        The view of the adversary in this hybrid is the same as the previous one because of sequential composition.

        \item In this hybrid, we replace the real inputs used by \ssim with dummy inputs.
        The simulator generates randomly sampled secret shares for the honest parties' inputs $\secs{\spredictionX}, \secs{\spredictionY},
        \secs{\sGset_t},  \secs{\srecord}$ for $j \notin \Cset$.
        For the public parameters $\eta_t$, it uses the $\eta_t$ received from the ideal functionality.
        The view of the adversary in this hybrid is indistinguishable from the previous one, due to the privacy property of the secret sharing scheme,
        and because none of the secret sharings are opened during the protocol execution.
        Note that in the last step, the simulator still uses the real outputs of \textsc{RankUsers} in its simulation towards \sadversary.

        \item %
        In this hybrid, we replace the real output to the \sadversary with the output of the ideal functionality.
        The simulator $\ssim$ receives the output of the ideal functionality $\secs{O}$, sets the honest clients' shares to $\secs{O}_j$ for $j \notin \Cset$,
        and sends the malicious clients' shares to the adversary as the output of \textsc{RankUsers}.
        The view of the adversary in this hybrid is indistinguishable from the previous one, because the output is identical in both worlds.

        \item This hybrid is defined as the previous one, with the only difference being that the simulator now does not receive the inputs of the honest parties.
        Because the simulator no longer relies on receiving inputs from the honest parties, the view of the adversary is perfectly indistinguishable from the previous hybrid.

    \end{hybrid}
\end{proof}

\begin{theorem}[Security of \protocolCamelPreprocessing]
\label{thm:camel:preprocessing}
Given a set of $\sNumParties$ parties,
an adversary \sadversary who controls a set $\sPartyCorrupted$ of at most $\sNumParties - 1$ corrupted parties.
there exists a \PPT simulator $\ssim$ in the $(\idealSGD)$-hybrid model such that the distributions:
\begin{equation*}
    \begin{gathered}
        \{ \texttt{Ideal}_{\protocolCamelPreprocessing, \ssim(\advAux),\sPartyCorrupted } (\secspartyinputsCondensed, \secs{\smodelW}, E, 
        \lambda) \\ \}_{\secspartyinputsCondensed, \secs{\smodelW}, E, \advAux, \lambda} \\
        \approx \\
        \{ \real_{\idealCamelPreprocessing, \sadversary(\advAux), \sPartyCorrupted } (\secspartyinputsCondensed, \secs{\smodelW}, E, \lambda) \\
        \}_{\secspartyinputsCondensed, \secs{\smodelW}, E, \advAux, \lambda}
    \end{gathered}
\end{equation*}

are computationally indistinguishable, where $\spartyinputsCondensed = \spartyinputs$ is a list of training datasets for each \gls{r:inputparty},
\smodelW is the model, $E$ is the number of epochs, 
$\lambda$ is the security parameter,
and $\advAux \in \{0,1 \}^*$ is an auxiliary input by the adversary to capture malicious strategy.

\end{theorem}

\begin{proof}
    We define a simulator $\ssim$ through a series of subsequent modifications to the real execution.

    \begin{hybrid}
        \item The view of $\sadversary$ in this hybrid is distributed exactly as the view of $\sadversary$ in $\real$.
        \item In this hybrid, the real execution is emulated by a simulator that knows the real inputs of the parties $\secs{\spartyinput_i},\secs{\smodelW}$
        and runs a full execution of the protocol with $\sadversary$, which includes emulating the ideal interaction with \idealSGD.
        The view of the adversary in this hybrid is the same as the previous one.

        \item In this hybrid, we replace the secret-shared inputs to the ideal functionality on behalf of the parties controlled by \sadversary
        with inputs extracted from \sadversary.
        Through internally simulating \idealSGD towards \sadversary using \ssimSGD, the simulator directly receives the malicious secret shares of the inputs and the model $\secs{\spartyinput_1}_j,\ldots,\secs{\spartyinput_N}_j,\secs{\smodelW}_j$ for $j \in \Cset$,
        which it forwards to the ideal functionality.
        The view of the adversary in this hybrid is the same as the previous one because of sequential composition.

        \item In this hybrid, we replace the public input to the ideal functionality with the input from the protocol.
        The simulator receives the public input $E$ from the ideal functionality.
        The view of the adversary in this hybrid is indistinguishable from the previous one, because the public input $E$ is the same as in the real protocol.

        \item In this hybrid, we replace the real inputs used by \ssim with dummy inputs.
        The simulator generates randomly sampled secret shares for the honest parties' inputs $\secs{\spartyinput_1}_j,\ldots,\secs{\spartyinput_N}_j,\secs{\smodelW}_j \sample \sring$ for $j \notin \Cset$.
        The view of the adversary in this hybrid is indistinguishable from the previous one, due to the privacy property of the secret sharing scheme.
        Note that in Step~\ref{camel:preprocessing:sgd}, the simulator still uses the real outputs of \idealSGD in its simulation towards \sadversary.

        \item In this hybrid, we replace the real output to the \sadversary with the output of the ideal functionality.
        The simulator $\ssim$ receives the output of the ideal functionality $\left\{\secs{\smodelW^\prime_{-1}}, \ldots, \secs{\smodelW^\prime_{-N}}\right\}$ and forwards it to the parties controlled the adversary in Step~\ref{camel:preprocessing:sgd}.
        We now argue why the view of the adversary in this hybrid is indistinguishable from the previous one.

        The ideal outputs, $\left\{\secs{\smodelW^\prime_{-1}}, \ldots, \secs{\smodelW^\prime_{-N}}\right\}$, are derived using the unlearning function \Unl, whereas the real outputs are computed via \idealSGD with the objective $\mathcal{L}$~(c.f.~Step~\ref{camel:preprocessing:sgd}).
        As shown in Equation~\ref{eq:objective}, the objective $\mathcal{L}$ matches that of the unlearning function \Unl.
        Therefore, the outputs of the ideal functionality and the real protocol are the same, and the view of the adversary is indistinguishable from its view in the previous hybrid.

        \item This hybrid is defined as the previous one, with the only difference being that the simulator now does not receive the inputs of the honest parties.
        Because the simulator no longer relies on receiving inputs from the honest parties, the view of the adversary is perfectly indistinguishable from the previous hybrid.

        \end{hybrid}
\end{proof}

\begin{theorem}[Security of \protocolCamelOnline]
    \label{thm:camel:online}
    Given a set of $\sNumParties$ parties,
    an adversary \sadversary who controls a set $\sPartyCorrupted$ of at most $\sNumParties - 1$ corrupted parties.
    there exists a \PPT simulator $\ssim$ in the $(\idealLoss)$-hybrid model such that the distributions:
    \begin{equation*}
        \begin{gathered}
            \{ \texttt{Ideal}_{\protocolCamelOnline, \ssim(\advAux),\sPartyCorrupted } (\secs{\spredictionX},\secs{\spredictionY}, \scache,
            \lambda) \\ \}_{\secs{\spredictionX},\secs{\spredictionY}, \scache, \advAux, \lambda} \\
            \approx \\
            \{ \real_{\idealCamelOnline, \sadversary(\advAux), \sPartyCorrupted } (\secs{\spredictionX},\secs{\spredictionY}, \scache,
            \lambda) \\ \}_{\secs{\spredictionX},\secs{\spredictionY}, \scache, \advAux, \lambda}
        \end{gathered}
    \end{equation*}

    are computationally indistinguishable, where $\spartyinputsCondensed = \spartyinputs$ is a list of training datasets for each \gls{r:inputparty},
    \smodelW is the model, $E$ is the number of epochs, 
    $\scache = \secs{\smodelW_{-1}}, \ldots, \secs{\smodelW_{-N}}$ is the set of unlearned models,
    $\lambda$ is the security parameter,
    and $\advAux \in \{0,1 \}^*$ is an auxiliary input by the adversary to capture malicious strategy.

\end{theorem}

\begin{proof}
    We define a simulator $\ssim$ through a series of subsequent modifications to the real execution.

    \begin{hybrid}
        \item The view of $\sadversary$ in this hybrid is distributed exactly as the view of $\sadversary$ in $\real$.

        \item In this hybrid, the real execution is emulated by a simulator that knows the real inputs of the parties $\secs{\spredictionX},\secs{\spredictionY},\secs{\smodelW_{-1}}, \ldots, \secs{\smodelW_{-N}}$
        and runs a full execution of the protocol with $\sadversary$, which includes emulating the ideal interactions with the ideal functionalities.
        The view of the adversary in this hybrid is the same as the previous one.

        \item In this hybrid, we replace the real inputs to the ideal functionality on behalf of the malicious parties
        with inputs extracted from \sadversary.
        The simulator receives the inputs $\secs{\spredictionX},\secs{\spredictionY},\secs{\smodelW_{-1}}, \ldots, \secs{\smodelW_{-N}}$ from \sadversary in~Step~\ref{camel:online:score} by emulating
        the ideal functionality \idealLoss.
        The view of the adversary in this hybrid is the same as the previous one because of sequential composition.

        \item In this hybrid, we replace the real inputs used by \ssim with dummy inputs.
        The simulator generates randomly sampled secret shares for the honest parties' inputs $\secs{\spredictionX},\secs{\spredictionY},\secs{\smodelW_{-1}}, \ldots, \secs{\smodelW_{-N}}$ for $j \notin \Cset$.
        The view of the adversary in this hybrid is indistinguishable from the previous one, due to the privacy property of the secret sharing scheme,
        and because none of the secret sharings are opened during the protocol execution.
        Note that in the last step, the simulator still uses the real outputs of \textsc{RankUsers} in its simulation towards \sadversary.

        \item %
        In this hybrid, we replace the real output to the \sadversary with the output of the ideal functionality.
        The simulator $\ssim$ receives the output of the ideal functionality $\secs{O}$, sets the honest clients' shares to $\secs{O}_j$ for $j \notin \Cset$,
        and sends the malicious clients' shares to the adversary as the output of \textsc{RankUsers}.
        The view of the adversary in this hybrid is indistinguishable from the previous one, because the output is identical in both worlds.

        \item This hybrid is defined as the previous one, with the only difference being that the simulator now does not receive the inputs of the honest parties.
        Because the simulator no longer relies on receiving inputs from the honest parties, the view of the adversary is perfectly indistinguishable from the previous hybrid.

    \end{hybrid}
\end{proof}

\end{myhideenv}

\begin{myhideenv}
    
\begin{algorithm}[htbp]
\caption{User-level Traceback Procedure with Heuristic Sample Selection}
\label{alg:traceback-heuristic}
\begin{algorithmic}[1]
    \Input
    User datasets $D_1, \ldots, D_m$,
    Misclassification event $\tilde{x},\tilde{y}$,
    Training record $R = \{(\theta_{t-1}, \eta_t, G_t)\}_{t \in \cT}$,
    Score parameter $k$
    \Ensure
    Ranked owners by responsibility scores $\mathcal{U}$

    \State \Comment{Compute dataset gradients $g_t$ (omitted)}
    \State \Comment{Pre-compute projected traceback gradients}    
    \For{$t \in \cT$}
        \State $\widehat{g_t} \gets G_t \nabla_{\theta_W} \ell(\theta_{t-1}; \tilde{x})$ \Comment{$G_t, \theta_{t-1}$ from $R$}
    \EndFor
    \vspace{0.5em}
    
    \For{$i = 1, \ldots, m$}
        \State \Comment{Extract $l$ relevant samples}
        \State $\mathbf{h} \gets \mathbf{0}$
        \State $I_i \gets \textsc{Indices}(i)$ \Comment{User $i$'s indices}
        \For{$j = 1, \ldots, \abs{I_i}$}
            \State $\mathbf{h}^{(j)} \gets \sum_{t \in \cT} \eta_t \iprod{g_t^{(I_i^{(j)})}}{\widehat{g_t}}$
        \EndFor
        \State $\widetilde{I_i} \gets \topk{l}_{\mathbf{h}}(I_i)$
        \vspace{0.5em}

        \State \Comment{Compute cosine scores}
        \State ${S_i} \gets \left\{\sum_{t \in \cT} \eta_t \frac{\iprod{g_t^{(\mathrm{idx})}}{\widehat{g_t}}}{\norm*{g_t^{(\mathrm{idx})}}_2 \norm*{\widehat{g_t}}_2} : \text{idx} \in \widetilde{I_i} \right\}$
        \vspace{0.25em}
        
        \State ${s_i} \gets \frac{1}{k} \sum_{j=1}^k \topk{k}({S_i})_j$
    \EndFor
    \State \Return $\textsc{RankUsers}({s_1}, \ldots, {s_m})$
\end{algorithmic}
\end{algorithm}

{
\floatname{algorithm}{Protocol}
\setlength\arraycolsep{1pt}
\begin{algorithm}[htbp]
\caption{\protocolOursHeuristicOnline}
\label{protocol:traceback-heuristic}
\begin{algorithmic}[1]
    \Input
        Misclassification event $\secs{\tilde{x}},\secs{\tilde{y}}$,
        Pre-processed training gradients $\{\secs{g_t^{(i)}}\}_{t\in\cT,i\in\abs{\Dtr}}$, 
        Training record $\secs{R}=\{(\secs{\theta_{t-1}},\eta_t,\secs{G_t})\}_{t\in\cT}$,
        Score parameters $k,l$
    \Ensure
        Ranked owners by responsibility scores $\{i:U_i\}$
    
    \State \Comment{Pre-compute projected traceback gradients}    
    \For{$t \in \cT$}
                \State Run 
        $\secs{\widehat{g_t}}
           \gets\idealGradient(\secs{G_t},\secs{\theta_{t-1}},\secs{\tilde{x}},\secs{\tilde{y}})$
    \EndFor
    \vspace{2em}
    
    \For{$i = 1, \ldots, m$}
        \State \Comment{Extract $l$ relevant samples}
        \State $\mathbf{h} \gets \mathbf{0}$
        \State $I_i \gets \textsc{Indices}(i)$ \Comment{User $i$'s indices}
        \For{$t \in \cT$}
        \State Run $\secs{\text{B}_t} \gets \idealDotprod(\widehat{g_t},\widehat{g_t})$
        \For{$j = 1, \ldots, \abs{I_i}$}
            \State Run $\secs{\mathbf{d}^{(j)}_t} \gets \idealDotprod(\secs{g_t^{(I_i^{(j)})}},\secs{\widehat{g_t}})$
            \State Run $\secs{\text{A}^{(j)}_t} \gets \idealDotprodNoTrunc(g_t^{\mathrm{(I_i^{(j)})}},g_t^{\mathrm{(I_i^{(j)}})})$
        \EndFor
        
        \EndFor
        \State $\secs{\mathbf{d}^{(j)}} \gets \sum_{t \in \cT} \eta_t \cdot \secs{\mathbf{d}^{(j)}_t}$ for $j \in [\abs{I_i}]$

        \State $\mat{M} \gets \begin{pmatrix}
    \secs{\mathbf{d}^{(1)}} & \left\{\secs{\mathbf{d}^{(1)}_{t}}\right\}_{t \in \cT} & \left\{ \secs{{\text{A}}^{(1)}_{t}} \right\}_{t \in \cT} \\
    \secs{\mathbf{d}^{(2)}}       & \left\{\secs{\mathbf{d}^{(2)}_{t}}\right\}_{t \in \cT} & \left\{ \secs{{\text{A}}^{(2)}_{t}} \right\}_{t \in \cT} \\
    \hdotsfor{3} \\
    \secs{\mathbf{d}^{(\abs{I_i})}}     & \left\{\secs{\mathbf{d}^{(\abs{{I}_i})}_{t}}\right\}_{t \in \cT} & \left\{ \secs{{\text{A}}^{(\abs{{I}_i})}_{t}} \right\}_{t \in \cT}\end{pmatrix}$
    \State Run $\mat{P} \gets \idealSort(\mat{M})$
    \State $\mat{M}_l \gets \mat{P}\left[0\ldots l \right]$
        \vspace{2em}

        \State \Comment{Compute cosine scores}
    \For{$j = 1,\ldots,l$}
    \For{$t \in \cT$}
        \State Run $\secs{{S_{t}^{(j)}}} \gets  \eta_t \cdot \idealMul(\secs{\mathbf{d}^{(j)}_{t}},
        \idealMul(\idealRecip(\idealTrunc(\secs{\text{A}^{(j)}_{t}}))),\idealRecip(\secs{\text{B}_{t}})))$
    \EndFor
    \State ${S^{(j)}} \gets \sum_{t \in \cT} \eta_t \cdot \secs{{S_{t}^{(j)}}}$
    \EndFor
        
        \State ${s_i} \gets \frac{1}{k} \sum_{j=1}^k \topk{k}({S_i})_j$
    \EndFor
    \State \Return $\textsc{RankUsers}({s_1}, \ldots, {s_m})$
\end{algorithmic}
\end{algorithm}
}

\subsection{Additional Optimizations for \sysgradient}
\label{appx:mpc_opt}
In this appendix, we discuss additional optimizations for \sysgradient, presented in~\Cref{protocol:ours:preprocessing,protocol:ours:online} and its heuristic version, presented in Protocol~\ref{protocol:traceback-heuristic}.

Realizing the traceback algorithm within the PPML setting presents significant challenges related to scalability and numerical precision, primarily due to the representation of values and the computational overhead of cryptographic operations.
    The main overhead of realizing \Cref{alg:traceback} as \sysgradient~(c.f.~\Cref{protocol:ours:preprocessing,protocol:ours:online}) arises from computing the gradients,
    computing the cosine distances, and finding the $\topk{k}$ samples.
    While gradient computation is inherently tied to the specific machine learning model and largely mirrors the training process, our focus here lies on optimizing the latter two tasks. Specifically, we address challenges related to numerical precision and propose three key optimizations to enhance the efficiency of the traceback algorithm in a secure computation setting: approximating the square root operation in the cosine computation, applying heuristic sample selection (cf.~\Cref{protocol:traceback-heuristic} obliviously, and delaying truncation of values.

\myparagraph{Precision}
Real numbers are typically not represented in floating-point representation in PPML settings
due to the computational cost associated with emulating floating-point operations.
Instead, they are expressed using fixed-point representation.
    Numbers are encoded with a fixed number of fractional bits.
    Specifically, given a fixed-point number $x$ with $f$ fractional bits, its corresponding field element can be represented as $a \leftarrow x \cdot 2^f \in \mathbb{R}$.
    While addition in this representation maps directly to ring addition, multiplication requires an additional truncation step.
    This is because multiplying two fixed-point numbers $a$ and $b$ yields $ab = xy \cdot 2^{2f}$, shifting the result by $f$ bits.

    This representation makes it challenging to maintain numerical precision in the protocol, particularly during cosine distance computation.
    Calculating cosine similarity requires computing both the inner product of two gradients, $\iprod{g_t^{(\mathrm{idx})}}{\widehat{g_t}}$,
    and the product of their magnitudes, $\norm*{g_t^{(\mathrm{idx})}}_2 \cdot \norm*{\widehat{g_t}}_2$.
    When these gradients are small---such as those corresponding to data points on which the model has overfit---the product of their magnitudes can become exceedingly small, leading to a loss in precision because the fixed-point representation cannot accurately represent the result.
    This issue has affects not only the choice of the ring size but also the optimizations applied in the protocol.
    In particular, reordering operations that are typically commutative can inadvertently exacerbate precision errors

\myparagraph{Optimizing Cosine Computation}
    The primary performance bottleneck in the traceback algorithm lies in evaluating the cosine similarity between training and test gradients.
    This bottleneck arises from the reliance on fixed-point division and square root operations, both of which are significantly more expensive in secure computation settings.
    Integer division and square root protocols such as those based on Goldschmidt and Raphson-Newton iterations~\cite{Aly2019-zz} require significantly more communication rounds compared to other operations, increasing both latency and bandwidth usage.
    A natural approach to simplify the square root operation is to rewrite the divisor as a single square root over the product of two squared norms, i.e., $\norm*{g_t^{(\mathrm{idx})}}_2 \norm*{\widehat{g_t}}_2 = \sqrt{\sum_i (g_{t,i}^{(\mathrm{idx})})^2 \cdot \sum_i (\widehat{g_{t,i}})^2}$ leveraging the commutativity of multiplication and square roots (for positive numbers).
    However, this method inadvertently exacerbates precision issues. Specifically, the product of the squared norms
    can become exceedingly small, demanding higher precision to accurately represent the result.

Instead, we optimize the cosine computation by decomposing the division into two separate divisions by the individual gradient magnitudes:
\begin{equation*}
    \begin{split}
        \frac{\iprod{g_t^{(\mathrm{idx})}}{\widehat{g_t}}}{\norm*{g_t^{(\mathrm{idx})}}_2 \norm*{\widehat{g_t}}_2}
        =
        \frac{\iprod{g_t^{(\mathrm{idx})}}{\widehat{g_t}}}{\sqrt{\sum_i (g_{t,i}^{(\mathrm{idx})})^2}}
        \cdot
        \frac{\iprod{g_t^{(\mathrm{idx})}}{\widehat{g_t}}}{\sqrt{\sum_i (\widehat{g_{t,i}})^2}}
    \end{split}
\end{equation*}
This decomposition allows us to compute the cosine similarity using two reciprocal square root operations of the individual squared gradient magnitudes.
The reciprocal square root appears in various deep learning primitives such as softmax, batch normalization, and optimizers such as Adam and AMSGrad and efficient approximation protocols exist~\cite{Lu2020-nl,Keller2022-quantizedtraining}.
Crucially, this approach preserves the same fixed-point precision requirements as the original division operation.
After computing the cosine distances, we select the $\topk{k}$ from the $d$ distance scores,
by obliviously sorting the scores and selecting the first $k$ samples.

\myparagraph{Oblivious Selection}
A challenge with realizing the heuristic protocol in MPC is that the computation must be data-oblivious.
In particular, the top-$l$ samples that we select with the heuristic cannot influence the subsequent computation, such as array accesses based on $\tilde{I}_i$.
A na\"{i}ve transformation to oblivious computation would entail doing the computation for all entries in $D_i$ and combining the result with an indicator bit based on whether we want to select the entry.
However, this removes the efficiency benefits of the algorithm, by requiring the computation of a cosine similarity score for the full input array.

To overcome this issue, we extend the vector that is sorted for $\topki{l}$ to a matrix that includes the data required for computing the final score $\Icosmean$.
We compute $\Icosmean$, for $j \in [\abs{\tilde{I}_i}]$, as
\begin{equation*}
    \sum_{t \in \cT} \eta_t \cdot \idealMul(\secs{\mathbf{d}^{(j)}_t},
    \idealMul(\idealRecip(
    \secs{\text{A}^{(j)}_t}),\idealRecip(
    \secs{\text{B}_t}))).
\end{equation*}
where $\secs{\mathbf{d}^{(j)}_t}$ is the inner product between the two gradients $\secs{\mathbf{d}^{(j)}_t}$
and $\secs{\text{A}^{(j)}_t}$ is the gradient magnitudes of the training sample and $\secs{\text{B}_t}$ is the gradient magnitudes of the test sample
for each checkpoint $t$.
We extend each row of the matrix $\mat{M}$ with $2\cdot \abs{\cT}$ columns to include $\left\{\secs{\mathbf{d}_{t}}\right\}_{t \in \cT}$ and $\left\{ \secs{\text{A}_{t}} \right\}_{t \in \cT}$ as
\begin{equation*}
    \mat{M} = \begin{pmatrix}
                  \secs{\mathbf{d}^{(1)}} & \left\{\secs{\mathbf{d}^{(1)}_{t}}\right\}_{t \in \cT} & \left\{ \secs{\text{A}^{(1)}_{t}} \right\}_{t \in \cT} \\
                  \secs{\mathbf{d}^{(2)}}       & \left\{\secs{\mathbf{d}^{(2)}_{t}}\right\}_{t \in \cT} & \left\{ \secs{\text{A}^{(2)}_{t}} \right\}_{t \in \cT} \\
                  \hdotsfor{3} \\
                  \secs{\mathbf{d}^{(\abs{I_i})}}     & \left\{\secs{\mathbf{d}^{(\abs{I_i})}_{t}}\right\}_{t \in \cT} & \left\{ \secs{\text{A}^{(\abs{{I}_i})}_{t}} \right\}_{t \in \cT}
    \end{pmatrix}
\end{equation*}
resulting in a $\abs{{I}_i} \times (2\abs{\cT} + 1)$-sized matrix that must be sorted. Sorting a matrix instead of a vector incurs some additional overhead in MPC. In particular, the radix-sorting algorithm used in our implementation incurs a linear communication overhead in the number of columns due to the shuffling sub-protocol, but no additional rounds~\cite{Hamada2014-obliviousradixsort}.

\myparagraph{Delayed Truncation}
Fixed-point representation is used in MPC to represent real numbers inside the ring (or field) $\mathbb{R}$ used in MPC protocols.
Given a fixed-point number $x$ with $f$ fractional bits, we can compute the corresponding field element as $a \leftarrow x \cdot 2^f \in \mathbb{R}$.
Multiplying $a$ with a second fixed-point number $a$ represented as $b \leftarrow x \cdot 2^f$ yields $ab = xy \cdot 2^{2f}$.
Therefore, we must truncate the result after multiplication through some truncation protocol modeled by \idealTrunc to obtain a fixed-point number with $f$ fractional bits.
Truncation is the most expensive part of computing fixed-point multiplications, because it is a non-linear operation that requires (a partial) bit-decomposition of the value~\cite{Keller2022-quantizedtraining,catrina2010secure}.

Our optimization is based on the observation that the sorting matrix $\mat{M}$ contains many elements that we will never use because they will not be selected as part of the $\topki{l}$ approximation. Therefore, we can delay the truncation of the elements in $\mat{M}$ until after the selection of the $\topki{l}$ elements. Instead of directly including $\secs{\text{N}^{(j)}_{t}}$ computed by \idealMul, we include $\secs{\tilde{\text{N}}^{(j)}_{t}}$, which is the result of a regular share multiplication using $\idealMulNoTrunc$. After selecting the $\topki{l}$ rows of $\mat{M}$, we truncate the partially fixed-point multiplied $\secs{\tilde{\text{N}}^{(j)}_{t}}$ of those rows. This optimizations saves $(\abs{D_i} - l) \cdot \abs{\cT}$ invocations of \idealTrunc.

\end{myhideenv}

\begin{myhideenv}

\section{Analysis of Unlearning}\label{appx:bad-unlearn}
In this section, we examine unlearning-based traceback from a theoretical perspective to understand in what scenarios it can be expected to detect poisoning.

In \Cref{appx:bad-approx-unlearn}, We examine the basic assumptions made by approximate unlearning, and discuss scenarios where they fail. We pair these findings with experimental results demonstrating the failure modes with concrete attacks. These results demonstrate that there exist attacks, some natural and some adaptive, which evade detection by approximate unlearning, motivating the need for better responsibility scores.

In \Cref{appx:bad-true-unlearn}, we state and prove a theorem characterizing a failure mode of \textit{true} unlearning for user-level detection; that is, in the case when the approximate unlearning routine $\unl(\theta; D'; \Dtr)$ is replaced with retraining from-scratch on the dataset $D' \setminus \Dtr$. This result implies that the vulnerabilities of approximate unlearning cannot be resolved simply by resorting to a more accurate unlearning procedure. In fact, approximate unlearning provides \textit{resilience} to multiple poisoned owners.

\subsection{Approximate Unlearning Failure Modes}\label{appx:bad-approx-unlearn}
Approximate unlearning is a heuristic technique for unlearning user datasets, and lacks formal guarantees about the resulting classifier. In this section, we analyze the properties of the approximate unlearning objective more carefully and do not assume the unlearned classifier necessarily behaves like a leave-one-out model. We show that by understanding what assumptions approximate unlearning makes, we can discover and explore scenarios where it fails.

\myparagraph{Distributional Analysis}
To make analysis tractable, we remove the learning aspect and consider a distributional setting. Let $\mu_c$ and $\mu_p$ be distributions over $\cZ = \cX \times \cY$, representing the clean and poisoned data distributions, respectively. Let $\alpha \in (0, 1)$ be the mixture weights between $\mu_c$ and $\mu_p$, so that the poisoned training distribution is the mixture $\alpha \mu_c + (1 - \alpha) \mu_p.$ Additionally, let $0 < \beta < \min(\alpha, 1 - \alpha)$ be the unlearning fraction representing the fraction of samples to be unlearned. For simplicity, we assume that this weight is equal for both $\mu_c$  and $\mu_p$, which is analogous to the assumption that poisoned and clean user datasets all have the same size. Specifically, we define the approximately unlearned distributions
\begin{align}
    \nu^{\mathrm{unl}}_{p} &= \alpha \mu_c + \beta \mu_p^{\mathrm{unl}} + (1 - \alpha - \beta) \mu_p \\
    \nu^{\mathrm{unl}}_{c} &= (\alpha - \beta) \mu_c + \beta \mu_c^{\mathrm{unl}} + (1 - \alpha) \mu_p,
\end{align}
where $\mu_p^{\mathrm{unl}}$ and $\mu_c^{\mathrm{unl}}$ represent the poisoned and clean distributions with labels replaced by a uniform posterior, respectively.

Our goal is to understand the behavior of $\nu^{\mathrm{unl}}_{p}(y^{\mathrm{atk}} | x^{\mathrm{atk}})$ and $\nu^{\mathrm{unl}}_{c}(y^{\mathrm{atk}} | x^{\mathrm{atk}})$ for some attack sample $(x^{\mathrm{atk}}, y^{\mathrm{atk}}) \in \cZ$. In particular, these quantities are the related to the loss-based influence score used by approximate unlearning. If these quantities are sufficiently close to each other, the unlearning-based score will be less reliable as an indicator of poisoning. We begin by stating a useful fact describing the input-conditioned label posterior for general mixtures:
\begin{lemma}\label{lem:mixture-posterior}
    Let $\mu_1, \ldots, \mu_n$ be distributions over $\cZ$, and let $0 < \alpha_1, \ldots, \alpha_n < 1$ be mixing weights (i.e., $\sum_{i} \alpha_i = 1$). Let $\nu = \sum_i \alpha_i \mu_i$. Then, for all $(x, y) \in \cZ$,
    \begin{equation}\label{eq:mixture-posterior}
        \nu(y | x) = \sum_{i=1}^n \lambda_i \mu_i(y | x)
    \end{equation}
    where
    \begin{equation}
        \lambda_i = \frac{\alpha_i \mu_i(x)}{\sum_{j=1}^n \alpha_j \mu_j(x)}.
    \end{equation}
\end{lemma}
\Cref{lem:mixture-posterior} is a straightforward consequence of Bayes' Theorem, and describes the posterior $\nu(y | x)$ as a convex combination of posteriors $\mu_1(y | x), \ldots, \mu_n(y | x)$. In particular, the posteriors under both $\nu^{\mathrm{unl}}_{p}$ and $\nu^{\mathrm{unl}}_{c}$ are given by
\begin{align}
    &\nu^{\mathrm{unl}}_{p}(y | x) \nonumber \\
    &= \frac{\alpha \mu_c(x) \mu_c(y | x) + \tfrac{\beta}{C} \mu_p(x) + (1 - \alpha - \beta) \mu_p(x) \mu_p(y | x)}{\alpha \mu_c(x) + (1 - \alpha) \mu_p(x)} \label{eq:mixture-posterior-psn} \\
    &\nu^{\mathrm{unl}}_{c}(y | x) \nonumber \\
    &= \frac{(\alpha - \beta) \mu_c(x) \mu_c(y | x) + \tfrac{\beta}{C} \mu_c(x) + (1 - \alpha) \mu_p(x) \mu_p(y | x)}{\alpha \mu_c(x) + (1 - \alpha) \mu_p(x)} \label{eq:mixture-posterior-clean}
\end{align}
where $C$ is the number of classes. We can inspect \Cref{eq:mixture-posterior-psn,eq:mixture-posterior-clean} for certain representative poisoning settings. For instance, for classification tasks with unambiguous labeling, it is likely that $\mu_c(y^{\mathrm{atk}} | x^{\mathrm{atk}}) = 0$, and so the posteriors simplify to
\begin{align*}
    \nu^{\mathrm{unl}}_{p}(y | x) &= \frac{\tfrac{\beta}{C} \mu_p(x) + (1 - \alpha - \beta) \mu_p(x) \mu_p(y | x)}{\alpha \mu_c(x) + (1 - \alpha) \mu_p(x)} \\
    \nu^{\mathrm{unl}}_{c}(y | x) &= \frac{\tfrac{\beta}{C} \mu_c(x) + (1 - \alpha) \mu_p(x) \mu_p(y | x)}{\alpha \mu_c(x) + (1 - \alpha) \mu_p(x)}.
\end{align*}
For backdoor attacks, it is additionally possible that $\mu_c(x) = 0$ if the trigger does not appear in benign samples. In this case, we obtain the more succinct characterization
\begin{align*}
    \nu^{\mathrm{unl}}_{p}(y | x) &= \mu_p(y | x) + \frac{\beta}{1 - \alpha} (1 / C - \mu_p(y | x)) \\
    \nu^{\mathrm{unl}}_{c}(y | x) &= \mu_p(y | x).
\end{align*}
The above relations suggest two conditions under which the loss-based influence score might fail to distinguish between benign and poisoned distributions: first, if the ratio $\beta / (1 - \alpha)$ is small (that is, if too few poisons are unlearned), and second, if the poisoned posterior $\mu_p(y | x)$ is sufficiently close to uniform. Note that in the first case, if $\mu_p(y | x)$ is large (near 1), then even a small decrease in the posterior from unlearning can correspond to a large loss difference. We believe this contributes to the success of approximate unlearning in detecting backdoor attacks, even in cases where poisons are distributed across multiple poisoned datasets. On the other hand, when $\mu_p(y | x)$ is nearly uniform, the differences in loss are less extreme.

\subsection{True Unlearning Failure Modes}\label{appx:bad-true-unlearn}
Given the deficiencies of approximate unlearning demonstrated in \Cref{appx:bad-approx-unlearn}, it may seem tempting to replace it with an unlearning mechanism which more accurately approximates leave-one-out retraining. However, we show that in fact true unlearning exacerbates the ability of the adversary to hide poisons among multiple poisoned parties.

\begin{figure}[htb]
    \centering
    \begin{subfigure}[b]{0.45\linewidth}
        \centering
        \includegraphics[width=\linewidth]{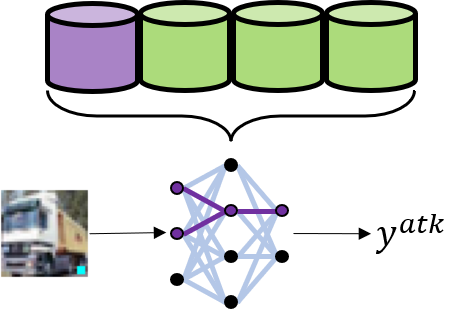}
        \caption{}
        \label{fig:attack_eg_1psn}
    \end{subfigure}
    \hfill
    \begin{subfigure}[b]{0.45\linewidth}
        \centering
        \includegraphics[width=\linewidth]{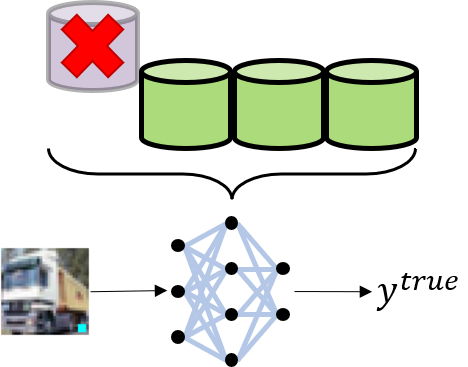}
        \caption{}
        \label{fig:attack_eg_1psn_unl}
    \end{subfigure}

    \vspace{0.5cm}
    
    \begin{subfigure}[b]{0.45\linewidth}
        \centering
        \includegraphics[width=\linewidth]{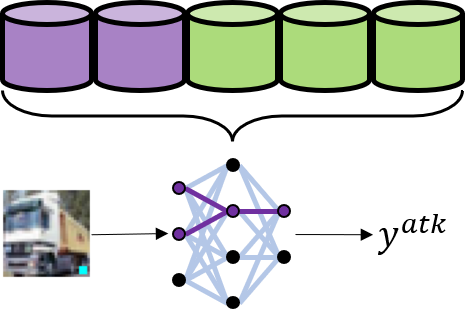}
        \caption{}
        \label{fig:attack_eg_2psn}
    \end{subfigure}
    \hfill
    \begin{subfigure}[b]{0.45\linewidth}
        \centering
        \includegraphics[width=\linewidth]{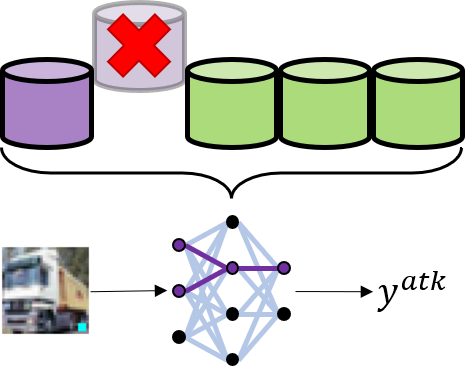}
        \caption{}
        \label{fig:attack_eg_2psn_unl}
    \end{subfigure}

    \caption{Different attack scenarios with concentrated versus distributed poisoning. (\subref{fig:attack_eg_1psn}) A single user dataset contains poisons, causing a malicious classification. (\subref{fig:attack_eg_1psn_unl}) Retraining the model after removing the malicious dataset removes the malicious behavior from the model. (\subref{fig:attack_eg_2psn}) Multiple user datasets contain poisons, causing a malicious misclassification. (\subref{fig:attack_eg_2psn_unl}). Removing only one of many poisoned datasets does not remove the malicious behavior from the retrained model.}
    \label{fig:attack_eg}
\end{figure}

In the case when the auditor produces the unlearned models $\theta_i'$ by retraining from scratch for each user $i$, each time leaving out the user dataset $D_i$, the effect of duplicated poisons across multiple user datasets is exacerbated. We consider a minimal theoretical example demonstrating this possibility. Consider two distributions $\mu_c$ and $\mu_p$ over $\cX \times \cY$, representing the clean and poisoned distributions, respectively. If the distributions $\mu_c$ and $\mu_p$ are compatible with each other, and if the model space $\Theta$ is sufficiently rich to simultaneously learn both distributions, then learning on any nontrivial mixture of those distributions yields the same set of loss-minimizing classifiers. Intuition for this idea is given \Cref{fig:attack_eg}. We formalize this idea below:

\begin{theorem}\label{thm:bad-true-unlearn}
    Let $\mu_c$ and $\mu_p$ be distributions over $\cZ = \cX \times \cY$, and let $\Theta \subseteq \bbR^d$ be a parameter space. For a nonnegative loss function $\cL(\theta; \mu) = \bbE_{(x, y) \sim \mu}[\ell(\theta; x, y)]$, denote
    \begin{align*}
        \Theta_c &:= \argmin_{\theta \in \Theta} \cL(\theta; \mu_c) \subseteq \Theta \\
        \Theta_p &:= \argmin_{\theta \in \Theta} \cL(\theta; \mu_p) \subseteq \Theta.
    \end{align*}
    If $\Theta_c \cap \Theta_p \ne \emptyset$, then for all $\alpha \in (0, 1)$, it holds that
    \begin{align*}
        \argmin_{\theta \in \Theta} \cL(\theta; \alpha \mu_c + (1 - \alpha) \mu_p) = \Theta_c \cap \Theta_p.
    \end{align*}
\end{theorem}
\begin{proof}
    Fix $\alpha \in (0, 1)$. For notational convenience, denote
    \begin{equation*}
        \nu := \alpha \mu_c + (1 - \alpha) \mu_p.
    \end{equation*}
    Additionally, denote by $L_c$ and $L_c$ the minimum loss achieved by those classifiers in $\Theta_c$ and $\Theta_p$, respectively.
    
    We show inclusions on both sides. For all $\theta' \in \Theta$, the mixture loss can be decomposed as
    \begin{equation}
        \cL(\theta'; \nu) = \alpha \cL(\theta'; \mu_c) + (1 - \alpha)\cL(\theta'; \mu_p) \label{eq:nu-loss}.
    \end{equation}
    Since $\cL(\theta; \mu_c) \ge L_c$ and $\cL(\theta; \mu_p) \ge L_p$ by definition, we also have for all $\theta' \in \Theta$ that
    \begin{equation}\label{eq:nu-loss-ineq}
        \cL(\theta'; \nu) \ge \alpha L_c + (1 - \alpha)L_p.
    \end{equation}
    Moreover, since $\Theta_c \cap \Theta_p \ne \emptyset$, this is in fact the minimum attainable loss on $\nu$, attained by at least those classifiers in $\Theta_c \cap \Theta_p$. That is, we have shown the inclusion
    \begin{equation*}
        \argmin_{\theta \in \Theta} \cL(\theta; \nu) \supseteq \Theta_c \cap \Theta_p.
    \end{equation*}
    
    Now, Suppose $\theta \in \argmin_{\theta \in \Theta} \cL(\theta; \nu)$. We must show that $\theta$ simultaneously minimizes loss on $\mu_c$ and $\mu_p$. Consider the case for $\mu_c$. We have shown that
    \begin{equation*}
        \alpha \cL(\theta; \mu_c) + (1 - \alpha)\cL(\theta; \mu_p) = \alpha L_c + (1 - \alpha) L_p.
    \end{equation*}
    Rearranging, we find that
    \begin{align*}
        \cL(\theta; \mu_c) &= \tfrac{1}{\alpha} \left[\alpha L_c + (1 - \alpha) (L_p - \cL(\theta; \mu_p)) \right] \\
        &\le L_c.
    \end{align*}
    Since $L_c$ minimizes loss on $\mu_c$, it follows that $\theta \in \Theta_c$. Symmetric argument shows that $\theta \in \Theta_p$, and thus
    \begin{equation*}
        \argmin_{\theta \in \Theta} \cL(\theta; \nu) \subseteq \Theta_c \cap \Theta_p. \tag*{\qedhere}
    \end{equation*}
\end{proof}

\myparagraph{Implications of \Cref{thm:bad-true-unlearn}}
The distributions $\mu_c$ and $\mu_p$ in \Cref{thm:bad-true-unlearn} can be replaced with the empirical distributions $D_i$ for each user. Interpreted this way, the theorem describes a simplified setting where each user dataset is one of two possible datasets, $D_c$ or $D_p$. If the necessary assumptions are satisfied, performing true unlearning on any individual dataset that is duplicated across multiple users corresponds to setting a different $\alpha \in (0, 1)$ (namely, $\alpha = \frac{a - 1}{a + b - 1}$, where $a$ and $b$ are the multiplicities of each dataset and we have removed one from those counted by $a$).

The condition that $\Theta_c \cap \Theta_p \ne 0$ requires that the learning tasks imposed by $\mu_c$ and $\mu_p$ are sufficiently compatible with each other, relative to the parameter space $\Theta$. This can be the case, for example, in backdoor attacks, where the supports of poisoned and clean distributions projected to the input space have trivial intersection (e.g., if no clean inputs contain the trigger). In this case, a sufficiently complex model can distinguish from which mixture an input originated and change its decision accordingly.

\end{myhideenv}

\begin{myhideenv}

\section{Cost-Effective Gradient Computation Details}
\label{appx:gradient}

\subsection{Cost Effective Gradient Computation}\label{sec:cost-effective-gradient}
Since training and traceback occur at different times, any inputs required for traceback must either be cached during training or re-supplied at traceback time. In particular, the gradients used in influence score computations can either be computed once during the training procedure and cached for future traceback executions, or recomputed during each traceback execution. Reusing cached gradients reduces runtime and removes the need for data owners to re-share private inputs during traceback. This is especially important in privacy-critical settings where data owners may be required to delete sensitive data when certain conditions are met.

Computing and storing sample-level gradients present two problems. First, computing full parameter gradients at each checkpoint can be expensive, since each checkpoint necessitates a forward and backward pass over the entire training set. Second, the long-term storage of sample-level gradients demands an untenable volume of storage space: the required storage grows linearly with the number of checkpoints $\abs{\cT}$, the training set size $\abs{\Dtr}$, and the dimension of the model $\dim(\Theta)$. To improve the traceback efficiency, we introduce two strategies focused on gradient computation and storage.

\myparagraph{Utilizing Final Layer Gradients}
Prior work has observed that under common initialization and layer normalization techniques, the final layer is responsible for most of the gradient norm variation \cite{katharopoulos_not_2019, hammoudeh_identifying_2022}. Moreover, when computing gradients with the backpropagation algorithm, a layer's gradient does not depend on earlier layers. By using gradients from only the model's later layers, the cost of gradient computation is reduced to essentially that of a single forward pass. We find that using gradients from the final one or two layers of the network generally results in a marginal reduction in traceback effectiveness while significantly improving runtime.

Writing the network's classification function as
\begin{equation*}
    f_\theta(x) = \mathrm{softmax}(\psi_{\theta_\psi}(\varphi_{\theta_\varphi}(x)))
\end{equation*}
for $\theta = [\theta_\varphi \enskip \theta_\psi]$ for parameterized feature extractor $\varphi$ and classification layers $\psi$, we replace all gradients $\nabla_\theta$ in a given gradient similarity score with $\nabla_{\theta_\psi}$. The resulting influence score can be interpreted as the mean gradient alignment at given classification parameters over a sequence of feature extractors.

\myparagraph{Gradient Projection and Caching}
A popular technique for data-oblivious dimensionality reduction is to use random projections to form low-dimensional \textit{sketches} \cite{johnson_extensions_1984}. Recently, Pruthi et al. showed how to use sketches to reduce the storage requirements for training data gradients in gradient-based influence methods \cite{pruthi_estimating_2020}. The idea is to apply a random projection into a low-dimensional vector space in a way which approximately preserves inner products. The projection dimension is controlled by the model trainer and can be chosen according to the desired trade-off between long-term storage constraints and future traceback performance.

We perform random gradient projection by choosing a random matrix $G \in \bbR^{r \times p}$, where $r$ is the desired projection dimension and $p$ is the native gradient dimension, where $\bbE[G^T G] = I_p$. The result induces an unbiased (over the randomness of $G$) estimator of the inner product between any two (fixed) vectors $u, v \in \bbR^p$:
\begin{equation*}
    \bbE[(Gu)^T(Gv)] = \bbE[u^T G^T G v] = u^T v.
\end{equation*}
Following Pruthi et al., we sample $G$ by choosing its entries i.i.d. from $\cN(0, 1/r)$. Additionally, we consider a variation of Pruthi et al.'s gradient projection strategy. Instead of re-using the same projection matrix across all checkpoints, we propose to sample a new projection matrix at each checkpoint. We do this to reduce the probability of obtaining a single low-quality projection for a given downstream misclassification event, at the cost of additional storage space. Despite requiring storage of multiple projection matrices, the overhead does not grow with the number of training samples. Hence, for large enough datasets, using many low-dimensional projections consumes less storage space than using a single high-dimensional projection applied at all checkpoints.

Low-dimensional gradient sketches can be directly \textit{reused} between traceback executions. Thus, we view the computation of projected gradients as a pre-processing step embedded inside the training procedure. At key checkpoints, the training algorithm generates projection matrices $G_t$ and projected gradients $g_t$ of the training loss. These projected gradients are then used for multiple downstream traceback computations.

\end{myhideenv}

\end{appendices}

\end{document}